\documentclass{article}
\usepackage[preprint]{neurips_2026}
\usepackage[utf8]{inputenc}
\usepackage[T1]{fontenc}
\usepackage{hyperref}
\usepackage{url}
\usepackage{booktabs}
\usepackage{amsfonts}
\usepackage{amssymb}
\usepackage{amsmath}
\usepackage{amsthm}
\usepackage{mathtools}
\usepackage{microtype}
\usepackage{xcolor}
\usepackage{enumitem}
\usepackage{multirow}
\usepackage{array}
\usepackage{graphicx}
\usepackage{subcaption}
\usepackage{bm}
\usepackage{algorithm}
\usepackage{algpseudocode}

\title{Why Do Aligned LLMs Remain Jailbreakable:\\Refusal-Escape Directions, Operator-Level Sources, and Safety--Utility Trade-off}

\author{%
Yu Chen\textsuperscript{1,2} \quad
Yuanhao Liu\textsuperscript{1,2} \quad
Qi Cao\textsuperscript{1}\thanks{Corresponding author.}
\\
\textsuperscript{1}Institute of Computing Technology, Chinese Academy of Sciences\\
\textsuperscript{2}University of Chinese Academy of Sciences, Beijing, China\\
\texttt{\{chenyu24s,caoqi\}@ict.ac.cn}
}

\newtheorem{theorem}{Theorem}
\newtheorem{proposition}{Proposition}

\newtheorem{definition}{Definition}

\newtheorem{assumption}{Assumption}

\newcommand{\R}{\mathbb{R}}
\newcommand{\inner}[2]{\left\langle #1, #2 \right\rangle_F}
\newcommand{\innersmall}[2]{\langle #1, #2 \rangle_F}

\newcommand{\normsmall}[1]{\lVert #1 \rVert_F}

\setlist[enumerate]{leftmargin=1.5em,labelsep=0.4em,itemsep=0.1em,topsep=0.15em}
\setlist[itemize]{leftmargin=1.25em,labelsep=0.4em,itemsep=0.08em,topsep=0.12em}

\begin{document}
\maketitle

\begin{abstract}
Aligned large language models (LLMs) remain vulnerable to jailbreak attacks. Recent mechanistic studies have identified latent features and representation shifts associated with jailbreak success, but they leave a more fundamental question open: \emph{why do aligned LLMs remain jailbreakable, and what structural vulnerabilities in the model make this possible?} We study this question through a continuous input-transformation view. Our theoretical finding is that aligned models can still exhibit \emph{Refusal-Escape Directions} (RED): local perturbation directions around a harmful input that shift the model's behavior from refusal to answering while preserving the model's harmful-semantics interpretation. From this perspective, a jailbreak is not only a successful discrete prompt construction, but can also be understood as a refusal-to-answer behavior transition induced by continuously perturbing a harmful input along RED. We then prove that RED can be exactly decomposed into contributions from \emph{operator-level sources} across the model's operator structure, and identify normalization, residual-wiring, and terminal sources as analytically constrained operator-level sources. To eliminate RED, the shared expressive modules---self-attention and MLP---must eliminate the contributions from these analytically constrained sources while preserving the mechanisms that support benign responses. These competing requirements give rise to a conditional \emph{safety--utility trade-off}. Experiments across multiple models and attack methods empirically analyze RED from two complementary perspectives and show that added token dimensions can expose RED, while successful jailbreaks exhibit refusal-to-answer shifts largely aligned with terminal-source contributions.
\end{abstract}

\section{Introduction}
\label{sec:intro}

Safety alignment methods such as reinforcement learning from human feedback are designed to make large language models helpful while refusing harmful requests \citep{christiano2017deep,ouyang2022training}. Yet aligned LLMs remain vulnerable to jailbreak attacks, including adversarial suffixes, prompt rewriting, and automated search over jailbreak templates \citep{zou2023universal,ding2024renellm,yu2023gptfuzzer,liu2023autodan}. These attacks can bypass alignment safeguards and induce models to produce harmful answers that they should otherwise refuse, creating serious safety risks. Understanding why this remains possible is therefore a central problem in AI safety.

Recent work has provided initial analyses of refusal behavior and jailbreak success in aligned LLMs. Refusal behavior has been associated with low-dimensional directions or sparse latent features \citep{arditi2024refusal,yeo2025sae}, and safety alignment can associate harmfulness recognition with refusal behavior \citep{zhou2024alignment}. Empirical analyses of jailbreaks suggest that successful attacks can suppress refusal-related signals and shift harmful prompts toward seemingly harmless regions in representation space \citep{ball2024latent,lin2024representation}. These studies provide valuable insight into how jailbreak attacks work, with additional related work discussed in Appendix~\ref{app:related_work}. However, they remain largely phenomenological and leave a more fundamental theoretical question unanswered: \emph{Why do aligned LLMs remain jailbreakable, and what structural vulnerabilities in the model make this possible?} This question motivates our study.

To address this question, we first seek to formalize how jailbreaks work. More precisely, we study how input modifications that turn a harmful prompt into its jailbreak counterpart propagate through an LLM and induce a harmful answer. We formalize these modifications as a continuous input transformation. This perspective motivates our central theoretical object---\emph{Refusal-Escape Directions} (RED): local perturbation directions around a harmful input that shift the model from refusal to answering, while preserving the model's harmful-semantics interpretation. By continuously perturbing a harmful input along RED, an attacker can accumulate these local behavior shifts into a full refusal-to-answer transition, eventually inducing the model to answer in accordance with the harmful semantics of the input.

To understand the structural origin of RED, we decompose RED into contributions from \emph{operator-level sources} across the model's operator structure. Among these sources, normalization, residual-wiring, and terminal sources play a special role as analytically constrained operator-level sources: eliminating these sources requires analytic vanishing conditions; otherwise, the source remains nonzero on a dense open subset. Thus, to eliminate RED, the shared expressive modules---self-attention and MLP---must eliminate contributions from these analytically constrained sources while preserving the mechanisms that support benign responses. These competing requirements give rise to a conditional \emph{safety--utility trade-off}. Our experiments across multiple models and attack methods empirically analyze RED from two perspectives and show that added token dimensions can expose RED, while successful jailbreaks exhibit refusal-to-answer shifts largely aligned with terminal-source contributions.

Our contributions are fourfold:
\begin{itemize}
\item We introduce \emph{Refusal-Escape Directions} (RED) as a theoretical object for studying jailbreakability under a continuous input-transformation view.

\item We prove that RED can be exactly decomposed into contributions from \emph{operator-level sources} according to the model's operator structure.

\item We identify normalization, residual-wiring, and terminal sources as analytically constrained operator-level sources, showing that eliminating RED gives rise to a conditional \emph{safety--utility trade-off}.

\item We empirically analyze RED from two complementary perspectives, showing that added token dimensions can expose RED, while successful jailbreaks exhibit refusal-to-answer shifts largely aligned with terminal-source contributions.
\end{itemize}

\section{A continuous input-transformation view of jailbreaks}
\label{sec:jailbreak}

We begin with a basic question: when a harmful prompt is modified into a jailbreak prompt, why does the model's behavior shift from refusal to answer? More precisely, how do input modifications propagate through the model and induce a harmful answer? To study jailbreaks from this perspective, we adopt a continuous view of the input modifications. First, we represent the harmful prompt and its jailbreak counterpart in a common input space of matched dimension, and formalize the modifications as a continuous transformation in this space, represented by an absolutely continuous curve
\[
X(\eta):[\eta_0,\eta_\star]\to \R^{n_0\times d_0},
\]
where \(X(\eta_0)\) and \(X(\eta_\star)\) are the input embedding matrices of the harmful prompt and its jailbreak counterpart, respectively; \(n_0\) is the matched input length, and \(d_0\) is the embedding dimension.

To induce a harmful answer, the transformation must preserve the underlying harmful semantics of the input. Let \(\mathcal P(X(\eta))\subseteq \R^{n_0\times d_0}\) denote the first-order perturbation directions at \(X(\eta)\) that preserve the model's harmful-semantics interpretation, which we formalize in the next section. We therefore restrict attention to transformations whose tangent direction satisfies
\[
X'(\eta)\in \mathcal P(X(\eta))
\qquad
\text{for almost every } \eta\in[\eta_0,\eta_\star].
\]

To study changes in model behavior along the input transformation, we adopt a nonzero target-behavior subspace \(\mathcal Y\subseteq\R^{n_T\times d_T}\) in the model's hidden-state space. This subspace is a first-order proxy for directions relevant to the model's answer-versus-refusal behavior, such as logit-difference directions, linear-probe directions, or latent features \citep{arditi2024refusal, he2024jailbreaklens}.

To analyze how the input transformation affects model behavior, we model the forward computation of an LLM as a differentiable chain
\begin{equation}
H^{(0)}(X)=X, 
\qquad
H^{(t+1)}(X)=F_t\!\left(H^{(t)}(X)\right),
\qquad
t=0,\dots,T-1,
\label{eq:state_chain}
\end{equation}
where \(F_t\) denotes the \(t\)-th model block, \(H^{(t)}(X)\) is the hidden-state matrix before block \(t\), and \(H^{(T)}(X)\in\R^{n_T\times d_T}\) lies in the hidden-state space where the target-behavior subspace \(\mathcal Y\) is defined. Let \(k=\dim(\mathcal Y)\ge 1\), and let \(\{Y_i\}_{i=1}^{k}\)  be any Frobenius-orthonormal basis of \(\mathcal Y\). We define the target-behavior signal by
\begin{equation}
\Phi_{\mathcal Y}(X)
:=
P_{\mathcal Y}H^{(T)}(X)
=
\sum_{i=1}^k
\inner{H^{(T)}(X)}{Y_i}Y_i ,
\label{eq:target_signal_def}
\end{equation}
where \(P_{\mathcal Y}\) is the Frobenius-orthogonal projection onto \(\mathcal Y\).

Along the transformation, \(\Phi_{\mathcal Y}(X(\eta))\) is an absolutely continuous curve in \(\mathcal Y\), with derivative
\[
\frac{d}{d\eta}\Phi_{\mathcal Y}(X(\eta))
=
P_{\mathcal Y}J_{0:T}(X(\eta))X'(\eta),
\]
where \(J_{0:T}(X):=D H^{(T)}(X)\) is the Jacobian from \(X\) to \(H^{(T)}(X)\). Accordingly, the cumulative target-subspace displacement along the transformation is
\begin{equation}
\Delta_{\mathcal Y}
:=
\int_{\eta_0}^{\eta_\star}
P_{\mathcal Y}J_{0:T}(X(\eta))X'(\eta)\,d\eta
=
P_{\mathcal Y}\!\left(
H^{(T)}(X(\eta_\star))-H^{(T)}(X(\eta_0))
\right).
\label{eq:cumulative_target_subspace_displacement}
\end{equation}

By continuously perturbing a harmful input along directions that preserve the model's harmful-semantics interpretation, a jailbreak succeeds if the cumulative displacement is sufficient to move the model's answer-versus-refusal target-behavior signal from a refusal-associated region toward an answer-associated region. The key objects to be characterized are therefore the perturbation directions \(X'(\eta)\) that lie in \(\mathcal P(X(\eta))\) and can change the target-behavior signal \(\Phi_{\mathcal Y}(X(\eta))\).

\section{Refusal-escape directions and operator-level decomposition}
\label{sec:refusal_escape_directions}

\subsection{Harmful-semantics-sensitive subspace field and refusal-escape directions}

To formalize the harmful-semantics-preserving set \(\mathcal P(X)\), we adopt a local subspace-field assumption. The use of subspaces is motivated by mechanistic-interpretability work that studies behaviorally relevant concepts through latent features, directions, or low-dimensional subspaces in the model's representation space \citep{bereska2024mechreview,somvanshi2025bridging}. Since the model computation is highly nonlinear, the input-side directions that affect the model's harmful-semantics interpretation may vary with the input. We therefore use a subspace field rather than a single global subspace.

\begin{assumption}[Harmful-semantics-sensitive subspace field]
\label{ass:u_star}
There exist an open region \(\mathcal X\subseteq\R^{n_0\times d_0}\), with the model's harmful-semantics interpretations on \(\mathcal X\) consistent with the attacker's intended harmful objective, and a nonzero subspace field \(X\mapsto\mathcal U(X)\), with \(\mathcal U(X)\subseteq\R^{n_0\times d_0}\), such that, for every \(X\in\mathcal X\), the model's harmful-semantics interpretation is locally first-order sensitive only to perturbation directions in \(\mathcal U(X)\).  Equivalently, the first-order harmful-semantics-preserving perturbation set is
\begin{equation}
\mathcal P(X):=\mathcal U(X)^\perp
=
\{\Delta X:\inner{U}{\Delta X}=0\text{ for every }U\in\mathcal U(X)\}.
\label{eq:PX_def}
\end{equation}
\end{assumption}

Under Assumption~\ref{ass:u_star}, let \(r_X:=\dim(\mathcal U(X))\ge 1\), and let \(\{U_i(X)\}_{i=1}^{r_X}\) be any Frobenius-orthonormal basis of \(\mathcal U(X)\). We define the projections onto \(\mathcal U(X)\) and \(\mathcal U(X)^\perp\) by
\begin{equation}
P_{\mathcal U(X)}(V)
:=
\sum_{i=1}^{r_X}\inner{V}{U_i(X)}U_i(X),
\qquad
P_{\mathcal U(X)^\perp}(V)
:=
V-\sum_{i=1}^{r_X}\inner{V}{U_i(X)}U_i(X).
\label{eq:subspace_projection}
\end{equation}
For a subspace \(\mathcal S\), write \(P_{\mathcal U(X)}\mathcal S:=\{P_{\mathcal U(X)}W:W\in\mathcal S\}\) and \(P_{\mathcal U(X)^\perp}\mathcal S:=\{P_{\mathcal U(X)^\perp}W:W\in\mathcal S\}\).

Within the harmful-semantics-preserving set \(\mathcal P(X)\), perturbation directions relevant to jailbreak are those that can change the answer-versus-refusal target-behavior signal \(\Phi_{\mathcal Y}(X)\). We single out these directions and call them the Refusal-Escape Directions.

\begin{definition}[Refusal-Escape Directions (RED)]
For a target-behavior subspace \(\mathcal Y\) and \(X\in\mathcal X\), define the Refusal-Escape Directions by
\begin{equation}
\mathcal R_{\mathcal Y}(X)
:=
P_{\mathcal U(X)^\perp}
\bigl(J_{0:T}(X)^*\mathcal Y\bigr)
=
\left\{
P_{\mathcal U(X)^\perp}
\bigl(J_{0:T}(X)^*Y\bigr)
:
Y\in\mathcal Y
\right\},
\label{eq:red_def}
\end{equation}
where \(J_{0:T}(X)^*\) denotes the Frobenius adjoint of \(J_{0:T}(X)\). For any \(Y\in\mathcal Y\), we denote the refusal-escape direction associated with \(Y\) by
\begin{equation}
R_Y(X)
:=
P_{\mathcal U(X)^\perp}
\bigl(J_{0:T}(X)^*Y\bigr)
\in \mathcal R_{\mathcal Y}(X).
\label{eq:red_direction_def}
\end{equation}
\end{definition}

\begin{proposition}[Basic properties of RED]
\label{prop:r0_basic}
For any \(X\in\mathcal X\), \(\mathcal R_{\mathcal Y}(X)\) satisfies:
\begin{enumerate}
    \item \textbf{Harmful-semantics preservation:}
    \(\mathcal R_{\mathcal Y}(X)\subseteq\mathcal P(X)\).

    \item \textbf{Directions orthogonal to \(\mathcal R_{\mathcal Y}(X)\) within \(\mathcal P(X)\) have no first-order effect on the target-behavior signal:}
    for every \(\Delta X\in\mathcal P(X)\), the first-order variation of the target-behavior signal satisfies
    \begin{equation}
    D\Phi_{\mathcal Y}(X)[\Delta X]
    =
    P_{\mathcal Y}J_{0:T}(X)\Delta X
    =
    P_{\mathcal Y}J_{0:T}(X)
    \bigl[
    P_{\mathcal R_{\mathcal Y}(X)}\Delta X
    \bigr].
    \label{eq:red_first_order_effect}
    \end{equation}
    In particular, if \(\Delta X\perp \mathcal R_{\mathcal Y}(X)\), then \(D\Phi_{\mathcal Y}(X)[\Delta X]=0\).
\end{enumerate}
A detailed proof is given in Appendix~\ref{app:proof_r0_basic}.
\end{proposition}

Proposition~\ref{prop:r0_basic} gives the operational meaning of RED. Within the harmful-semantics-preserving set \(\mathcal P(X)\), a perturbation \(\Delta X\) can change the answer-versus-refusal target-behavior signal \(\Phi_{\mathcal Y}(X)\) at first order only through its projection onto \(\mathcal R_{\mathcal Y}(X)\). Ideally, within the harmful region \(\mathcal X\), the safety alignment mechanism would make answer-versus-refusal behavior depend only on the harmful semantics itself: once the harmful semantics are preserved, the model's refusal behavior should not be locally steerable toward answering. A nontrivial RED breaks this ideal relation. It provides semantics-preserving local degrees of freedom that can still modulate the model's answer-versus-refusal behavior, and is therefore the structural vulnerability that we use to study jailbreakability.

\subsection{Operator-level decomposition of refusal-escape directions}
\label{sec:decomposition}

To understand how RED arises, we trace it through the model's blocks. Fix \(X\in\mathcal X\), and let \(J_t(X)=D F_t(H^{(t)}(X))\) denote the Jacobian of block \(t\) at \(H^{(t)}(X)\). We first push forward the harmful-semantics-sensitive subspace through the model:
\begin{equation}
\mathcal U_0(X):=\mathcal U(X),
\qquad
\mathcal U_{t+1}(X):=J_t(X)\mathcal U_t(X),
\qquad
t=0,\ldots,T-1 .
\label{eq:forward_chain}
\end{equation}
Here \(\mathcal U_t(X)\) is the harmful-semantics-sensitive hidden-state subspace at step \(t\).

Let \(J_t(X)^*\) denote the Frobenius adjoint of \(J_t(X)\). For any \(Y\in\mathcal Y\), we pull back the target-behavior direction \(Y\) through the model by
\begin{equation}
G_T(Y;X):=Y,
\qquad
G_t(Y;X):=J_t(X)^*G_{t+1}(Y;X),
\qquad
t=T-1,\ldots,0 .
\label{eq:backward_direction_chain}
\end{equation}
Thus \(G_t(Y;X)\) is the target-behavior sensitivity direction at step \(t\).

For brevity, we write \(\mathcal U_t:=\mathcal U_t(X)\) and \(G_t:=G_t(Y;X)\). We then decompose \(G_t\) at each step into its harmful-semantics-aligned component and its remaining component:
\begin{equation}
A_t:=P_{\mathcal U_t}G_t,
\qquad
R_t:=P_{\mathcal U_t^\perp}G_t,
\qquad
G_t=A_t+R_t .
\label{eq:aligned_orthogonal_direction}
\end{equation}
The direction \(A_t\) is the component of the target-behavior sensitivity direction that lies in the harmful-semantics-sensitive hidden-state subspace \(\mathcal U_t\). It represents how the model's harmful-semantics interpretation influences the answer-versus-refusal behavior at step \(t\). In contrast, \(R_t\) is the remaining component after removing the \(\mathcal U_t\)-aligned part. By the chain rule, \(G_0=J_{0:T}(X)^*Y\). Since \(\mathcal U_0=\mathcal U(X)\), we have \(R_0=P_{\mathcal U(X)^\perp}(J_{0:T}(X)^*Y)=R_Y(X)\), so \(R_0\) is exactly the refusal-escape direction associated with \(Y\).

Two source types are important for understanding how \(R_0\) arises. The first is the \emph{leakage source}:
\begin{equation}
B_t
:=
P_{\mathcal U_t^\perp}
\bigl(J_t(X)^*A_{t+1}\bigr),
\qquad
t=0,\ldots,T-1.
\label{eq:Bt}
\end{equation}
It measures how the target-behavior sensitivity carried by the harmful-semantics-aligned component \(A_{t+1}\) leaks into the remaining component at step \(t\). The second source is the \emph{terminal source}:
\begin{equation}
R_T
=
P_{\mathcal U_T^\perp}Y .
\label{eq:RT}
\end{equation}
It measures the terminal mismatch between the target-behavior direction \(Y\) and the propagated harmful-semantics-sensitive subspace \(\mathcal U_T\), capturing the part of the model's answer-versus-refusal behavior that cannot be explained by its harmful-semantics interpretation.

We can now decompose RED into the contributions induced by these sources.

\begin{proposition}[Exact decomposition of RED]
\label{prop:exact_decomposition}
Fix \(X\in\mathcal X\). For every \(Y\in\mathcal Y\) and every \(t=0,\ldots,T-1\),
\begin{equation}
R_t
=
B_t
+
J_t(X)^*R_{t+1}.
\label{eq:recurrence}
\end{equation}
Consequently, writing
\begin{equation}
S_t(Y;X)
:=
\bigl(J_0(X)^*\circ\cdots\circ J_{t-1}(X)^*\bigr)B_t,
\qquad
S_T(Y;X)
:=
\bigl(J_0(X)^*\circ\cdots\circ J_{T-1}(X)^*\bigr)R_T,
\label{eq:generic_contributions}
\end{equation}
where the empty composition is understood as the identity, we have
\begin{equation}
R_Y(X)
=
R_0
=
\sum_{t=0}^{T-1}S_t(Y;X)+S_T(Y;X).
\label{eq:exact_sum_compact}
\end{equation}
For the whole target-behavior subspace \(\mathcal Y\), RED is obtained by collecting these decompositions over all \(Y\in\mathcal Y\):
\begin{equation}
\mathcal R_{\mathcal Y}(X)
=
\left\{
\sum_{t=0}^{T-1}S_t(Y;X)+S_T(Y;X)
:
Y\in\mathcal Y
\right\}.
\label{eq:red_set_of_direction_decomp}
\end{equation}
A detailed proof is given in Appendix~\ref{app:proof_decomp}.
\end{proposition}

Proposition~\ref{prop:exact_decomposition} gives an exact decomposition of each refusal-escape direction into input-side contributions obtained by transporting leakage sources and the terminal source back through adjoint-Jacobian transport channels. \(S_t(Y;X)\) is the contribution induced by the \(t\)-th leakage source. It captures a channel, mediated by the harmful-semantics-aligned component \(A_{t+1}\), that can modulate how the model's harmful-semantics interpretation influences the answer-versus-refusal behavior. In contrast, \(S_T(Y;X)\) is the contribution induced by the terminal source. It captures a side-channel outside the model's harmful-semantics interpretation that can directly affect the answer-versus-refusal behavior.

We next specialize the block-level decomposition to the operator structure of modern pre-norm residual LLMs, such as LLaMA-family architectures \citep{grattafiori2024llama3}. This yields an exact decomposition of RED into contributions from \emph{operator-level sources}:
\begin{equation}
\mathcal R_{\mathcal Y}(X)
=
\left\{
S_{\mathrm{norm}}^\Sigma(Y;X)
+
S_{\mathrm{attn}}^\Sigma(Y;X)
+
S_{\mathrm{mlp}}^\Sigma(Y;X)
+
S_{\mathrm{res}}^\Sigma(Y;X)
+
S_T(Y;X)
:
Y\in\mathcal Y
\right\}.
\label{eq:operator_level_red_set}
\end{equation}
Here, for each target-behavior direction \(Y\in\mathcal Y\), \(S_{\mathrm{norm}}^\Sigma(Y;X)\), \(S_{\mathrm{attn}}^\Sigma(Y;X)\), \(S_{\mathrm{mlp}}^\Sigma(Y;X)\), and \(S_{\mathrm{res}}^\Sigma(Y;X)\) respectively denote the sums, across all blocks, of transported input-side contributions induced by normalization, self-attention, MLP, and residual-wiring leakage sources; \(S_T(Y;X)\) is the transported input-side contribution induced by the terminal source. The explicit derivation is given in Appendix~\ref{app:operators}.

\section{Safety--utility trade-off for eliminating the refusal-escape directions}
\label{sec:tradeoff}

We now analyze under what conditions RED can be eliminated in an LLM, meaning \(\mathcal R_{\mathcal Y}(X)=\{0\}\) for all \(X\in\mathcal X\), and how this elimination condition gives rise to a conditional safety--utility trade-off. We begin by formalizing the real-analytic regularity assumptions used in this section.

\begin{assumption}[Real-analytic regularity on a harmful region]
\label{ass:analytic}
There exists a connected open region \(\Omega\supseteq\mathcal X\) such that the Jacobians \(J_t\) and their adjoints, the projection fields \(P_{\mathcal U_t}\) and \(P_{\mathcal U_t^\perp}\) associated with the propagated harmful-semantics-sensitive subspace fields \(X\mapsto\mathcal U_t(X)\), and the action fields \(\mathcal E_\Theta^\Sigma\), \(\mathcal E_h^\Sigma\), and \(\mathcal E_b^\Sigma\) defined in this section are well-defined and real-analytic on \(\Omega\). It follows that the operator-level sources and transported input-side contributions computed from these Jacobians, adjoints, and projection fields are also real-analytic on \(\Omega\).
\end{assumption}

This assumption is natural on local regions where the model's discrete structure is fixed. For architectures with piecewise-smooth activations or routing-based modules, this means choosing \(\Omega\) inside a region with a fixed ReLU sign pattern, MoE route, or attention mask. Within such a region, the relevant LLM computations are compositions of real-analytic operators, so the quantities in Assumption~\ref{ass:analytic} are real-analytic whenever they are well-defined.

\subsection{Analytically constrained operator-level sources of refusal-escape directions}

\begin{definition}[Analytically constrained operator-level source]
\label{def:rigid_source}
Under Assumption~\ref{ass:analytic}, an operator-level source \(q\) is \emph{analytically constrained} if eliminating it on \(\mathcal X\) is equivalent to an analytic vanishing condition \(\mathfrak C_q\) holding on \(\Omega\). For any leakage source \(B\), elimination means \(B(X)=0\) on \(\mathcal X\); for the terminal source, it means \(R_T(X)=0\) on \(\mathcal X\). If \(\mathfrak C_q\) fails, the operator-level source is nonzero on a dense open subset of \(\Omega\).
\end{definition}

For pre-norm residual LLMs, the exact decomposition of RED contains three classes of analytically constrained operator-level sources: normalization sources, residual-wiring sources, and the terminal source. Eliminating these sources requires normalization-consistency, residual-branch-consistency, or target-subspace-consistency conditions determined by the surrounding input--output geometry, rather than direct adjustment of the sources themselves. Detailed proofs are given in Appendix~\ref{app:rigid_sources}.

\subsection{Exact elimination fields and the safety--utility trade-off}

By Definition~\ref{def:rigid_source}, analytically constrained sources are generically nonzero on a dense open subset of \(\Omega\) unless their analytic vanishing conditions are satisfied. Consequently, eliminating RED on \(\mathcal X\) requires eliminating contributions from these analytically constrained sources. In a pre-norm residual LLM, we model this elimination burden as falling on the shared expressive modules \(\mathcal M_c\), namely the self-attention and MLP modules, which provide high expressive capacity through attention mixing and nonlinear feed-forward transformations \citep{kajitsuka2024selfattention,hornik1989multilayer}.

Under the exact decomposition in Proposition~\ref{prop:exact_decomposition}, a module in \(\mathcal M_c\) can affect RED in three ways: first, by altering the surrounding input--output geometry of an analytically constrained source \(q\) so that the analytic vanishing condition \(\mathfrak C_q\) becomes satisfied; second, by generating its own leakage sources to offset the transported contributions from analytically constrained sources; and third, by changing its adjoint-Jacobian map, thereby attenuating or blocking the adjoint-Jacobian transport channels that carry downstream sources back to the input side. We unify these mechanisms under the notion of the module's \emph{RED-elimination action}. Accordingly, for each \(g\in\mathcal M_c\), we define a RED-elimination action field \(X\mapsto\mathcal E_\Theta^{g}(X)\) on \(\Omega\), induced by module \(g\), including geometry adjustment, leakage-source offset, and transport-channel modification. Then
\begin{equation}
\mathcal E_\Theta^\Sigma(X):=\sum_{g\in\mathcal M_c}\mathcal E_\Theta^{g}(X)
\label{eq:red_elimination_action_field}
\end{equation}
is the total RED-elimination action field of the shared expressive modules.

Once the analytically constrained sources and the remaining part of the network are fixed, exact elimination of RED on \(\mathcal X\) requires the shared expressive modules to realize a specific analytic field on \(\mathcal X\), denoted by \(\mathcal E_h^\Sigma(X)\). This is the \emph{harmful-region RED-elimination field}: the total action required to eliminate RED contributions from analytically constrained sources on \(\mathcal X\). Thus, exact elimination requires
\begin{equation}
\mathcal E_\Theta^\Sigma(X)=\mathcal E_h^\Sigma(X),
\qquad \forall X\in\mathcal X.
\label{eq:harmful_requirement}
\end{equation}

Now let \(\mathcal X_b\subseteq\Omega\) be a nonempty open benign region, and let \(\mathcal E_b^\Sigma(X)\) denote the \emph{benign-region behavioral field} that the same shared expressive modules must realize there in order to preserve the intended benign-response behavior. We next show that exact harmful-region RED elimination can be incompatible with this benign-region requirement.

\begin{theorem}[Conditional incompatibility of exact RED elimination]
\label{thm:tradeoff}
Under Assumption~\ref{ass:analytic}, suppose that the shared expressive modules \(\mathcal M_c\) are required to satisfy the harmful-region RED-elimination constraint
\[
\mathcal E_\Theta^\Sigma(X)=\mathcal E_h^\Sigma(X),
\qquad \forall X\in\mathcal X,
\]
and simultaneously the benign-region behavioral constraint
\[
\mathcal E_\Theta^\Sigma(X)=\mathcal E_b^\Sigma(X),
\qquad \forall X\in\mathcal X_b.
\]
If
\[
\mathcal E_h^\Sigma\not\equiv \mathcal E_b^\Sigma
\qquad \text{on }\Omega,
\]
then no single parameter setting \(\Theta\) can satisfy both exact requirements. A detailed proof is given in Appendix~\ref{app:proof_tradeoff}.
\end{theorem}

Theorem~\ref{thm:tradeoff} formalizes the conditional safety--utility trade-off. Exact elimination of RED on \(\mathcal X\) pins the total RED-elimination action field of the shared expressive modules \(\mathcal E_\Theta^\Sigma\) to the harmful-region RED-elimination field \(\mathcal E_h^\Sigma\). Since these fields are real-analytic on the connected region \(\Omega\), this equality extends from \(\mathcal X\) to all of \(\Omega\). Therefore, the same modules cannot simultaneously realize a different benign-region behavioral field \(\mathcal E_b^\Sigma\), unless the harmful and benign fields are analytically identical.

\section{Experiments}
\label{sec:experiments}

To make the theory operational, we instantiate reference choices for the input transformation \(X(\eta)\), the target-behavior subspace \(\mathcal Y\), and the harmful-semantics-sensitive subspace \(\mathcal U(X)\). These choices provide one concrete reference setting for empirically analyzing RED and jailbreak behavior under the continuous input-transformation view. Under this reference setting, we empirically analyze RED from two complementary perspectives. In Section~\ref{model_analysis}, we study harmful inputs directly and examine how added token dimensions affect RED. In Section~\ref{jailbreak_analysis}, we analyze successful jailbreak cases across different attack methods and show that their refusal-to-answer shifts are often largely aligned with terminal-source contributions.

\subsection{Experimental settings}

\paragraph{Models and jailbreak samples.}
We study five aligned pre-norm residual LLMs: Qwen3-4B/14B \citep{yang2025qwen3}, Llama-3.1-8B-Instruct \citep{grattafiori2024llama3}, and Gemma-3-4B/12B-IT \citep{gemmateam2025gemma3}. We analyze five jailbreak attacks: GCG \citep{zou2023universal}, AutoDAN \citep{liu2023autodan}, GPTFuzzer \citep{yu2023gptfuzzer}, TAP \citep{mehrotra2024tree}, and ReNeLLM \citep{ding2024renellm}. For each model--attack pair, we collect 100 harmful--jailbreak prompt pairs where the harmful prompt is refused and the jailbreak prompt is answered. We use Qwen3Guard-Gen-8B to evaluate refusal and jailbreak success \citep{zhao2025qwen3guard}.

We analyze each harmful--jailbreak pair separately and use pair-specific reference choices for the input transformation, target-behavior subspace, and harmful-semantics-sensitive subspace.

\paragraph{Reference input transformation.}
To construct a continuous input transformation from a harmful prompt to its jailbreak counterpart, we align the token-embedding matrices of the two prompts into a common input space of matched dimension using a recursive alignment algorithm. The algorithm first matches the longest common token blocks and, when no shared block remains, aligns the most similar token pair in embedding space. Unmatched spans are padded with all-zero token embeddings, which serve as placeholders without explicit lexical meaning. This yields the boundary points \(X(\eta_0)\) and \(X(\eta_\star)\), while preserving shared harmful semantic fragments as much as possible. We then take the linear interpolation \(X(\eta)=X(\eta_0)+\frac{\eta-\eta_0}{\eta_\star-\eta_0}\bigl(X(\eta_\star)-X(\eta_0)\bigr)\) as the reference input transformation.

\paragraph{Reference target-behavior subspace.}
We use the final-token logit-difference direction between the answer-inducing jailbreak prompt and the refusal-inducing harmful prompt as a reference proxy for the model's answer-versus-refusal behavior. We pull this logit-difference direction back to the final hidden-state space through the output head, and use the one-dimensional subspace spanned by the resulting hidden-state direction as the reference target-behavior subspace \(\mathcal Y_{\mathrm{ref}}\).

\paragraph{Reference harmful-semantics-sensitive subspace.}
We use the original harmful input before placeholder padding as the reference point. At this clean, placeholder-free harmful input, the aligned model is expected to base its answer-versus-refusal behavior primarily on its harmful-semantics interpretation. We therefore treat the input-side target-behavior sensitivity subspace \(J_{0:T}(\widehat X)^*\mathcal Y_{\mathrm{ref}}\) at this clean input as a reference proxy for the harmful-semantics-sensitive variation, and denote it by \(\widehat{\mathcal U}\). We then map \(\widehat{\mathcal U}\) into the aligned input space using the same padding pattern as the harmful input; the resulting one-dimensional subspace is used as the fixed reference harmful-semantics-sensitive subspace \(\mathcal U_{\mathrm{ref}}\) along the reference transformation. This fixed subspace is a reference instantiation of the general subspace field in Assumption~\ref{ass:u_star}, used only for the pair-specific empirical analysis.

\paragraph{Reference refusal-escape direction.}
Since the reference target-behavior subspace \(\mathcal Y_{\mathrm{ref}}\) is one-dimensional, let \(Y_{\mathrm{ref}}\) be its Frobenius-orthonormal basis vector whose positive direction points away from refusal and toward answering. We compute \(R_{\mathrm{ref}}(X):=P_{\mathcal U_{\mathrm{ref}}^\perp}(J_{0:T}(X)^*Y_{\mathrm{ref}})\) as the reference refusal-escape direction. Its Frobenius norm \(\normsmall{R_{\mathrm{ref}}(X)}\) measures the first-order sensitivity of this direction for changing the reference target-behavior signal \(\Phi_{\mathcal Y_{\mathrm{ref}}}(X)=\innersmall{H^{(T)}(X)}{Y_{\mathrm{ref}}}Y_{\mathrm{ref}}\).

Additional details and justifications for the experimental settings are provided in Appendix~\ref{app:settings}.

\subsection{Observation 1: added token dimensions expose RED at harmful inputs}
\label{model_analysis}

\begin{figure}[t]
    \centering
    \begin{subfigure}[t]{0.49\linewidth}
        \centering
        \includegraphics[width=\linewidth]{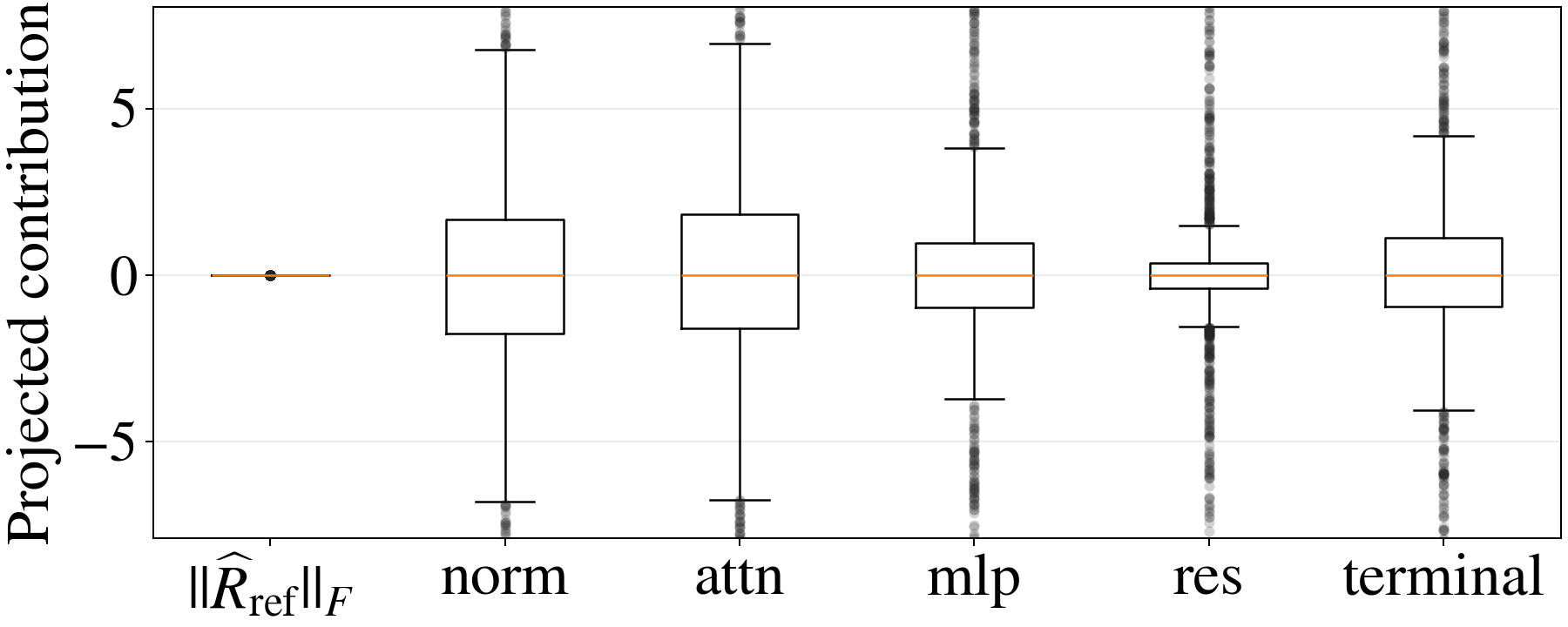}
        \caption{Before padding (projected onto \(e_1\)).}
        \label{fig:operator_level_leakage_before}
    \end{subfigure}
    \hfill
    \begin{subfigure}[t]{0.49\linewidth}
        \centering
        \includegraphics[width=\linewidth]{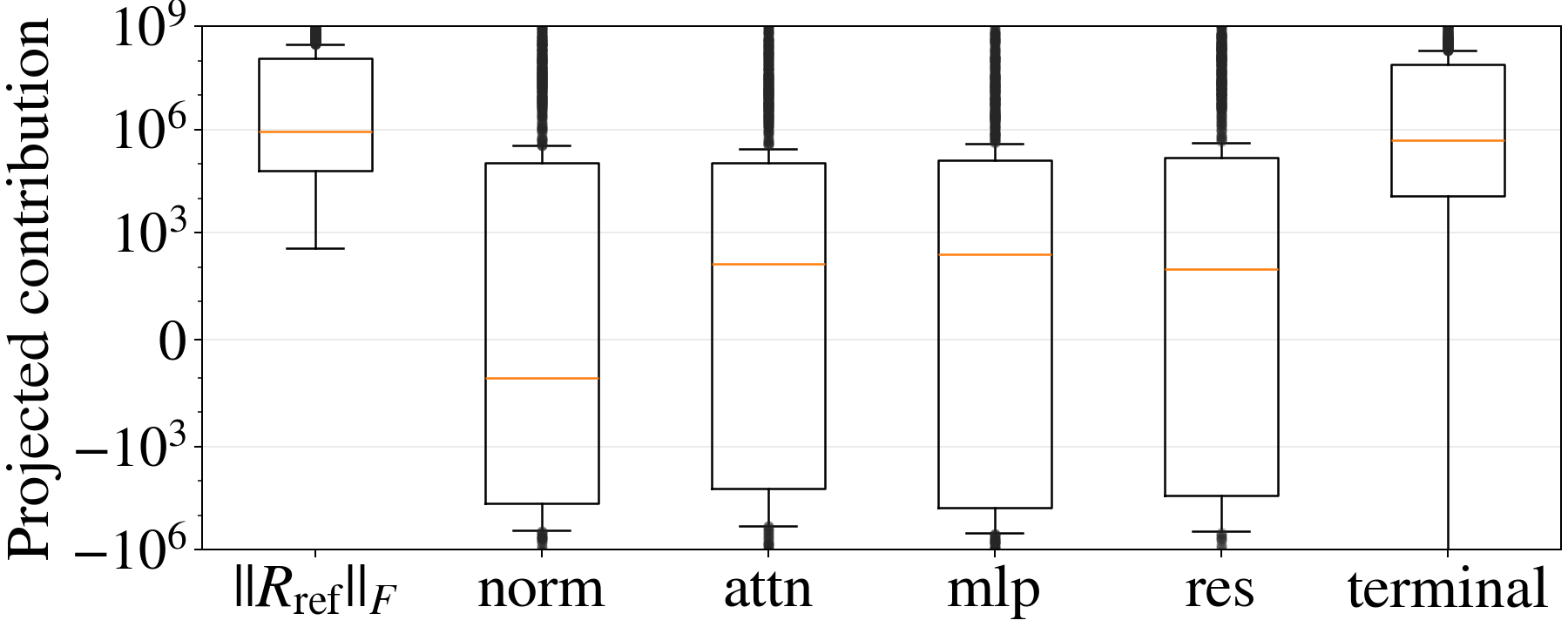}
        \caption{After padding (projected onto \(R_{\mathrm{ref}}\)).}
        \label{fig:operator_level_leakage_after}
    \end{subfigure}
    \caption{Reference refusal-escape direction and its operator-level contributions at harmful inputs, aggregated over all samples across models.
    (a) shows the unpadded baseline, where the corresponding refusal-escape direction is zero by construction; operator-level contributions are visualized by signed projection onto \(e_1\).
    (b) shows the padded harmful input, where added token dimensions expose a nonzero \(R_{\mathrm{ref}}(X)\); operator-level contributions are visualized by signed projection onto \(R_{\mathrm{ref}}(X)/\normsmall{R_{\mathrm{ref}}(X)}\).}
    \label{fig:operator_level_leakage}
\end{figure}

We first ask how added token dimensions, a common feature of many jailbreak attacks that introduce additional tokens into the original harmful prompt, affect RED at a harmful input.

Before placeholder padding, \(\widehat{\mathcal U}=J_{0:T}(\widehat X)^*\mathcal Y_{\mathrm{ref}}\) is the input-side target-behavior sensitivity subspace at the placeholder-free harmful input \(\widehat X\) under our reference setting, so the corresponding refusal-escape direction satisfies \(\widehat{R}_{\mathrm{ref}}(\widehat X)=P_{\widehat{\mathcal U}^\perp}(J_{0:T}(\widehat X)^*Y_{\mathrm{ref}})=0\) by construction, which serves as a controlled baseline. After placeholder padding, we compute \(R_{\mathrm{ref}}(X)\) and its Frobenius norm \(\normsmall{R_{\mathrm{ref}}(X)}\), and then decompose both cases into the operator-level contributions, following \eqref{eq:operator_level_red_set}. In the padded case, we measure each source's signed contribution by projecting it onto \(R_{\mathrm{ref}}(X)/\normsmall{R_{\mathrm{ref}}(X)}\). For the unpadded baseline, where \(\widehat{R}_{\mathrm{ref}}(\widehat X)=0\), we instead project onto the one-hot coordinate basis vector \(e_1\).

Figure~\ref{fig:operator_level_leakage}a shows the unpadded baseline: although nonzero contributions from different operator-level sources are already present, what keeps \(\widehat{R}_{\mathrm{ref}}(\widehat{X})\) at zero is their mutual cancellation at the input side. After placeholder padding, Figure~\ref{fig:operator_level_leakage}b shows that this cancellation no longer holds and a much larger nonzero \(R_{\mathrm{ref}}(X)\) appears. This occurs because the newly introduced token dimensions create new transport channels through which the leakage and terminal sources can be transported to the input side, so cancellation identities that previously held at the input side no longer hold automatically. As a result, RED becomes explicit at the harmful input.

Figure~\ref{fig:operator_level_leakage}b further shows that all operator-level sources leave observable contributions after padding, with the terminal source contributing most stably and prominently. As discussed in Section~\ref{sec:decomposition}, leakage sources affect the answer-versus-refusal behavior through channels mediated by harmful-semantics-aligned components, whereas the terminal source captures a side-channel outside the model's harmful-semantics interpretation that can directly affect the answer-versus-refusal behavior. This suggests that the influence from harmful-semantics interpretation to answer-versus-refusal behavior is often relatively robust, while the model still retains substantial side-channels outside the harmful-semantics interpretation. This also helps explain the Competing Objectives attack in \citet{wei2023jailbroken}, where competing objectives can steer behavior through factors outside the harmful request itself. Model-specific results are presented and discussed in Appendix~\ref{app:model_results}.

\subsection{Observation 2: refusal-to-answer shifts largely align with terminal-source contributions}

\label{jailbreak_analysis}

\begin{figure}[t]
    \centering
    \begin{subfigure}[t]{0.32\linewidth}
        \centering
        \includegraphics[width=\linewidth]{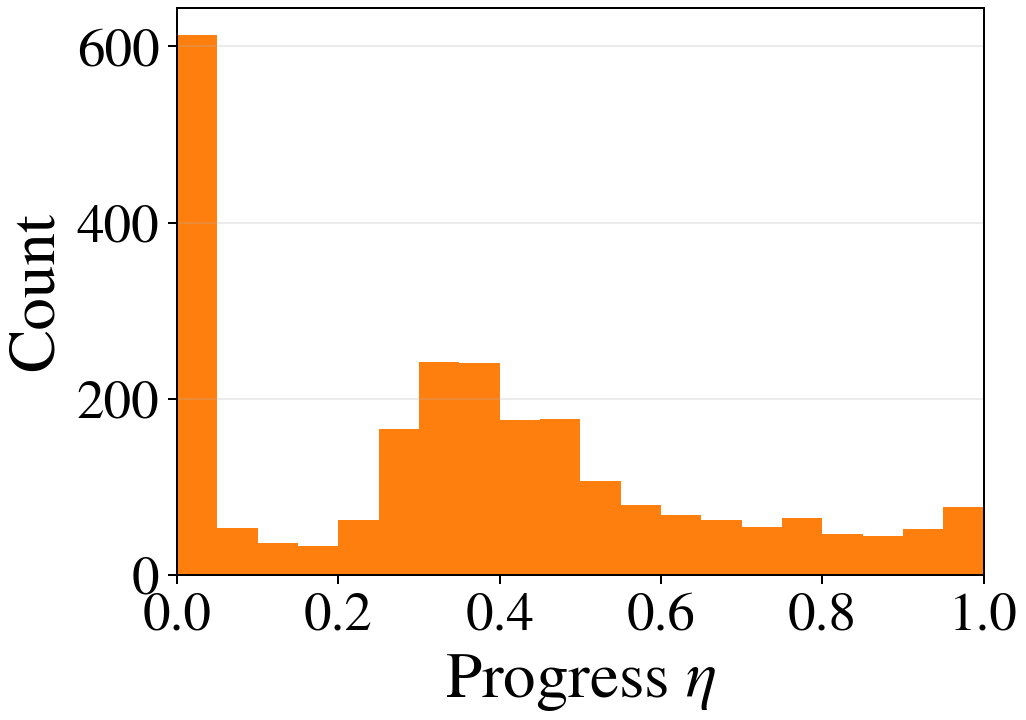}
        \caption{Earliest jailbreak progress.}
        \label{fig:jailbreak_progress}
    \end{subfigure}
    \hfill
    \begin{subfigure}[t]{0.32\linewidth}
        \centering
        \includegraphics[width=\linewidth]{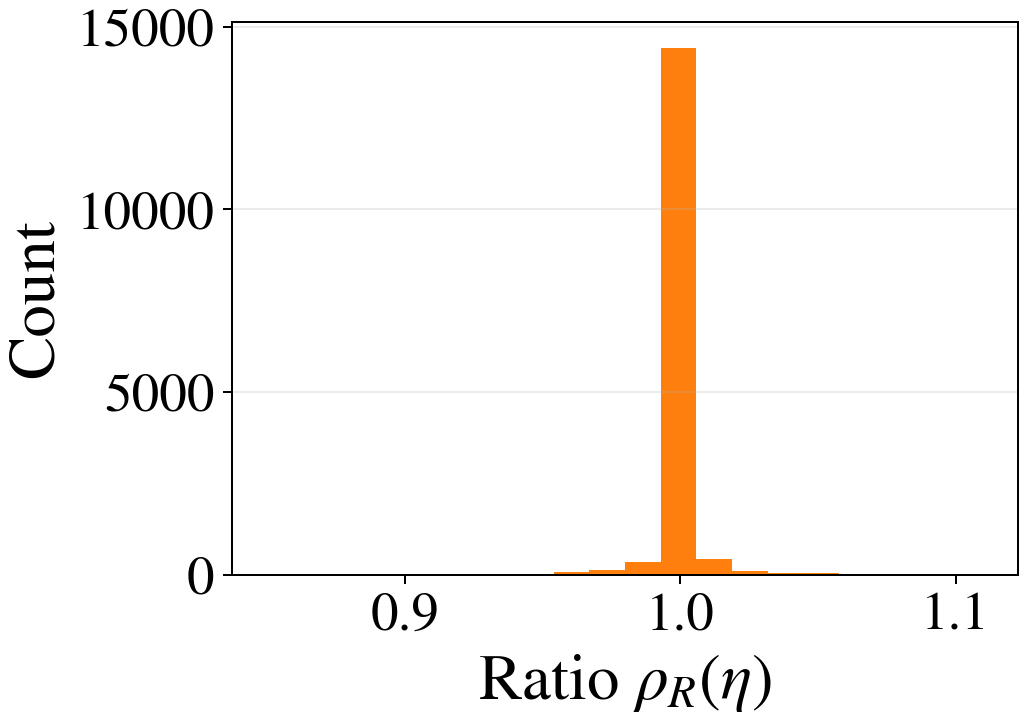}
        \caption{Signed \(R_{\mathrm{ref}}\) ratio.}
        \label{fig:jailbreak_ratio}
    \end{subfigure}
    \hfill
    \begin{subfigure}[t]{0.32\linewidth}
        \centering
        \includegraphics[width=\linewidth]{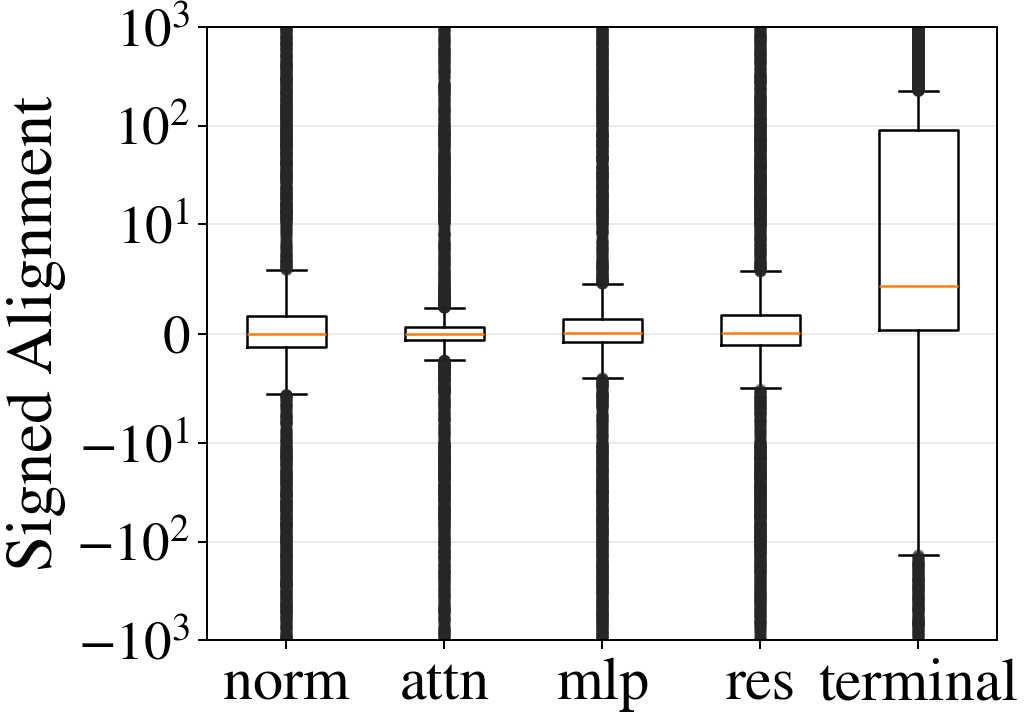}
        \caption{Signed operator-level alignment.}
        \label{fig:jailbreak_family}
    \end{subfigure}
    \caption{Jailbreak analysis under the continuous input-transformation view, aggregated across models and attack methods. 
    (a) shows the earliest progress along the input transformation at which Qwen3Guard-Gen-8B first judges the output induced by an intermediate input as a successful jailbreak. 
    (b) shows the signed ratio between the alignment of the transformation tangent with \(R_{\mathrm{ref}}\) and its alignment with the input-side target-behavior sensitivity direction \(J_{0:T}(X)^*Y_{\mathrm{ref}}\).
    (c) shows the signed alignment between the transformation tangent and each operator-level contribution \(S_f^\Sigma\).}
    \label{fig:jailbreak_main}
\end{figure}

We next analyze successful jailbreak cases across different attack methods under our reference setting. For each harmful--jailbreak pair, we track the earliest progress along the reference input transformation at which Qwen3Guard-Gen-8B first judges the output induced by an intermediate input as a successful jailbreak. Figure~\ref{fig:jailbreak_main}a shows that a substantial fraction of cases are judged successful in the first half of the transformation, and many within the first \(5\%\). This suggests that the final discrete jailbreak prompt is not the only informative object to analyze; the intermediate process under the continuous input-transformation view is also worth studying.

Along the reference input transformation, we extract four intermediate points at different stages of the refusal-to-answer transition for each sample; details are given in Appendix~\ref{app:attack_sample}. At each point, we compute the signed ratio \(\rho_R(\eta)=\innersmall{X'(\eta)}{R_{\mathrm{ref}}(X(\eta))}/\innersmall{X'(\eta)}{J_{0:T}(X(\eta))^*Y_{\mathrm{ref}}}\), which measures how much of the local reference target-behavior change is accounted for by the reference refusal-escape direction. We also compute the signed alignment \(\innersmall{X'(\eta)}{S_f^\Sigma(X(\eta))}\), where \(S_f^\Sigma\) ranges over the normalization, self-attention, MLP, residual-wiring, and terminal-source contributions defined in \eqref{eq:operator_level_red_set}. This quantity indicates how each operator-level contribution accounts for the local reference target-behavior change along the reference input transformation. Figure~\ref{fig:jailbreak_main}b shows that \(\rho_R(\eta)\) is consistently close to \(1\) across models and attacks, indicating that most of the local refusal-to-answer shift along the transformation is accounted for by \(R_{\mathrm{ref}}\). Figure~\ref{fig:jailbreak_main}c shows that terminal-source contributions are most consistently aligned in the positive direction, whereas other operator-level contributions are more dispersed. Together with Figure~\ref{fig:operator_level_leakage}b, this suggests that successful jailbreaks can be understood as exploiting refusal-escape directions already visible at harmful inputs, especially the terminal-source side-channel, rather than as creating a separate mechanism along the transformation. These results provide an implication for defense: alignment should not only enumerate completed jailbreak-prompt constructions, but should also aim to suppress local RED around harmful inputs, with particular attention to terminal-source contributions. Attack-specific and model-specific results are presented and discussed in Appendix~\ref{app:attack_results}.

\section{Conclusion}

Under the continuous input-transformation view, we introduced Refusal-Escape Directions (RED) as structural vulnerabilities for studying jailbreakability. We showed how RED arises through the model's operator structure, decomposes into operator-level source contributions, and why its elimination gives rise to a conditional safety--utility trade-off. Our experiments across multiple models and attack methods empirically analyzed RED from two complementary perspectives: added token dimensions can expose RED at harmful inputs, and successful jailbreaks exhibit refusal-to-answer shifts largely aligned with terminal-source contributions. Together, these results suggest that alignment should not focus only on enumerating completed jailbreak-prompt constructions, but should also aim to suppress local RED around harmful inputs, with particular attention to terminal-source contributions. Limitations and broader impacts are discussed in Appendices~\ref{app:limitations} and~\ref{app:broader_impacts}, respectively.


\bibliographystyle{plainnat}
\bibliography{references}

\clearpage
\appendix

\section{Related work}
\label{app:related_work}

\paragraph{Safety alignment.}
Safety alignment aims to make LLMs helpful while reducing harmful behavior. RLHF is a standard post-training approach for instruction-following assistants: it learns a reward model from human preference data and optimizes the model against this reward signal, thereby encouraging helpful responses while discouraging harmful ones \citep{christiano2017deep,stiennon2020learning,ouyang2022training}. DPO simplifies this preference-optimization pipeline by directly optimizing a preference loss, and is likewise used to steer models toward preferred and safer behaviors without an explicit reinforcement-learning stage \citep{rafailov2023direct}. Other alignment methods further refine safety objectives, separate helpfulness and harmlessness preferences, use constitutional feedback, or remove undesirable knowledge and behaviors while preserving model capabilities \citep{bai2022constitutional,dai2024saferlhf,ji2024pku,yao2024unlearning,zhang2024safeunlearning}. These methods provide important tools for reducing harmful behavior in LLMs. Our work studies a complementary question: after such alignment procedures, why can models still remain vulnerable to jailbreak attacks? We analyze this question through a RED-centered theoretical framework, which characterizes local structural vulnerabilities around harmful inputs and provides implications for future safety-alignment strategies.

\paragraph{Jailbreak attacks.}
A large body of work shows that aligned models remain vulnerable to jailbreak attacks based on adversarial suffix optimization, fuzzing, rewriting, tree search, long-context steering, and multi-turn interaction \citep{wei2023jailbroken,zou2023universal,yu2023gptfuzzer,liu2023autodan,mehrotra2024tree,ding2024renellm,anil2024manyshot,liu2024autodanturbo,liu2024dra,russinovich2024crescendo,jiang2024artprompt,zhu2024advprefix}. Among them, GCG optimizes adversarial suffixes through gradient-guided token search, showing that even nonsensical suffixes can bypass refusal mechanisms \citep{zou2023universal}.  AutoDAN uses a hierarchical genetic algorithm to search for more natural and stealthy jailbreak prompts, improving semantic meaningfulness and transferability compared with purely token-level suffix optimization \citep{liu2023autodan}.  GPTFuzzer frames jailbreak discovery as fuzzing: it mutates seed jailbreak templates and keeps effective variants using feedback from a judge model \citep{yu2023gptfuzzer}. TAP uses an attacker LLM to search over a tree of candidate attacks, pruning unlikely branches to improve query efficiency \citep{mehrotra2024tree}. ReNeLLM rewrites harmful requests and embeds them in plausible scenarios, modifying the original prompt while preserving the harmful semantics \citep{ding2024renellm}. We view the success of these jailbreak methods through the continuous input-transformation perspective, which motivates our RED-centered theoretical framework for analyzing the structural vulnerabilities in aligned models that such attacks may exploit.

\paragraph{Mechanistic studies of refusal and jailbreak.}
Recent work has provided initial analyses of refusal behavior and jailbreak success in aligned LLMs. \citet{arditi2024refusal} show that refusal behavior can be associated with a low-dimensional direction, while \citet{yeo2025sae} use sparse autoencoders to identify latent refusal features associated with refusal behavior. \citet{zhou2024alignment} further show that pretrained models can already recognize harmfulness, and that safety alignment associates this harmfulness recognition with refusal behavior. For jailbreak behavior, \citet{ball2024latent} show that successful jailbreaks can suppress refusal-related signals, while \citet{lin2024representation} show that jailbreaks can shift harmful prompts toward seemingly harmless regions in representation space. \citet{wei2023jailbroken} study safety-training failures through competing objectives and mismatched generalization; our RED-centered analysis helps explain how Competing Objectives attacks can steer behavior through side-channels outside the model's harmful-semantics interpretation. Other mechanistic and representation-level studies further analyze safety fine-tuning, refusal-related representations, harmfulness perception, and activation-level shifts under jailbreak attacks \citep{jain2024safety,he2024jailbreaklens,prakash2025beyond,gao2024shaping}. These studies provide valuable insight into how jailbreak attacks work, but remain largely phenomenological. We instead take a theoretical perspective: by modeling how input modifications from a harmful prompt to its jailbreak counterpart propagate through an LLM and induce a harmful answer, we develop a RED-centered framework for analyzing the structural vulnerabilities exploited by jailbreak attacks.

\section{Limitations and future work}
\label{app:limitations}

\paragraph{Local first-order formalization of harmful semantics.}
Our formalization of the model's harmful-semantics interpretation is motivated by mechanistic-interpretability work that studies behaviorally relevant concepts through latent features, directions, or low-dimensional subspaces in representation space \citep{bereska2024mechreview,somvanshi2025bridging}. Since LLM computations are highly nonlinear, we do not impose a single global harmful-semantics-sensitive subspace. Instead, we use a subspace field \(X\mapsto\mathcal U(X)\) to accommodate the case where the input-side directions affecting the model's harmful-semantics interpretation vary with the input. This local first-order formalization is sufficient for our goal of characterizing semantics-preserving directions that can still affect answer-versus-refusal behavior. In real models, however, harmful-semantics interpretation may also depend on higher-order or nonlocal structure in the input space. Extending the RED-centered framework beyond the local first-order setting, and studying how richer semantic structure interacts with RED, is an important direction for future work.

\paragraph{Conditional safety--utility trade-off.}
Our analysis shows that exact RED elimination gives rise to a conditional safety--utility trade-off. The key condition is whether the harmful-region RED-elimination field and the benign-region behavioral field are analytically identical on the surrounding open region. This identical-field requirement is highly restrictive: eliminating RED near harmful inputs and preserving benign-response behavior generally impose different functional requirements on the same shared expressive modules. This provides a mechanism by which safety and utility requirements can come into conflict. Nevertheless, because real models are complex and benign behavior is difficult to specify exactly, we do not prove that every concrete model must satisfy the non-identical-field condition. Future work should study this condition more directly, both theoretically and empirically. Moreover, practical defenses may not require exact RED elimination; suppressing RED below a useful threshold may be sufficient. Quantifying the approximate safety--utility trade-offs arising from partial RED suppression is therefore an important next step.

\paragraph{Reference choices for empirical analysis.}
To make the RED-centered framework operational for empirical analysis, we instantiate reference choices for the input transformation \(X(\eta)\), the target-behavior subspace \(\mathcal Y\), and the harmful-semantics-sensitive subspace \(\mathcal U(X)\) for each harmful--jailbreak pair. To avoid making the analysis depend on additional probes, sparse-feature estimators, or other auxiliary mechanisms whose validity, stability, or causal interpretation may introduce separate uncertainties, we use pair-specific reference choices. Specifically, the reference target-behavior subspace is defined by the harmful--jailbreak logit-difference direction, and the reference harmful-semantics-sensitive subspace is defined by the input-side target-behavior sensitivity at the clean, placeholder-free harmful input. These choices yield stable and directly comparable empirical objects for analyzing RED at harmful inputs and successful jailbreak transitions. They should therefore be understood as a controlled reference instantiation of the theory, not as the only possible empirical realization of RED. Future work should explore richer and more independent constructions of the reference target-behavior and harmful-semantics-sensitive subspaces, including pair-independent probes, harmful--benign contrast directions, and sparse-feature analyses, while first establishing their stability, effectiveness, and causal interpretability. It should also strengthen causal validation by comparing successful jailbreaks with failed attempts and benign prompt transformations, and by testing whether suppressing RED or its operator-level contributions reduces jailbreak success.

\section{Broader impacts}
\label{app:broader_impacts}

This work aims to improve the robustness of aligned LLMs by identifying structural vulnerabilities that enable jailbreaks and clarifying why eliminating them can induce a conditional safety--utility trade-off. The analysis may support better diagnosis and defense of aligned models. At the same time, mechanistic insights about jailbreakability could potentially be misused to better understand how to bypass refusal mechanisms. To mitigate this risk, we do not propose a new jailbreak attack or provide turnkey offensive procedures; existing jailbreak methods are used only for evaluation, and our analysis is framed around diagnosis and defense.

\section{Proofs}
\label{app:proofs}

\subsection{Proof of Proposition~\ref{prop:r0_basic}}
\label{app:proof_r0_basic}

\begin{proof}
Fix \(X\in\mathcal X\). By the definition of RED,
\[
\mathcal R_{\mathcal Y}(X)
=
P_{\mathcal U(X)^\perp}
\bigl(J_{0:T}(X)^*\mathcal Y\bigr)
=
\left\{
P_{\mathcal U(X)^\perp}
\bigl(J_{0:T}(X)^*Y\bigr)
:
Y\in\mathcal Y
\right\}.
\]
Every element of \(\mathcal R_{\mathcal Y}(X)\) lies in \(\mathcal U(X)^\perp\). Under Assumption~\ref{ass:u_star}, \(\mathcal P(X)=\mathcal U(X)^\perp\). Therefore,
\[
\mathcal R_{\mathcal Y}(X)\subseteq\mathcal P(X),
\]
which proves harmful-semantics preservation.

We next prove that directions orthogonal to \(\mathcal R_{\mathcal Y}(X)\) within \(\mathcal P(X)\) have no first-order effect on the target-behavior signal. Let \(\Delta X\in\mathcal P(X)=\mathcal U(X)^\perp\). For any \(Y\in\mathcal Y\), the Frobenius adjoint relation gives
\[
\inner{J_{0:T}(X)\Delta X}{Y}
=
\inner{\Delta X}{J_{0:T}(X)^*Y}.
\]
Decompose \(J_{0:T}(X)^*Y\) into its \(\mathcal U(X)\)-aligned and \(\mathcal U(X)^\perp\)-components:
\[
J_{0:T}(X)^*Y
=
P_{\mathcal U(X)}
\bigl(J_{0:T}(X)^*Y\bigr)
+
P_{\mathcal U(X)^\perp}
\bigl(J_{0:T}(X)^*Y\bigr).
\]
Since \(\Delta X\in\mathcal U(X)^\perp\), we have
\[
\inner{\Delta X}{
P_{\mathcal U(X)}
\bigl(J_{0:T}(X)^*Y\bigr)
}
=0.
\]
Hence
\[
\begin{aligned}
\inner{\Delta X}{J_{0:T}(X)^*Y}
&=
\inner{\Delta X}{
P_{\mathcal U(X)}
\bigl(J_{0:T}(X)^*Y\bigr)
}
+
\inner{\Delta X}{
P_{\mathcal U(X)^\perp}
\bigl(J_{0:T}(X)^*Y\bigr)
} \\
&=
\inner{\Delta X}{
P_{\mathcal U(X)^\perp}
\bigl(J_{0:T}(X)^*Y\bigr)
}.
\end{aligned}
\]

By the definition of \(\mathcal R_{\mathcal Y}(X)\),
\[
P_{\mathcal U(X)^\perp}
\bigl(J_{0:T}(X)^*Y\bigr)
\in
\mathcal R_{\mathcal Y}(X).
\]
Now decompose \(\Delta X\) into its \(\mathcal R_{\mathcal Y}(X)\)-component and its orthogonal complement:
\[
\Delta X
=
P_{\mathcal R_{\mathcal Y}(X)}\Delta X
+
P_{\mathcal R_{\mathcal Y}(X)^\perp}\Delta X.
\]
Since
\(P_{\mathcal R_{\mathcal Y}(X)^\perp}\Delta X\)
is orthogonal to every element of \(\mathcal R_{\mathcal Y}(X)\), and
\(P_{\mathcal U(X)^\perp}(J_{0:T}(X)^*Y)\in\mathcal R_{\mathcal Y}(X)\), we have
\[
\inner{
P_{\mathcal R_{\mathcal Y}(X)^\perp}\Delta X
}{
P_{\mathcal U(X)^\perp}
\bigl(J_{0:T}(X)^*Y\bigr)
}
=0.
\]
Therefore,
\[
\begin{aligned}
\inner{\Delta X}{
P_{\mathcal U(X)^\perp}
\bigl(J_{0:T}(X)^*Y\bigr)
}
&=
\inner{
P_{\mathcal R_{\mathcal Y}(X)}\Delta X
+
P_{\mathcal R_{\mathcal Y}(X)^\perp}\Delta X
}{
P_{\mathcal U(X)^\perp}
\bigl(J_{0:T}(X)^*Y\bigr)
} \\
&=
\inner{
P_{\mathcal R_{\mathcal Y}(X)}\Delta X
}{
P_{\mathcal U(X)^\perp}
\bigl(J_{0:T}(X)^*Y\bigr)
}
\\
& \quad +
\inner{
P_{\mathcal R_{\mathcal Y}(X)^\perp}\Delta X
}{
P_{\mathcal U(X)^\perp}
\bigl(J_{0:T}(X)^*Y\bigr)
} \\
&=
\inner{
P_{\mathcal R_{\mathcal Y}(X)}\Delta X
}{
P_{\mathcal U(X)^\perp}
\bigl(J_{0:T}(X)^*Y\bigr)
}.
\end{aligned}
\]

Moreover,
\(P_{\mathcal R_{\mathcal Y}(X)}\Delta X\in\mathcal R_{\mathcal Y}(X)\), and the first part of the proof gives
\(\mathcal R_{\mathcal Y}(X)\subseteq\mathcal U(X)^\perp\). Hence
\[
P_{\mathcal R_{\mathcal Y}(X)}\Delta X
\in
\mathcal U(X)^\perp,
\]
so
\[
\inner{
P_{\mathcal R_{\mathcal Y}(X)}\Delta X
}{
P_{\mathcal U(X)}
\bigl(J_{0:T}(X)^*Y\bigr)
}
=0.
\]
Using again the decomposition of \(J_{0:T}(X)^*Y\), we get
\[
\begin{aligned}
\inner{
P_{\mathcal R_{\mathcal Y}(X)}\Delta X
}{
J_{0:T}(X)^*Y
}
&=
\inner{
P_{\mathcal R_{\mathcal Y}(X)}\Delta X
}{
P_{\mathcal U(X)}
\bigl(J_{0:T}(X)^*Y\bigr)
+
P_{\mathcal U(X)^\perp}
\bigl(J_{0:T}(X)^*Y\bigr)
} \\
&=
\inner{
P_{\mathcal R_{\mathcal Y}(X)}\Delta X
}{
P_{\mathcal U(X)}
\bigl(J_{0:T}(X)^*Y\bigr)
}
\\
& \quad +
\inner{
P_{\mathcal R_{\mathcal Y}(X)}\Delta X
}{
P_{\mathcal U(X)^\perp}
\bigl(J_{0:T}(X)^*Y\bigr)
} \\
&=
\inner{
P_{\mathcal R_{\mathcal Y}(X)}\Delta X
}{
P_{\mathcal U(X)^\perp}
\bigl(J_{0:T}(X)^*Y\bigr)
}.
\end{aligned}
\]
Thus,
\[
\inner{
P_{\mathcal R_{\mathcal Y}(X)}\Delta X
}{
P_{\mathcal U(X)^\perp}
\bigl(J_{0:T}(X)^*Y\bigr)
}
=
\inner{
P_{\mathcal R_{\mathcal Y}(X)}\Delta X
}{
J_{0:T}(X)^*Y
}.
\]
Applying the Frobenius adjoint relation again gives
\[
\inner{
P_{\mathcal R_{\mathcal Y}(X)}\Delta X
}{
J_{0:T}(X)^*Y
}
=
\inner{
J_{0:T}(X)
\bigl[
P_{\mathcal R_{\mathcal Y}(X)}\Delta X
\bigr]
}{Y}.
\]

Combining the displayed equalities, for every \(Y\in\mathcal Y\),
\[
\begin{aligned}
\inner{J_{0:T}(X)\Delta X}{Y}
&=
\inner{\Delta X}{J_{0:T}(X)^*Y} \\
&=
\inner{\Delta X}{
P_{\mathcal U(X)^\perp}
\bigl(J_{0:T}(X)^*Y\bigr)
} \\
&=
\inner{
P_{\mathcal R_{\mathcal Y}(X)}\Delta X
}{
P_{\mathcal U(X)^\perp}
\bigl(J_{0:T}(X)^*Y\bigr)
} \\
&=
\inner{
P_{\mathcal R_{\mathcal Y}(X)}\Delta X
}{
J_{0:T}(X)^*Y
} \\
&=
\inner{
J_{0:T}(X)
\bigl[
P_{\mathcal R_{\mathcal Y}(X)}\Delta X
\bigr]
}{Y}.
\end{aligned}
\]
Therefore, the two hidden-state perturbations
\(J_{0:T}(X)\Delta X\) and
\(J_{0:T}(X)[P_{\mathcal R_{\mathcal Y}(X)}\Delta X]\)
have the same inner product with every \(Y\in\mathcal Y\). Equivalently, their projections onto \(\mathcal Y\) are equal:
\[
P_{\mathcal Y}J_{0:T}(X)\Delta X
=
P_{\mathcal Y}J_{0:T}(X)
\bigl[
P_{\mathcal R_{\mathcal Y}(X)}\Delta X
\bigr].
\]
Since \(\Phi_{\mathcal Y}(X)=P_{\mathcal Y}H^{(T)}(X)\), its first-order variation is
\[
D\Phi_{\mathcal Y}(X)[\Delta X]
=
P_{\mathcal Y}J_{0:T}(X)\Delta X.
\]
Thus,
\[
D\Phi_{\mathcal Y}(X)[\Delta X]
=
P_{\mathcal Y}J_{0:T}(X)\Delta X
=
P_{\mathcal Y}J_{0:T}(X)
\bigl[
P_{\mathcal R_{\mathcal Y}(X)}\Delta X
\bigr].
\]

In particular, if \(\Delta X\perp\mathcal R_{\mathcal Y}(X)\), then
\(P_{\mathcal R_{\mathcal Y}(X)}\Delta X=0\). The identity above gives
\[
D\Phi_{\mathcal Y}(X)[\Delta X]=0.
\]
Thus orthogonal directions in \(\mathcal P(X)\) have no first-order effect on the target-behavior signal. This completes the proof.
\end{proof}

\subsection{Proof of Proposition~\ref{prop:exact_decomposition}}
\label{app:proof_decomp}

\begin{proof}
Fix \(X\in\mathcal X\) and \(Y\in\mathcal Y\). All quantities below are evaluated at this fixed \(X\) and for this fixed target-behavior direction \(Y\).

We first prove the recurrence. For each \(t=0,\ldots,T-1\), the decomposition
\[
G_{t+1}=A_{t+1}+R_{t+1}
\]
gives
\[
\begin{aligned}
G_t
&=
J_t(X)^*G_{t+1} \\
&=
J_t(X)^*
\bigl(A_{t+1}+R_{t+1}\bigr) \\
&=
J_t(X)^*A_{t+1}
+
J_t(X)^*R_{t+1}.
\end{aligned}
\]
Projecting onto \(\mathcal U_t^\perp\), we obtain
\[
\begin{aligned}
R_t
&=
P_{\mathcal U_t^\perp}G_t \\
&=
P_{\mathcal U_t^\perp}
\bigl(
J_t(X)^*A_{t+1}
+
J_t(X)^*R_{t+1}
\bigr) \\
&=
P_{\mathcal U_t^\perp}
\bigl(J_t(X)^*A_{t+1}\bigr)
+
P_{\mathcal U_t^\perp}
\bigl(J_t(X)^*R_{t+1}\bigr).
\end{aligned}
\]
By the definition of the leakage source,
\[
B_t
=
P_{\mathcal U_t^\perp}
\bigl(J_t(X)^*A_{t+1}\bigr).
\]
It remains to show that the second projection is redundant. For any \(U\in\mathcal U_t\), the forward recursion \(\mathcal U_{t+1}=J_t(X)\mathcal U_t\) gives \(J_t(X)U\in\mathcal U_{t+1}\). Since \(R_{t+1}=P_{\mathcal U_{t+1}^\perp}G_{t+1}\), we have \(R_{t+1}\in\mathcal U_{t+1}^\perp\). Therefore,
\[
\begin{aligned}
\inner{J_t(X)^*R_{t+1}}{U}
&=
\inner{R_{t+1}}{J_t(X)U} = 0.
\end{aligned}
\]
Since this holds for all \(U\in\mathcal U_t\), we have
\[
J_t(X)^*R_{t+1}\in\mathcal U_t^\perp.
\]
Thus
\[
P_{\mathcal U_t^\perp}
\bigl(J_t(X)^*R_{t+1}\bigr)
=
J_t(X)^*R_{t+1}.
\]
Substituting this into the projected decomposition yields
\[
R_t
=
B_t
+
J_t(X)^*R_{t+1},
\]
which proves \eqref{eq:recurrence}.

We now iterate the recurrence. Starting from \(t=0\),
\[
R_0
=
B_0+J_0(X)^*R_1.
\]
Applying the recurrence at \(t=1\), we have
\[
R_1
=
B_1+J_1(X)^*R_2.
\]
Substituting this into the expression for \(R_0\) gives
\[
\begin{aligned}
R_0
&=
B_0
+
J_0(X)^*
\bigl(
B_1+J_1(X)^*R_2
\bigr) \\
&=
B_0
+
J_0(X)^*B_1
+
\bigl(J_0(X)^*\circ J_1(X)^*\bigr)R_2.
\end{aligned}
\]
Applying the recurrence again at \(t=2\),
\[
R_2
=
B_2+J_2(X)^*R_3,
\]
and substituting gives
\[
\begin{aligned}
R_0
&=
B_0
+
J_0(X)^*B_1
+
\bigl(J_0(X)^*\circ J_1(X)^*\bigr)
\bigl(
B_2+J_2(X)^*R_3
\bigr) \\
&=
B_0
+
J_0(X)^*B_1
+
\bigl(J_0(X)^*\circ J_1(X)^*\bigr)B_2
+
\bigl(J_0(X)^*\circ J_1(X)^*\circ J_2(X)^*\bigr)R_3.
\end{aligned}
\]

In general, for any \(s=0,\ldots,T-1\), we claim that
\begin{equation}
R_0
=
\sum_{j=0}^{s}
\bigl(J_0(X)^*\circ\cdots\circ J_{j-1}(X)^*\bigr)B_j
+
\bigl(J_0(X)^*\circ\cdots\circ J_s(X)^*\bigr)R_{s+1},
\label{eq:proof_decomp_partial_direction}
\end{equation}
where the empty composition is understood as the identity. The case \(s=0\) is
\[
R_0
=
B_0+J_0(X)^*R_1.
\]
Assume \eqref{eq:proof_decomp_partial_direction} holds for some \(s<T-1\). By the recurrence at \(t=s+1\),
\[
R_{s+1}
=
B_{s+1}+J_{s+1}(X)^*R_{s+2}.
\]
Therefore,
\[
\begin{aligned}
&
\bigl(J_0(X)^*\circ\cdots\circ J_s(X)^*\bigr)R_{s+1}
\\
&=
\bigl(J_0(X)^*\circ\cdots\circ J_s(X)^*\bigr)
\bigl(
B_{s+1}+J_{s+1}(X)^*R_{s+2}
\bigr)
\\
&=
\bigl(J_0(X)^*\circ\cdots\circ J_s(X)^*\bigr)B_{s+1}
+
\bigl(J_0(X)^*\circ\cdots\circ J_s(X)^*\circ J_{s+1}(X)^*\bigr)R_{s+2}
\\
&=
\bigl(J_0(X)^*\circ\cdots\circ J_s(X)^*\bigr)B_{s+1}
+
\bigl(J_0(X)^*\circ\cdots\circ J_{s+1}(X)^*\bigr)R_{s+2}.
\end{aligned}
\]
Substituting this into \eqref{eq:proof_decomp_partial_direction} gives
\[
R_0
=
\sum_{j=0}^{s+1}
\bigl(J_0(X)^*\circ\cdots\circ J_{j-1}(X)^*\bigr)B_j
+
\bigl(J_0(X)^*\circ\cdots\circ J_{s+1}(X)^*\bigr)R_{s+2}.
\]
This proves the induction step.

Setting \(s=T-1\) in \eqref{eq:proof_decomp_partial_direction}, we obtain
\[
R_0
=
\sum_{t=0}^{T-1}
\bigl(J_0(X)^*\circ\cdots\circ J_{t-1}(X)^*\bigr)B_t
+
\bigl(J_0(X)^*\circ\cdots\circ J_{T-1}(X)^*\bigr)R_T.
\]
By the definitions
\[
S_t(Y;X)
:=
\bigl(J_0(X)^*\circ\cdots\circ J_{t-1}(X)^*\bigr)B_t,
\qquad
S_T(Y;X)
:=
\bigl(J_0(X)^*\circ\cdots\circ J_{T-1}(X)^*\bigr)R_T,
\]
we have
\[
R_0
=
\sum_{t=0}^{T-1}S_t(Y;X)+S_T(Y;X).
\]
Since \(R_0=R_Y(X)\), this gives
\[
R_Y(X)
=
R_0
=
\sum_{t=0}^{T-1}S_t(Y;X)+S_T(Y;X),
\]
which proves \eqref{eq:exact_sum_compact}.

Finally, we prove the statement for the whole target-behavior subspace \(\mathcal Y\). By the definition of RED,
\[
\mathcal R_{\mathcal Y}(X)
=
\left\{
P_{\mathcal U(X)^\perp}
\bigl(J_{0:T}(X)^*Y\bigr)
:
Y\in\mathcal Y
\right\}.
\]
By the definition of \(R_Y(X)\),
\[
P_{\mathcal U(X)^\perp}
\bigl(J_{0:T}(X)^*Y\bigr)
=
R_Y(X).
\]
Hence
\[
\mathcal R_{\mathcal Y}(X)
=
\left\{
R_Y(X):Y\in\mathcal Y
\right\}.
\]
Using the direction-wise decomposition just proved,
\[
R_Y(X)
=
\sum_{t=0}^{T-1}S_t(Y;X)+S_T(Y;X)
\qquad
\text{for every }Y\in\mathcal Y.
\]
Therefore,
\[
\mathcal R_{\mathcal Y}(X)
=
\left\{
\sum_{t=0}^{T-1}S_t(Y;X)+S_T(Y;X)
:
Y\in\mathcal Y
\right\},
\]
which proves \eqref{eq:red_set_of_direction_decomp}. This completes the proof.
\end{proof}

\subsection{Derivation of the operator-level decomposition for pre-norm residual LLMs}
\label{app:operators}

This appendix specializes the generic source-and-transport decomposition in
\eqref{eq:recurrence}--\eqref{eq:exact_sum_compact} to the operator structure of
modern pre-norm residual LLMs. The goal is to derive the operator-level
decomposition of each refusal-escape direction into normalization,
self-attention, MLP, residual-wiring, and terminal-source contributions. We fix
\(X\in\mathcal X\) and \(Y\in\mathcal Y\); all quantities below are evaluated at
this fixed input and target-behavior direction.

We consider a pre-norm residual LLM. For each layer \(l=0,\dots,L-1\),
\begin{equation}
P_l = N_l^{\mathrm{attn}}(X_l), \qquad
O_l = \mathrm{Attn}_l(P_l), \qquad
Z_l = X_l + O_l,
\label{eq:app_attn_block}
\end{equation}
\begin{equation}
Q_l = N_l^{\mathrm{mlp}}(Z_l), \qquad
M_l = \mathrm{MLP}_l(Q_l), \qquad
X_{l+1} = Z_l + M_l,
\label{eq:app_mlp_block}
\end{equation}
followed by the final normalization
\begin{equation}
H^{(T)} = N^{\mathrm{fn}}(X_L).
\label{eq:app_final_norm_arch}
\end{equation}
Here \(N_l^{\mathrm{attn}}\), \(N_l^{\mathrm{mlp}}\), and \(N^{\mathrm{fn}}\) denote normalization operators, such as LayerNorm or RMSNorm.

Starting from the harmful-semantics-sensitive input subspace
\(\mathcal U_{X_0}:=\mathcal U(X)\), we propagate the subspaces through the
component operators in the architecture above:
\begin{equation}
\mathcal U_{P_l}:=J_{N_l^{\mathrm{attn}}}(X_l)\mathcal U_{X_l},\qquad
\mathcal U_{O_l}:=J_{\mathrm{Attn}_l}(P_l)\mathcal U_{P_l},\qquad
\mathcal U_{Z_l}:=J_{\mathcal B^{l,\mathrm{attn}}}(X_l)\mathcal U_{X_l},
\label{eq:app_attn_subspaces}
\end{equation}
\begin{equation}
\mathcal U_{Q_l}:=J_{N_l^{\mathrm{mlp}}}(Z_l)\mathcal U_{Z_l},\qquad
\mathcal U_{M_l}:=J_{\mathrm{MLP}_l}(Q_l)\mathcal U_{Q_l},\qquad
\mathcal U_{X_{l+1}}:=J_{\mathcal B^{l,\mathrm{mlp}}}(Z_l)\mathcal U_{Z_l},
\label{eq:app_mlp_subspaces}
\end{equation}
and finally
\begin{equation}
\mathcal U_T:=J_{N^{\mathrm{fn}}}(X_L)\mathcal U_{X_L}.
\label{eq:app_final_subspace}
\end{equation}

\subsubsection{Leakage sources in residual blocks}

We begin with a generic post-add residual block
\begin{equation}
\mathcal B(x)=x+F(x),
\label{eq:app_residual_block}
\end{equation}
where \(x\) is the identity input, \(F(x)\) is the learnable body output, and
\(\mathcal B(x)\) is the block output.

Let \(\mathcal U_{\mathrm{in}}\) be the propagated harmful-semantics-sensitive
subspace at the block input, and let
\[
\mathcal U_{\mathrm{out}}:=J_{\mathcal B}(x)\mathcal U_{\mathrm{in}}
\]
be the corresponding subspace at the block output. Let \(G_{\mathrm{out}}\) be
the target-behavior sensitivity direction at the block output, and define its
harmful-semantics-aligned component by
\begin{equation}
A_{\mathrm{out}}:=P_{\mathcal U_{\mathrm{out}}}(G_{\mathrm{out}}).
\label{eq:app_A_out}
\end{equation}

\begin{proposition}[Two-way exact decomposition of residual-block leakage source]
\label{prop:app_residual_block}
Define the residual-block leakage source by
\begin{equation}
B^{\mathrm{blk}}
:=
P_{\mathcal U_{\mathrm{in}}^\perp}
\bigl(J_{\mathcal B}(x)^*A_{\mathrm{out}}\bigr).
\label{eq:app_block_local_def}
\end{equation}
Then
\begin{equation}
B^{\mathrm{blk}}
=
B_{\mathrm{id}}^{\mathrm{blk}}
+
B_{\mathrm{body}}^{\mathrm{blk}},
\label{eq:app_block_local_split}
\end{equation}
where
\begin{equation}
B_{\mathrm{id}}^{\mathrm{blk}}
:=
P_{\mathcal U_{\mathrm{in}}^\perp}(A_{\mathrm{out}}),
\label{eq:app_block_id}
\end{equation}
and
\begin{equation}
B_{\mathrm{body}}^{\mathrm{blk}}
:=
P_{\mathcal U_{\mathrm{in}}^\perp}
\bigl(J_F(x)^*A_{\mathrm{out}}\bigr).
\label{eq:app_block_body}
\end{equation}
\end{proposition}

\begin{proof}
For any perturbation \(\Delta x\),
\[
J_{\mathcal B}(x)\Delta x
=
\Delta x+J_F(x)\Delta x.
\]
Hence, for any backward direction \(G\),
\[
\begin{aligned}
\inner{J_{\mathcal B}(x)^*G}{\Delta x}
&=
\inner{G}{J_{\mathcal B}(x)\Delta x} \\
&=
\inner{G}{\Delta x+J_F(x)\Delta x} \\
&=
\inner{G}{\Delta x}
+
\inner{G}{J_F(x)\Delta x} \\
&=
\inner{G}{\Delta x}
+
\inner{J_F(x)^*G}{\Delta x} \\
&=
\inner{G+J_F(x)^*G}{\Delta x}.
\end{aligned}
\]
Since this holds for every \(\Delta x\),
\begin{equation}
J_{\mathcal B}(x)^*G
=
G+J_F(x)^*G.
\label{eq:app_block_adjoint}
\end{equation}
Substituting \(G=A_{\mathrm{out}}\) into \eqref{eq:app_block_adjoint} gives
\[
J_{\mathcal B}(x)^*A_{\mathrm{out}}
=
A_{\mathrm{out}}
+
J_F(x)^*A_{\mathrm{out}}.
\]
Projecting onto \(\mathcal U_{\mathrm{in}}^\perp\), we obtain
\[
\begin{aligned}
B^{\mathrm{blk}}
&=
P_{\mathcal U_{\mathrm{in}}^\perp}
\bigl(J_{\mathcal B}(x)^*A_{\mathrm{out}}\bigr) \\
&=
P_{\mathcal U_{\mathrm{in}}^\perp}
\bigl(A_{\mathrm{out}}+J_F(x)^*A_{\mathrm{out}}\bigr) \\
&=
P_{\mathcal U_{\mathrm{in}}^\perp}(A_{\mathrm{out}})
+
P_{\mathcal U_{\mathrm{in}}^\perp}
\bigl(J_F(x)^*A_{\mathrm{out}}\bigr) \\
&=
B_{\mathrm{id}}^{\mathrm{blk}}
+
B_{\mathrm{body}}^{\mathrm{blk}}.
\end{aligned}
\]
This proves the claim.
\end{proof}

Proposition~\ref{prop:app_residual_block} shows that the leakage source of a
post-add residual block splits exactly into two parts: one carried by the
identity edge and the other by the learnable body edge.

\subsubsection{Attention block}

For the attention block in layer \(l\), define
\[
F^{l,\mathrm{attn}}:=\mathrm{Attn}_l\circ N_l^{\mathrm{attn}},
\qquad
\mathcal B^{l,\mathrm{attn}}(X_l)=X_l+F^{l,\mathrm{attn}}(X_l)=Z_l.
\]
Let \(G_{Z_l}\) be the target-behavior sensitivity direction at \(Z_l\), and
define the harmful-semantics-aligned component at the attention-block output by
\begin{equation}
A_{Z_l}:=P_{\mathcal U_{Z_l}}(G_{Z_l}).
\label{eq:app_A_Z}
\end{equation}

\begin{proposition}[Exact decomposition of attention-block leakage source]
\label{prop:app_attn_block}
Define the attention-block leakage source at \(X_l\) by
\begin{equation}
B^{l,\mathrm{attn}}
:=
P_{\mathcal U_{X_l}^\perp}
\bigl(J_{\mathcal B^{l,\mathrm{attn}}}(X_l)^*A_{Z_l}\bigr).
\label{eq:app_attn_source_def}
\end{equation}
Then
\begin{equation}
B^{l,\mathrm{attn}}
=
B^{l,\mathrm{attn}}_{\mathrm{norm1}}
+
B^{l,\mathrm{attn}}_{\mathrm{attn}}
+
B^{l,\mathrm{attn}}_{\mathrm{res1.add}}
+
B^{l,\mathrm{attn}}_{\mathrm{res1.id}},
\label{eq:app_attn_local_sum}
\end{equation}
where
\begin{align}
B^{l,\mathrm{attn}}_{\mathrm{res1.id}}
&:=P_{\mathcal U_{X_l}^\perp}(A_{Z_l}),
\label{eq:app_B_res1_id}
\\
B^{l,\mathrm{attn}}_{\mathrm{res1.add}}
&:=J_{N_l^{\mathrm{attn}}}(X_l)^*
J_{\mathrm{Attn}_l}(P_l)^*
\bigl(P_{\mathcal U_{O_l}^\perp}(A_{Z_l})\bigr),
\label{eq:app_B_res1_add}
\\
B^{l,\mathrm{attn}}_{\mathrm{attn}}
&:=J_{N_l^{\mathrm{attn}}}(X_l)^*
\bigl(P_{\mathcal U_{P_l}^\perp}
(J_{\mathrm{Attn}_l}(P_l)^*P_{\mathcal U_{O_l}}(A_{Z_l}))\bigr),
\label{eq:app_B_attn}
\\
B^{l,\mathrm{attn}}_{\mathrm{norm1}}
&:=P_{\mathcal U_{X_l}^\perp}
\bigl(J_{N_l^{\mathrm{attn}}}(X_l)^*
P_{\mathcal U_{P_l}}
(J_{\mathrm{Attn}_l}(P_l)^*P_{\mathcal U_{O_l}}(A_{Z_l}))\bigr).
\label{eq:app_B_norm1}
\end{align}
\end{proposition}

\begin{proof}
We prove the decomposition in three steps.

\paragraph{Step 1: Split the attention block into identity and body contributions.}
Applying Proposition~\ref{prop:app_residual_block} to
\[
\mathcal B^{l,\mathrm{attn}}(X_l)=X_l+F^{l,\mathrm{attn}}(X_l),
\qquad
F^{l,\mathrm{attn}}=\mathrm{Attn}_l\circ N_l^{\mathrm{attn}},
\]
gives
\[
\begin{aligned}
B^{l,\mathrm{attn}}
&=
P_{\mathcal U_{X_l}^\perp}
\bigl(J_{\mathcal B^{l,\mathrm{attn}}}(X_l)^*A_{Z_l}\bigr) \\
&=
P_{\mathcal U_{X_l}^\perp}(A_{Z_l})
+
P_{\mathcal U_{X_l}^\perp}
\bigl(J_{F^{l,\mathrm{attn}}}(X_l)^*A_{Z_l}\bigr).
\end{aligned}
\]
The first term is \(B^{l,\mathrm{attn}}_{\mathrm{res1.id}}\). It remains to
decompose the second term.

\paragraph{Step 2: Split the body-induced term at the residual-add input.}
Decompose \(A_{Z_l}\) relative to \(\mathcal U_{O_l}\):
\[
A_{Z_l}
=
P_{\mathcal U_{O_l}}(A_{Z_l})
+
P_{\mathcal U_{O_l}^\perp}(A_{Z_l}).
\]
Since
\[
J_{F^{l,\mathrm{attn}}}(X_l)^*
=
J_{N_l^{\mathrm{attn}}}(X_l)^*
J_{\mathrm{Attn}_l}(P_l)^*,
\]
we have
\[
\begin{aligned}
&
P_{\mathcal U_{X_l}^\perp}
\bigl(J_{F^{l,\mathrm{attn}}}(X_l)^*A_{Z_l}\bigr)
\\
&=
P_{\mathcal U_{X_l}^\perp}
\bigl(
J_{N_l^{\mathrm{attn}}}(X_l)^*
J_{\mathrm{Attn}_l}(P_l)^*
[
P_{\mathcal U_{O_l}}(A_{Z_l})
+
P_{\mathcal U_{O_l}^\perp}(A_{Z_l})
]
\bigr)
\\
&=
P_{\mathcal U_{X_l}^\perp}
\bigl(
J_{N_l^{\mathrm{attn}}}(X_l)^*
J_{\mathrm{Attn}_l}(P_l)^*
P_{\mathcal U_{O_l}}(A_{Z_l})
\bigr)
+
P_{\mathcal U_{X_l}^\perp}
\bigl(
J_{N_l^{\mathrm{attn}}}(X_l)^*
J_{\mathrm{Attn}_l}(P_l)^*
P_{\mathcal U_{O_l}^\perp}(A_{Z_l})
\bigr).
\end{aligned}
\]
For the second term, let \(W=P_{\mathcal U_{O_l}^\perp}(A_{Z_l})\). For any
\(V\in\mathcal U_{X_l}\),
\[
J_{N_l^{\mathrm{attn}}}(X_l)V\in\mathcal U_{P_l},
\qquad
J_{\mathrm{Attn}_l}(P_l)J_{N_l^{\mathrm{attn}}}(X_l)V\in\mathcal U_{O_l}.
\]
Since \(W\in\mathcal U_{O_l}^\perp\),
\[
\begin{aligned}
\inner{
J_{N_l^{\mathrm{attn}}}(X_l)^*
J_{\mathrm{Attn}_l}(P_l)^*W
}{V}
&=
\inner{
J_{\mathrm{Attn}_l}(P_l)^*W
}{
J_{N_l^{\mathrm{attn}}}(X_l)V
} \\
&=
\inner{
W
}{
J_{\mathrm{Attn}_l}(P_l)J_{N_l^{\mathrm{attn}}}(X_l)V
} \\
&=
0.
\end{aligned}
\]
Thus
\[
J_{N_l^{\mathrm{attn}}}(X_l)^*
J_{\mathrm{Attn}_l}(P_l)^*W
\in
\mathcal U_{X_l}^\perp,
\]
and therefore
\[
P_{\mathcal U_{X_l}^\perp}
\bigl(
J_{N_l^{\mathrm{attn}}}(X_l)^*
J_{\mathrm{Attn}_l}(P_l)^*
P_{\mathcal U_{O_l}^\perp}(A_{Z_l})
\bigr)
=
B^{l,\mathrm{attn}}_{\mathrm{res1.add}}.
\]

\paragraph{Step 3: Split the remaining aligned body term into attention and normalization contributions.}
Decompose the attention pullback of \(P_{\mathcal U_{O_l}}(A_{Z_l})\) relative to \(\mathcal U_{P_l}\):
\[
\begin{aligned}
J_{\mathrm{Attn}_l}(P_l)^*P_{\mathcal U_{O_l}}(A_{Z_l})
=
P_{\mathcal U_{P_l}}
\bigl(
J_{\mathrm{Attn}_l}(P_l)^*
P_{\mathcal U_{O_l}}(A_{Z_l})
\bigr) 
+
P_{\mathcal U_{P_l}^\perp}
\bigl(
J_{\mathrm{Attn}_l}(P_l)^*
P_{\mathcal U_{O_l}}(A_{Z_l})
\bigr).
\end{aligned}
\]
Substituting this into the remaining term gives
\[
\begin{aligned}
&
P_{\mathcal U_{X_l}^\perp}
\bigl(
J_{N_l^{\mathrm{attn}}}(X_l)^*
J_{\mathrm{Attn}_l}(P_l)^*
P_{\mathcal U_{O_l}}(A_{Z_l})
\bigr)
\\
&=
P_{\mathcal U_{X_l}^\perp}
\bigl(
J_{N_l^{\mathrm{attn}}}(X_l)^*
[
P_{\mathcal U_{P_l}}
(J_{\mathrm{Attn}_l}(P_l)^*P_{\mathcal U_{O_l}}(A_{Z_l}))
+
P_{\mathcal U_{P_l}^\perp}
(J_{\mathrm{Attn}_l}(P_l)^*P_{\mathcal U_{O_l}}(A_{Z_l}))
]
\bigr)
\\
&=
P_{\mathcal U_{X_l}^\perp}
\bigl(
J_{N_l^{\mathrm{attn}}}(X_l)^*
P_{\mathcal U_{P_l}}
(J_{\mathrm{Attn}_l}(P_l)^*P_{\mathcal U_{O_l}}(A_{Z_l}))
\bigr)
\\
& \quad +
P_{\mathcal U_{X_l}^\perp}
\bigl(
J_{N_l^{\mathrm{attn}}}(X_l)^*
P_{\mathcal U_{P_l}^\perp}
(J_{\mathrm{Attn}_l}(P_l)^*P_{\mathcal U_{O_l}}(A_{Z_l}))
\bigr).
\end{aligned}
\]
Let
\[
W
=
P_{\mathcal U_{P_l}^\perp}
\bigl(
J_{\mathrm{Attn}_l}(P_l)^*
P_{\mathcal U_{O_l}}(A_{Z_l})
\bigr).
\]
For any \(V\in\mathcal U_{X_l}\),
\[
J_{N_l^{\mathrm{attn}}}(X_l)V\in\mathcal U_{P_l},
\qquad
W\in\mathcal U_{P_l}^\perp,
\]
so
\[
\begin{aligned}
\inner{
J_{N_l^{\mathrm{attn}}}(X_l)^*W
}{V}
&=
\inner{
W
}{
J_{N_l^{\mathrm{attn}}}(X_l)V
} 
=
0.
\end{aligned}
\]
Thus \(J_{N_l^{\mathrm{attn}}}(X_l)^*W\in\mathcal U_{X_l}^\perp\), and
the projection is redundant on this term. Hence
\[
P_{\mathcal U_{X_l}^\perp}
\bigl(J_{N_l^{\mathrm{attn}}}(X_l)^*W\bigr)
=
J_{N_l^{\mathrm{attn}}}(X_l)^*W
=
B^{l,\mathrm{attn}}_{\mathrm{attn}}.
\]
The first term is 
\[
P_{\mathcal U_{X_l}^\perp}
\bigl(
J_{N_l^{\mathrm{attn}}}(X_l)^*
P_{\mathcal U_{P_l}}
(J_{\mathrm{Attn}_l}(P_l)^*P_{\mathcal U_{O_l}}(A_{Z_l}))
\bigr)
= 
B^{l,\mathrm{attn}}_{\mathrm{norm1}}.
\]
Combining the four terms gives \eqref{eq:app_attn_local_sum}.
\end{proof}

\subsubsection{MLP block}

For the MLP block in layer \(l\), define
\[
F^{l,\mathrm{mlp}}:=\mathrm{MLP}_l\circ N_l^{\mathrm{mlp}},
\qquad
\mathcal B^{l,\mathrm{mlp}}(Z_l)=Z_l+F^{l,\mathrm{mlp}}(Z_l)=X_{l+1}.
\]
Let \(G_{X_{l+1}}\) be the target-behavior sensitivity direction at \(X_{l+1}\), and define
\begin{equation}
A_{X_{l+1}}:=P_{\mathcal U_{X_{l+1}}}(G_{X_{l+1}}).
\label{eq:app_A_Xnext}
\end{equation}

\begin{proposition}[Exact decomposition of MLP-block leakage source]
\label{prop:app_mlp_block}
Define the MLP-block leakage source at \(Z_l\) by
\begin{equation}
B^{l,\mathrm{mlp}}
:=
P_{\mathcal U_{Z_l}^\perp}
\bigl(J_{\mathcal B^{l,\mathrm{mlp}}}(Z_l)^*A_{X_{l+1}}\bigr).
\label{eq:app_mlp_source_def}
\end{equation}
Then
\begin{equation}
B^{l,\mathrm{mlp}}
=
B^{l,\mathrm{mlp}}_{\mathrm{norm2}}
+
B^{l,\mathrm{mlp}}_{\mathrm{mlp}}
+
B^{l,\mathrm{mlp}}_{\mathrm{res2.add}}
+
B^{l,\mathrm{mlp}}_{\mathrm{res2.id}},
\label{eq:app_mlp_local_sum}
\end{equation}
where
\begin{align}
B^{l,\mathrm{mlp}}_{\mathrm{res2.id}}
&:=P_{\mathcal U_{Z_l}^\perp}(A_{X_{l+1}}),
\label{eq:app_B_res2_id}
\\
B^{l,\mathrm{mlp}}_{\mathrm{res2.add}}
&:=
J_{N_l^{\mathrm{mlp}}}(Z_l)^*
J_{\mathrm{MLP}_l}(Q_l)^*
\bigl(P_{\mathcal U_{M_l}^\perp}(A_{X_{l+1}})\bigr),
\label{eq:app_B_res2_add}
\\
B^{l,\mathrm{mlp}}_{\mathrm{mlp}}
&:=J_{N_l^{\mathrm{mlp}}}(Z_l)^*
\bigl(P_{\mathcal U_{Q_l}^\perp}
(J_{\mathrm{MLP}_l}(Q_l)^*P_{\mathcal U_{M_l}}(A_{X_{l+1}}))\bigr),
\label{eq:app_B_mlp}
\\
B^{l,\mathrm{mlp}}_{\mathrm{norm2}}
&:=P_{\mathcal U_{Z_l}^\perp}
\bigl(J_{N_l^{\mathrm{mlp}}}(Z_l)^*
P_{\mathcal U_{Q_l}}
(J_{\mathrm{MLP}_l}(Q_l)^*P_{\mathcal U_{M_l}}(A_{X_{l+1}}))\bigr).
\label{eq:app_B_norm2}
\end{align}
\end{proposition}

\begin{proof}
The proof mirrors the attention-block case, with the attention operator replaced by the MLP operator and \(X_l,P_l,O_l,Z_l\) replaced by \(Z_l,Q_l,M_l,X_{l+1}\). We give the full calculation.

Applying Proposition~\ref{prop:app_residual_block} to
\[
\mathcal B^{l,\mathrm{mlp}}(Z_l)=Z_l+F^{l,\mathrm{mlp}}(Z_l),
\qquad
F^{l,\mathrm{mlp}}=\mathrm{MLP}_l\circ N_l^{\mathrm{mlp}},
\]
gives
\[
\begin{aligned}
B^{l,\mathrm{mlp}}
&=
P_{\mathcal U_{Z_l}^\perp}
\bigl(J_{\mathcal B^{l,\mathrm{mlp}}}(Z_l)^*A_{X_{l+1}}\bigr) \\
&=
P_{\mathcal U_{Z_l}^\perp}(A_{X_{l+1}})
+
P_{\mathcal U_{Z_l}^\perp}
\bigl(J_{F^{l,\mathrm{mlp}}}(Z_l)^*A_{X_{l+1}}\bigr).
\end{aligned}
\]
The first term is \(B^{l,\mathrm{mlp}}_{\mathrm{res2.id}}\).

Decompose \(A_{X_{l+1}}\) relative to \(\mathcal U_{M_l}\):
\[
A_{X_{l+1}}
=
P_{\mathcal U_{M_l}}(A_{X_{l+1}})
+
P_{\mathcal U_{M_l}^\perp}(A_{X_{l+1}}).
\]
Using
\[
J_{F^{l,\mathrm{mlp}}}(Z_l)^*
=
J_{N_l^{\mathrm{mlp}}}(Z_l)^*
J_{\mathrm{MLP}_l}(Q_l)^*,
\]
we get
\[
\begin{aligned}
&
P_{\mathcal U_{Z_l}^\perp}
\bigl(J_{F^{l,\mathrm{mlp}}}(Z_l)^*A_{X_{l+1}}\bigr)
\\
&=
P_{\mathcal U_{Z_l}^\perp}
\bigl(
J_{N_l^{\mathrm{mlp}}}(Z_l)^*
J_{\mathrm{MLP}_l}(Q_l)^*
P_{\mathcal U_{M_l}}(A_{X_{l+1}})
\bigr)
+
P_{\mathcal U_{Z_l}^\perp}
\bigl(
J_{N_l^{\mathrm{mlp}}}(Z_l)^*
J_{\mathrm{MLP}_l}(Q_l)^*
P_{\mathcal U_{M_l}^\perp}(A_{X_{l+1}})
\bigr).
\end{aligned}
\]
For the second term, let \(W=P_{\mathcal U_{M_l}^\perp}(A_{X_{l+1}})\). For any
\(V\in\mathcal U_{Z_l}\),
\[
J_{N_l^{\mathrm{mlp}}}(Z_l)V\in\mathcal U_{Q_l},
\qquad
J_{\mathrm{MLP}_l}(Q_l)J_{N_l^{\mathrm{mlp}}}(Z_l)V\in\mathcal U_{M_l}.
\]
Since \(W\in\mathcal U_{M_l}^\perp\),
\[
\begin{aligned}
\inner{
J_{N_l^{\mathrm{mlp}}}(Z_l)^*
J_{\mathrm{MLP}_l}(Q_l)^*W
}{V}
&=
\inner{
W
}{
J_{\mathrm{MLP}_l}(Q_l)J_{N_l^{\mathrm{mlp}}}(Z_l)V
} 
=
0.
\end{aligned}
\]
Thus the projection \(P_{\mathcal U_{Z_l}^\perp}\) is redundant on this term, and the second term is
\(B^{l,\mathrm{mlp}}_{\mathrm{res2.add}}\).

Next decompose the MLP pullback of \(P_{\mathcal U_{M_l}}(A_{X_{l+1}})\) relative to \(\mathcal U_{Q_l}\):
\[
\begin{aligned}
J_{\mathrm{MLP}_l}(Q_l)^*P_{\mathcal U_{M_l}}(A_{X_{l+1}})
&=
P_{\mathcal U_{Q_l}}
\bigl(
J_{\mathrm{MLP}_l}(Q_l)^*
P_{\mathcal U_{M_l}}(A_{X_{l+1}})
\bigr) 
\\
& \quad +
P_{\mathcal U_{Q_l}^\perp}
\bigl(
J_{\mathrm{MLP}_l}(Q_l)^*
P_{\mathcal U_{M_l}}(A_{X_{l+1}})
\bigr).
\end{aligned}
\]
Substituting this decomposition gives
\[
\begin{aligned}
&
P_{\mathcal U_{Z_l}^\perp}
\bigl(
J_{N_l^{\mathrm{mlp}}}(Z_l)^*
J_{\mathrm{MLP}_l}(Q_l)^*
P_{\mathcal U_{M_l}}(A_{X_{l+1}})
\bigr)
\\
&=
P_{\mathcal U_{Z_l}^\perp}
\bigl(
J_{N_l^{\mathrm{mlp}}}(Z_l)^*
P_{\mathcal U_{Q_l}}
(J_{\mathrm{MLP}_l}(Q_l)^*P_{\mathcal U_{M_l}}(A_{X_{l+1}}))
\bigr)
\\
& \quad +
P_{\mathcal U_{Z_l}^\perp}
\bigl(
J_{N_l^{\mathrm{mlp}}}(Z_l)^*
P_{\mathcal U_{Q_l}^\perp}
(J_{\mathrm{MLP}_l}(Q_l)^*P_{\mathcal U_{M_l}}(A_{X_{l+1}}))
\bigr).
\end{aligned}
\]
Let
\[
W=
P_{\mathcal U_{Q_l}^\perp}
\bigl(
J_{\mathrm{MLP}_l}(Q_l)^*
P_{\mathcal U_{M_l}}(A_{X_{l+1}})
\bigr).
\]
For any \(V\in\mathcal U_{Z_l}\), \(J_{N_l^{\mathrm{mlp}}}(Z_l)V\in\mathcal U_{Q_l}\), while \(W\in\mathcal U_{Q_l}^\perp\). Therefore,
\[
\begin{aligned}
\inner{
J_{N_l^{\mathrm{mlp}}}(Z_l)^*W
}{V}
&=
\inner{W}{J_{N_l^{\mathrm{mlp}}}(Z_l)V} 
=
0.
\end{aligned}
\]
Thus the projection is redundant on this term, and it is exactly
\(B^{l,\mathrm{mlp}}_{\mathrm{mlp}}\). The remaining projected term is
\(B^{l,\mathrm{mlp}}_{\mathrm{norm2}}\). Combining all four terms proves
\eqref{eq:app_mlp_local_sum}.
\end{proof}

\subsubsection{Final normalization and the transported input-side decomposition}

The final normalization is
\[
H^{(T)}=N^{\mathrm{fn}}(X_L).
\]
At the terminal hidden state, decompose the target-behavior direction \(Y\) as
\[
Y=A_T+R_T,
\qquad
A_T:=P_{\mathcal U_T}(Y),
\qquad
R_T:=P_{\mathcal U_T^\perp}(Y).
\]
The final-normalization leakage source is
\begin{equation}
B_{\mathrm{fn}}
:=
P_{\mathcal U_{X_L}^\perp}
\bigl(J_{N^{\mathrm{fn}}}(X_L)^*A_T\bigr).
\label{eq:app_B_fn}
\end{equation}

We now define the cumulative adjoint transports from each local source position back to the input side. For an attention-block source at \(X_l\), define
\[
J_{l,\mathrm{attn}}^{\leftarrow}
:=
J_{\mathcal B^{0,\mathrm{attn}}}(X_0)^*
\circ
J_{\mathcal B^{0,\mathrm{mlp}}}(Z_0)^*
\circ\cdots\circ
J_{\mathcal B^{l-1,\mathrm{attn}}}(X_{l-1})^*
\circ
J_{\mathcal B^{l-1,\mathrm{mlp}}}(Z_{l-1})^*,
\]
with the empty composition understood as the identity when \(l=0\). For an MLP-block source at \(Z_l\), define
\[
J_{l,\mathrm{mlp}}^{\leftarrow}
:=
J_{l,\mathrm{attn}}^{\leftarrow}
\circ
J_{\mathcal B^{l,\mathrm{attn}}}(X_l)^*.
\]
Equivalently, \(J_{l,\mathrm{mlp}}^{\leftarrow}\) first transports a direction from \(Z_l\) back to \(X_l\) through the \(l\)-th attention block, and then transports it from \(X_l\) back to the input side. Finally, define the transport from \(X_L\) back to the input side by
\[
J_L^{\leftarrow}
:=
J_{\mathcal B^{0,\mathrm{attn}}}(X_0)^*
\circ
J_{\mathcal B^{0,\mathrm{mlp}}}(Z_0)^*
\circ\cdots\circ
J_{\mathcal B^{L-1,\mathrm{attn}}}(X_{L-1})^*
\circ
J_{\mathcal B^{L-1,\mathrm{mlp}}}(Z_{L-1})^*.
\]

Transporting the attention-block contributions gives
\[
S_{l,\mathrm{norm1}}(Y;X)
:=
J_{l,\mathrm{attn}}^{\leftarrow}
\bigl(B^{l,\mathrm{attn}}_{\mathrm{norm1}}\bigr),
\qquad
S_{l,\mathrm{attn}}(Y;X)
:=
J_{l,\mathrm{attn}}^{\leftarrow}
\bigl(B^{l,\mathrm{attn}}_{\mathrm{attn}}\bigr),
\]
\[
S_{l,\mathrm{res1.add}}(Y;X)
:=
J_{l,\mathrm{attn}}^{\leftarrow}
\bigl(B^{l,\mathrm{attn}}_{\mathrm{res1.add}}\bigr),
\qquad
S_{l,\mathrm{res1.id}}(Y;X)
:=
J_{l,\mathrm{attn}}^{\leftarrow}
\bigl(B^{l,\mathrm{attn}}_{\mathrm{res1.id}}\bigr).
\]
Similarly, transporting the MLP-block contributions gives
\[
S_{l,\mathrm{norm2}}(Y;X)
:=
J_{l,\mathrm{mlp}}^{\leftarrow}
\bigl(B^{l,\mathrm{mlp}}_{\mathrm{norm2}}\bigr),
\qquad
S_{l,\mathrm{mlp}}(Y;X)
:=
J_{l,\mathrm{mlp}}^{\leftarrow}
\bigl(B^{l,\mathrm{mlp}}_{\mathrm{mlp}}\bigr),
\]
\[
S_{l,\mathrm{res2.add}}(Y;X)
:=
J_{l,\mathrm{mlp}}^{\leftarrow}
\bigl(B^{l,\mathrm{mlp}}_{\mathrm{res2.add}}\bigr),
\qquad
S_{l,\mathrm{res2.id}}(Y;X)
:=
J_{l,\mathrm{mlp}}^{\leftarrow}
\bigl(B^{l,\mathrm{mlp}}_{\mathrm{res2.id}}\bigr).
\]
The transported final-normalization contribution is
\begin{equation}
S_{\mathrm{fn}}(Y;X)
:=
J_L^{\leftarrow}(B_{\mathrm{fn}}),
\label{eq:app_S_fn}
\end{equation}
and the transported terminal-source contribution is
\begin{equation}
S_T(Y;X)
:=
J_L^{\leftarrow}
\bigl(J_{N^{\mathrm{fn}}}(X_L)^*R_T\bigr).
\label{eq:app_S_T}
\end{equation}

We group these transported input-side contributions into the operator-level families used in the main text:
\[
S_{\mathrm{norm}}^\Sigma(Y;X)
:=
\sum_{l=0}^{L-1}
\bigl(
S_{l,\mathrm{norm1}}(Y;X)
+
S_{l,\mathrm{norm2}}(Y;X)
\bigr)
+
S_{\mathrm{fn}}(Y;X),
\]
\[
S_{\mathrm{attn}}^\Sigma(Y;X)
:=
\sum_{l=0}^{L-1}S_{l,\mathrm{attn}}(Y;X),
\qquad
S_{\mathrm{mlp}}^\Sigma(Y;X)
:=
\sum_{l=0}^{L-1}S_{l,\mathrm{mlp}}(Y;X),
\]
\[
S_{\mathrm{res}}^\Sigma(Y;X)
:=
\sum_{l=0}^{L-1}
\bigl(
S_{l,\mathrm{res1.add}}(Y;X)
+
S_{l,\mathrm{res1.id}}(Y;X)
+
S_{l,\mathrm{res2.add}}(Y;X)
+
S_{l,\mathrm{res2.id}}(Y;X)
\bigr).
\]

The refusal-escape direction \(R_Y(X)\) is the sum of all transported leakage contributions and the transported terminal-source contribution. Substituting the local decompositions of the attention and MLP leakage sources gives
\[
\begin{aligned}
R_Y(X)
&=
\sum_{l=0}^{L-1}
J_{l,\mathrm{attn}}^{\leftarrow}
\bigl(B^{l,\mathrm{attn}}\bigr)
+
\sum_{l=0}^{L-1}
J_{l,\mathrm{mlp}}^{\leftarrow}
\bigl(B^{l,\mathrm{mlp}}\bigr)
+
S_{\mathrm{fn}}(Y;X)
+
S_T(Y;X)
\\
&=
\sum_{l=0}^{L-1}
J_{l,\mathrm{attn}}^{\leftarrow}
\bigl(
B^{l,\mathrm{attn}}_{\mathrm{norm1}}
+
B^{l,\mathrm{attn}}_{\mathrm{attn}}
+
B^{l,\mathrm{attn}}_{\mathrm{res1.add}}
+
B^{l,\mathrm{attn}}_{\mathrm{res1.id}}
\bigr)
\\
&\quad+
\sum_{l=0}^{L-1}
J_{l,\mathrm{mlp}}^{\leftarrow}
\bigl(
B^{l,\mathrm{mlp}}_{\mathrm{norm2}}
+
B^{l,\mathrm{mlp}}_{\mathrm{mlp}}
+
B^{l,\mathrm{mlp}}_{\mathrm{res2.add}}
+
B^{l,\mathrm{mlp}}_{\mathrm{res2.id}}
\bigr)
+
S_{\mathrm{fn}}(Y;X)
+
S_T(Y;X)
\\
&=
\sum_{l=0}^{L-1}
\bigl(
S_{l,\mathrm{norm1}}(Y;X)
+
S_{l,\mathrm{attn}}(Y;X)
+
S_{l,\mathrm{res1.add}}(Y;X)
+
S_{l,\mathrm{res1.id}}(Y;X)
\bigr)
\\
&\quad+
\sum_{l=0}^{L-1}
\bigl(
S_{l,\mathrm{norm2}}(Y;X)
+
S_{l,\mathrm{mlp}}(Y;X)
+
S_{l,\mathrm{res2.add}}(Y;X)
+
S_{l,\mathrm{res2.id}}(Y;X)
\bigr)
\\
& \quad
+
S_{\mathrm{fn}}(Y;X)
+
S_T(Y;X)
\\
&=
\left[
\sum_{l=0}^{L-1}
\bigl(
S_{l,\mathrm{norm1}}(Y;X)
+
S_{l,\mathrm{norm2}}(Y;X)
\bigr)
+
S_{\mathrm{fn}}(Y;X)
\right]
\\
&\quad+
\sum_{l=0}^{L-1}S_{l,\mathrm{attn}}(Y;X)
+
\sum_{l=0}^{L-1}S_{l,\mathrm{mlp}}(Y;X)
\\
&\quad+
\sum_{l=0}^{L-1}
\bigl(
S_{l,\mathrm{res1.add}}(Y;X)
+
S_{l,\mathrm{res1.id}}(Y;X)
+
S_{l,\mathrm{res2.add}}(Y;X)
+
S_{l,\mathrm{res2.id}}(Y;X)
\bigr)
\\
& \quad
+
S_T(Y;X)
\\
&=
S_{\mathrm{norm}}^\Sigma(Y;X)
+
S_{\mathrm{attn}}^\Sigma(Y;X)
+
S_{\mathrm{mlp}}^\Sigma(Y;X)
+
S_{\mathrm{res}}^\Sigma(Y;X)
+
S_T(Y;X).
\end{aligned}
\]
Thus, for every fixed \(Y\in\mathcal Y\),
\[
R_Y(X)
=
S_{\mathrm{norm}}^\Sigma(Y;X)
+
S_{\mathrm{attn}}^\Sigma(Y;X)
+
S_{\mathrm{mlp}}^\Sigma(Y;X)
+
S_{\mathrm{res}}^\Sigma(Y;X)
+
S_T(Y;X).
\]

Finally, by the definition of RED,
\[
\mathcal R_{\mathcal Y}(X)
=
\left\{
R_Y(X):Y\in\mathcal Y
\right\}.
\]
Substituting the operator-level decomposition of \(R_Y(X)\) for each \(Y\in\mathcal Y\), we obtain
\[
\mathcal R_{\mathcal Y}(X)
=
\left\{
S_{\mathrm{norm}}^\Sigma(Y;X)
+
S_{\mathrm{attn}}^\Sigma(Y;X)
+
S_{\mathrm{mlp}}^\Sigma(Y;X)
+
S_{\mathrm{res}}^\Sigma(Y;X)
+
S_T(Y;X)
:
Y\in\mathcal Y
\right\},
\]
which is exactly \eqref{eq:operator_level_red_set}. This completes the derivation.

\subsection{Proofs for analytically constrained operator-level sources}
\label{app:rigid_sources}

\subsubsection{Normalization as an analytically constrained operator-level source}

We consider normalization sources, including \(N_l^{\mathrm{attn}}\), \(N_l^{\mathrm{mlp}}\), and \(N^{\mathrm{fn}}\), when instantiated as LayerNorm or RMSNorm \citep{ba2016layer,zhang2019root}. These normalization operators act row-wise on hidden-state matrices. If \(Z\in\mathbb R^{n\times d}\) has token rows \(z_1,\ldots,z_n\), then
\[
N(Z)
=
\begin{bmatrix}
n(z_1)\\
\vdots\\
n(z_n)
\end{bmatrix},
\]
where \(n:\mathbb R^d\to\mathbb R^d\) is the corresponding single-token LayerNorm or RMSNorm map. Thus, \(J_N(Z)\) is block diagonal across token positions, with one row-wise Jacobian block for each token row.

Let \(Z(X)\in\mathbb R^{n\times d}\) denote the hidden-state matrix input to a normalization operator \(N\), and let
\[
Z_+(X):=N(Z(X))
\]
denote its output. Let \(\mathcal U_Z(X)\) be the propagated harmful-semantics-sensitive subspace at the normalization input. The output-side propagated subspace is
\[
J_N(Z(X))\mathcal U_Z(X).
\]
For a fixed target-behavior direction \(Y\in\mathcal Y\), let \(G_+(Y;X)\) be the target-behavior sensitivity direction at the normalization output. The harmful-semantics-aligned output component is
\[
A_+^{N}(Y;X)
:=
P_{J_N(Z(X))\mathcal U_Z(X)}G_+(Y;X).
\]
This is the local instance of the aligned component \(A_{t+1}\) in the generic leakage source
\[
B_t=P_{\mathcal U_t^\perp}\bigl(J_t(X)^*A_{t+1}\bigr).
\]
Thus, for a normalization operator \(N\), the normalization leakage source is
\begin{equation}
B_N(Y;X)
:=
P_{\mathcal U_Z(X)^\perp}
\bigl(
J_N(Z(X))^*A_+^{N}(Y;X)
\bigr).
\label{eq:app_norm_source_def}
\end{equation}

For LayerNorm, write the rows of \(Z(X)\) as \(z_i(X)\). Let
\[
D_\gamma:=\operatorname{Diag}(\gamma),
\qquad
C:=I-\frac1d\mathbf 1\mathbf 1^\top,
\qquad
\sigma_i(X):=\sqrt{\frac1d\|Cz_i(X)\|_2^2+\varepsilon},
\]
and
\[
\bar z_i(X):=Cz_i(X),
\qquad
\mathcal W_i^{\mathrm{LN}}(X)
:=
C-\frac{1}{d\sigma_i(X)^2}\bar z_i(X)\bar z_i(X)^\top .
\]
For any perturbation \(\Delta Z\), the LayerNorm Jacobian acts row-wise as
\begin{equation}
\bigl[J_{\mathrm{LN}}(Z(X))\Delta Z\bigr]_i
=
\frac1{\sigma_i(X)}
D_\gamma
\mathcal W_i^{\mathrm{LN}}(X)
\Delta z_i .
\label{eq:app_ln_jacobian_row}
\end{equation}
Therefore, the output-side subspace induced by LayerNorm is
\begin{equation}
\mathcal V_{\mathrm{LN}}(X)
:=
J_{\mathrm{LN}}(Z(X))\mathcal U_Z(X)
=
\left\{
\Delta Z_+:
[\Delta Z_+]_i
=
\frac1{\sigma_i(X)}
D_\gamma
\mathcal W_i^{\mathrm{LN}}(X)
\Delta z_i,\ 
\Delta Z\in\mathcal U_Z(X)
\right\}.
\label{eq:app_V_LN}
\end{equation}
The LayerNorm output-aligned component is
\[
A_+^{\mathrm{LN}}(Y;X)
:=
P_{\mathcal V_{\mathrm{LN}}(X)}G_+(Y;X).
\]
Pulling this component back through the LayerNorm adjoint gives a matrix \(L_{\mathrm{LN}}(Y;X)\) whose \(i\)-th row is
\begin{equation}
\bigl[L_{\mathrm{LN}}(Y;X)\bigr]_i
=
\frac1{\sigma_i(X)}
\mathcal W_i^{\mathrm{LN}}(X)
D_\gamma
\bigl[A_+^{\mathrm{LN}}(Y;X)\bigr]_i .
\label{eq:app_L_LN_row}
\end{equation}
Equivalently,
\[
L_{\mathrm{LN}}(Y;X)
=
J_{\mathrm{LN}}(Z(X))^*
P_{\mathcal V_{\mathrm{LN}}(X)}G_+(Y;X).
\]
The LayerNorm leakage source is therefore
\begin{equation}
B_{\mathrm{LN}}(Y;X)
=
P_{\mathcal U_Z(X)^\perp}
\bigl(
L_{\mathrm{LN}}(Y;X)
\bigr).
\label{eq:app_B_LN_explicit}
\end{equation}
The LayerNorm consistency condition is
\[
\mathfrak C_{\mathrm{LN}}(Y):
\qquad
L_{\mathrm{LN}}(Y;X)\in\mathcal U_Z(X),
\qquad
\forall X\in\Omega.
\]

For RMSNorm, define
\[
\rho_i(X):=\sqrt{\frac1d\|z_i(X)\|_2^2+\varepsilon},
\qquad
\mathcal W_i^{\mathrm{RMS}}(X)
:=
I-\frac{1}{d\rho_i(X)^2}z_i(X)z_i(X)^\top .
\]
For any perturbation \(\Delta Z\), the RMSNorm Jacobian acts row-wise as
\begin{equation}
\bigl[J_{\mathrm{RMS}}(Z(X))\Delta Z\bigr]_i
=
\frac1{\rho_i(X)}
D_\gamma
\mathcal W_i^{\mathrm{RMS}}(X)
\Delta z_i .
\label{eq:app_rms_jacobian_row}
\end{equation}
Hence the output-side subspace induced by RMSNorm is
\begin{equation}
\mathcal V_{\mathrm{RMS}}(X)
:=
J_{\mathrm{RMS}}(Z(X))\mathcal U_Z(X)
=
\left\{
\Delta Z_+:
[\Delta Z_+]_i
=
\frac1{\rho_i(X)}
D_\gamma
\mathcal W_i^{\mathrm{RMS}}(X)
\Delta z_i,\ 
\Delta Z\in\mathcal U_Z(X)
\right\}.
\label{eq:app_V_RMS}
\end{equation}
The RMSNorm output-aligned component is
\[
A_+^{\mathrm{RMS}}(Y;X)
:=
P_{\mathcal V_{\mathrm{RMS}}(X)}G_+(Y;X).
\]
Pulling this component back through the RMSNorm adjoint gives a matrix \(L_{\mathrm{RMS}}(Y;X)\) whose \(i\)-th row is
\begin{equation}
\bigl[L_{\mathrm{RMS}}(Y;X)\bigr]_i
=
\frac1{\rho_i(X)}
\mathcal W_i^{\mathrm{RMS}}(X)
D_\gamma
\bigl[A_+^{\mathrm{RMS}}(Y;X)\bigr]_i .
\label{eq:app_L_RMS_row}
\end{equation}
Equivalently,
\[
L_{\mathrm{RMS}}(Y;X)
=
J_{\mathrm{RMS}}(Z(X))^*
P_{\mathcal V_{\mathrm{RMS}}(X)}G_+(Y;X).
\]
The RMSNorm leakage source is
\begin{equation}
B_{\mathrm{RMS}}(Y;X)
=
P_{\mathcal U_Z(X)^\perp}
\bigl(
L_{\mathrm{RMS}}(Y;X)
\bigr).
\label{eq:app_B_RMS_explicit}
\end{equation}
The RMSNorm consistency condition is
\[
\mathfrak C_{\mathrm{RMS}}(Y):
\qquad
L_{\mathrm{RMS}}(Y;X)\in\mathcal U_Z(X),
\qquad
\forall X\in\Omega.
\]

\begin{proposition}[Normalization as an analytically constrained operator-level source]
\label{prop:norm_rigid}
Under Assumption~\ref{ass:analytic}, the following hold:
\begin{enumerate}
    \item \(B_{\mathrm{LN}}(Y;X)=0\) for all \(X\in\mathcal X\) if and only if \(\mathfrak C_{\mathrm{LN}}(Y)\) holds on \(\Omega\). Likewise, \(B_{\mathrm{RMS}}(Y;X)=0\) for all \(X\in\mathcal X\) if and only if \(\mathfrak C_{\mathrm{RMS}}(Y)\) holds on \(\Omega\).

    \item If \(\mathfrak C_{\mathrm{LN}}(Y)\) fails on \(\Omega\), then \(B_{\mathrm{LN}}(Y;X)\neq 0\) on a dense open subset of \(\Omega\). The same statement holds for RMSNorm.
\end{enumerate}
\end{proposition}

\begin{proof}
We prove the LayerNorm case; the RMSNorm case is identical.

By \eqref{eq:app_B_LN_explicit},
\[
B_{\mathrm{LN}}(Y;X)
=
P_{\mathcal U_Z(X)^\perp}
\bigl(
L_{\mathrm{LN}}(Y;X)
\bigr).
\]
Therefore, pointwise in \(X\),
\[
\begin{aligned}
B_{\mathrm{LN}}(Y;X)=0
&\iff
P_{\mathcal U_Z(X)^\perp}
\bigl(
L_{\mathrm{LN}}(Y;X)
\bigr)=0 \\
&\iff
L_{\mathrm{LN}}(Y;X)\in\mathcal U_Z(X).
\end{aligned}
\]
This is exactly the pointwise LayerNorm consistency condition.

Suppose \(B_{\mathrm{LN}}(Y;X)=0\) for all \(X\in\mathcal X\). Define
\[
T_{\mathrm{LN}}(Y;X)
:=
P_{\mathcal U_Z(X)^\perp}
\bigl(
L_{\mathrm{LN}}(Y;X)
\bigr).
\]
Then
\[
T_{\mathrm{LN}}(Y;X)=0,
\qquad
\forall X\in\mathcal X.
\]
By Assumption~\ref{ass:analytic}, \(T_{\mathrm{LN}}(Y;\cdot)\) is real-analytic on \(\Omega\). Since \(\mathcal X\subseteq\Omega\) is nonempty and open and \(\Omega\) is connected, the analytic identity theorem gives
\[
T_{\mathrm{LN}}(Y;X)\equiv 0,
\qquad
\forall X\in\Omega.
\]
Hence
\[
L_{\mathrm{LN}}(Y;X)\in\mathcal U_Z(X),
\qquad
\forall X\in\Omega.
\]
Thus \(\mathfrak C_{\mathrm{LN}}(Y)\) holds on \(\Omega\).

Conversely, if \(\mathfrak C_{\mathrm{LN}}(Y)\) holds on \(\Omega\), then
\[
L_{\mathrm{LN}}(Y;X)\in\mathcal U_Z(X),
\qquad
\forall X\in\Omega.
\]
Thus
\[
P_{\mathcal U_Z(X)^\perp}
\bigl(
L_{\mathrm{LN}}(Y;X)
\bigr)=0,
\qquad
\forall X\in\Omega.
\]
Therefore
\[
B_{\mathrm{LN}}(Y;X)=0,
\qquad
\forall X\in\Omega,
\]
and in particular \(B_{\mathrm{LN}}(Y;X)=0\) for all \(X\in\mathcal X\). This proves the first claim.

For the second claim, suppose that \(\mathfrak C_{\mathrm{LN}}(Y)\) fails on \(\Omega\). If the zero set of \(B_{\mathrm{LN}}(Y;\cdot)\) had nonempty interior in \(\Omega\), then there would exist a nonempty open set \(O\subseteq\Omega\) such that
\[
B_{\mathrm{LN}}(Y;X)=0,
\qquad
\forall X\in O.
\]
Equivalently,
\[
T_{\mathrm{LN}}(Y;X)=0,
\qquad
\forall X\in O.
\]
Since \(T_{\mathrm{LN}}(Y;\cdot)\) is real-analytic on the connected region \(\Omega\), the analytic identity theorem gives
\[
T_{\mathrm{LN}}(Y;X)\equiv 0,
\qquad
\forall X\in\Omega.
\]
Therefore
\[
L_{\mathrm{LN}}(Y;X)\in\mathcal U_Z(X),
\qquad
\forall X\in\Omega,
\]
so \(\mathfrak C_{\mathrm{LN}}(Y)\) holds on \(\Omega\), contradicting the assumption. Hence the zero set of \(B_{\mathrm{LN}}(Y;\cdot)\) has empty interior in \(\Omega\). Since \(B_{\mathrm{LN}}(Y;\cdot)\) is continuous, its complement is open and dense. Thus \(B_{\mathrm{LN}}(Y;X)\neq0\) on a dense open subset of \(\Omega\).

The RMSNorm proof is identical after replacing \(L_{\mathrm{LN}}\), \(T_{\mathrm{LN}}\), and \(B_{\mathrm{LN}}\) with \(L_{\mathrm{RMS}}\), \(T_{\mathrm{RMS}}\), and \(B_{\mathrm{RMS}}\), respectively.
\end{proof}

The conditions \(\mathfrak C_{\mathrm{LN}}(Y)\) and \(\mathfrak C_{\mathrm{RMS}}(Y)\) are normalization-consistency conditions. They require a pushforward--pullback consistency between the input-side subspace \(\mathcal U_Z(X)\) and the output-side target-behavior sensitivity \(G_+(Y;X)\). For example, in RMSNorm, every input-side direction \(\Delta Z\in\mathcal U_Z(X)\), with token rows \(\Delta z_i\), is pushed forward row-wise by
\[
\Delta z_i
\mapsto
\frac1{\rho_i(X)}
D_\gamma
\mathcal W_i^{\mathrm{RMS}}(X)
\Delta z_i
\]
to form the output-side subspace \(\mathcal V_{\mathrm{RMS}}(X)\). The component of \(G_+(Y;X)\) aligned with this output-side subspace,
\[
A_+^{\mathrm{RMS}}(Y;X)
=
P_{\mathcal V_{\mathrm{RMS}}(X)}G_+(Y;X),
\]
is then pulled back row-wise by
\[
a_i
\mapsto
\frac1{\rho_i(X)}
\mathcal W_i^{\mathrm{RMS}}(X)
D_\gamma
a_i,
\]
where \(a_i\) denotes the \(i\)-th row of \(A_+^{\mathrm{RMS}}(Y;X)\). Eliminating the RMSNorm leakage source requires the resulting pulled-back matrix \(L_{\mathrm{RMS}}(Y;X)\) to lie inside the original input-side subspace \(\mathcal U_Z(X)\). The LayerNorm condition is analogous, with \(\rho_i(X)\) and \(\mathcal W_i^{\mathrm{RMS}}(X)\) replaced by \(\sigma_i(X)\) and \(\mathcal W_i^{\mathrm{LN}}(X)\). Whether this condition holds is determined by the input-side geometry \((Z(X),\mathcal U_Z(X))\) and the output-side sensitivity geometry \(G_+(Y;X)\).

\subsubsection{Residual wiring as an analytically constrained operator-level source}

We next consider a generic post-add residual block
\[
\mathcal B(H)=H+F(H),
\]
where \(H\) is the residual-stream input, \(F(H)\) is the body-branch output, and \(\mathcal B(H)\) is the residual-block output. In the attention block of Appendix~\ref{app:operators}, this generic notation corresponds to
\[
H=X_l,
\qquad
F=F^{l,\mathrm{attn}}=\mathrm{Attn}_l\circ N_l^{\mathrm{attn}},
\qquad
\mathcal B(H)=Z_l.
\]
In the MLP block, it corresponds to
\[
H=Z_l,
\qquad
F=F^{l,\mathrm{mlp}}=\mathrm{MLP}_l\circ N_l^{\mathrm{mlp}},
\qquad
\mathcal B(H)=X_{l+1}.
\]

Let \(H(X)\) denote the residual-block input induced by the global input \(X\), and let \(\mathcal U_{\mathrm{in}}(X)\) be the propagated harmful-semantics-sensitive subspace at this input. The body-branch subspace is
\[
\mathcal U_{\mathrm{body}}(X)
:=
J_F(H(X))\mathcal U_{\mathrm{in}}(X).
\]
Since
\[
J_{\mathcal B}(H(X))
=
I+J_F(H(X)),
\]
the output-side propagated subspace is
\[
\mathcal U_{\mathrm{out}}(X)
:=
J_{\mathcal B}(H(X))\mathcal U_{\mathrm{in}}(X)
=
\bigl(I+J_F(H(X))\bigr)\mathcal U_{\mathrm{in}}(X).
\]
Equivalently,
\[
\mathcal U_{\mathrm{out}}(X)
=
\left\{
U+J_F(H(X))U:
U\in\mathcal U_{\mathrm{in}}(X)
\right\}.
\]
Thus, the residual-output subspace is induced by the paired sum of the identity-branch direction \(U\) and the corresponding body-branch direction \(J_F(H(X))U\).

For a fixed target-behavior direction \(Y\in\mathcal Y\), let \(G_{\mathrm{out}}(Y;X)\) be the target-behavior sensitivity direction at the residual-block output. Its harmful-semantics-aligned output component is
\[
A_{\mathrm{out}}(Y;X)
:=
P_{\mathcal U_{\mathrm{out}}(X)}G_{\mathrm{out}}(Y;X).
\]
This is the output-side aligned component obtained by projecting \(G_{\mathrm{out}}(Y;X)\) onto the residual-output subspace \(\mathcal U_{\mathrm{out}}(X)\).

Residual wiring produces two raw branch-consistency requirements. The identity edge requires the output-aligned component to lie in the residual-stream input subspace, and the body-add edge requires the same output-aligned component to lie in the body-branch subspace. Accordingly, define
\[
B_{\mathrm{id}}(Y;X)
:=
P_{\mathcal U_{\mathrm{in}}(X)^\perp}
\bigl(A_{\mathrm{out}}(Y;X)\bigr),
\]
and
\[
B_{\mathrm{add}}(Y;X)
:=
P_{\mathcal U_{\mathrm{body}}(X)^\perp}
\bigl(A_{\mathrm{out}}(Y;X)\bigr).
\]
We collect these two residual-wiring sources as
\[
B_{\mathrm{res}}(Y;X)
:=
\bigl(B_{\mathrm{id}}(Y;X),B_{\mathrm{add}}(Y;X)\bigr).
\]

In the attention block, these two raw residual-wiring terms correspond to
\[
P_{\mathcal U_{X_l}^\perp}(A_{Z_l})
\qquad\text{and}\qquad
P_{\mathcal U_{O_l}^\perp}(A_{Z_l}),
\]
which, after the appropriate branch adjoint transports, give the residual-wiring contributions in \eqref{eq:app_B_res1_id} and \eqref{eq:app_B_res1_add}. In the MLP block, they correspond to
\[
P_{\mathcal U_{Z_l}^\perp}(A_{X_{l+1}})
\qquad\text{and}\qquad
P_{\mathcal U_{M_l}^\perp}(A_{X_{l+1}}),
\]
which, after the appropriate branch adjoint transports, give the residual-wiring contributions in \eqref{eq:app_B_res2_id} and \eqref{eq:app_B_res2_add}.

The corresponding residual-branch-consistency condition is
\[
\mathfrak C_{\mathrm{res}}(Y):
\qquad
A_{\mathrm{out}}(Y;X)
\in
\mathcal U_{\mathrm{in}}(X)\cap\mathcal U_{\mathrm{body}}(X),
\qquad
\forall X\in\Omega.
\]

\begin{proposition}[Residual wiring as an analytically constrained operator-level source]
\label{prop:residual_rigid}
Under Assumption~\ref{ass:analytic}, the following hold:
\begin{enumerate}
    \item \(B_{\mathrm{res}}(Y;X)=0\) for all \(X\in\mathcal X\) if and only if \(\mathfrak C_{\mathrm{res}}(Y)\) holds on \(\Omega\).

    \item If \(\mathfrak C_{\mathrm{res}}(Y)\) fails on \(\Omega\), then \(B_{\mathrm{res}}(Y;X)\neq 0\) on a dense open subset of \(\Omega\). Equivalently, on a dense open subset of \(\Omega\), at least one of \(B_{\mathrm{id}}(Y;X)\) or \(B_{\mathrm{add}}(Y;X)\) is nonzero.
\end{enumerate}
\end{proposition}

\begin{proof}
First, pointwise in \(X\), we have
\[
\begin{aligned}
B_{\mathrm{id}}(Y;X)=0
&\iff
P_{\mathcal U_{\mathrm{in}}(X)^\perp}
\bigl(A_{\mathrm{out}}(Y;X)\bigr)=0 \\
&\iff
A_{\mathrm{out}}(Y;X)\in\mathcal U_{\mathrm{in}}(X),
\end{aligned}
\]
and
\[
\begin{aligned}
B_{\mathrm{add}}(Y;X)=0
&\iff
P_{\mathcal U_{\mathrm{body}}(X)^\perp}
\bigl(A_{\mathrm{out}}(Y;X)\bigr)=0 \\
&\iff
A_{\mathrm{out}}(Y;X)\in\mathcal U_{\mathrm{body}}(X).
\end{aligned}
\]
Therefore,
\[
\begin{aligned}
B_{\mathrm{res}}(Y;X)=0
&\iff
B_{\mathrm{id}}(Y;X)=0
\ \text{and}\
B_{\mathrm{add}}(Y;X)=0 \\
&\iff
A_{\mathrm{out}}(Y;X)\in\mathcal U_{\mathrm{in}}(X)
\ \text{and}\
A_{\mathrm{out}}(Y;X)\in\mathcal U_{\mathrm{body}}(X) \\
&\iff
A_{\mathrm{out}}(Y;X)
\in
\mathcal U_{\mathrm{in}}(X)\cap\mathcal U_{\mathrm{body}}(X).
\end{aligned}
\]

Suppose \(B_{\mathrm{res}}(Y;X)=0\) for all \(X\in\mathcal X\). Define
\[
T_{\mathrm{res}}(Y;X)
:=
\left(
P_{\mathcal U_{\mathrm{in}}(X)^\perp}
\bigl(A_{\mathrm{out}}(Y;X)\bigr),
\,
P_{\mathcal U_{\mathrm{body}}(X)^\perp}
\bigl(A_{\mathrm{out}}(Y;X)\bigr)
\right).
\]
Then
\[
T_{\mathrm{res}}(Y;X)=0,
\qquad
\forall X\in\mathcal X.
\]
By Assumption~\ref{ass:analytic}, \(T_{\mathrm{res}}(Y;\cdot)\) is real-analytic on \(\Omega\). Since \(\mathcal X\subseteq\Omega\) is nonempty and open and \(\Omega\) is connected, the analytic identity theorem gives
\[
T_{\mathrm{res}}(Y;X)\equiv 0,
\qquad
\forall X\in\Omega.
\]
Thus,
\[
P_{\mathcal U_{\mathrm{in}}(X)^\perp}
\bigl(A_{\mathrm{out}}(Y;X)\bigr)=0,
\qquad
P_{\mathcal U_{\mathrm{body}}(X)^\perp}
\bigl(A_{\mathrm{out}}(Y;X)\bigr)=0,
\qquad
\forall X\in\Omega.
\]
Equivalently,
\[
A_{\mathrm{out}}(Y;X)
\in
\mathcal U_{\mathrm{in}}(X)\cap\mathcal U_{\mathrm{body}}(X),
\qquad
\forall X\in\Omega.
\]
Hence \(\mathfrak C_{\mathrm{res}}(Y)\) holds on \(\Omega\).

Conversely, if \(\mathfrak C_{\mathrm{res}}(Y)\) holds on \(\Omega\), then
\[
A_{\mathrm{out}}(Y;X)
\in
\mathcal U_{\mathrm{in}}(X)\cap\mathcal U_{\mathrm{body}}(X),
\qquad
\forall X\in\Omega.
\]
Therefore,
\[
P_{\mathcal U_{\mathrm{in}}(X)^\perp}
\bigl(A_{\mathrm{out}}(Y;X)\bigr)=0,
\qquad
P_{\mathcal U_{\mathrm{body}}(X)^\perp}
\bigl(A_{\mathrm{out}}(Y;X)\bigr)=0,
\qquad
\forall X\in\Omega.
\]
Thus,
\[
B_{\mathrm{res}}(Y;X)=0,
\qquad
\forall X\in\Omega,
\]
and in particular \(B_{\mathrm{res}}(Y;X)=0\) for all \(X\in\mathcal X\). This proves the first claim.

For the second claim, suppose that \(\mathfrak C_{\mathrm{res}}(Y)\) fails on \(\Omega\). If the zero set of \(B_{\mathrm{res}}(Y;\cdot)\) had nonempty interior in \(\Omega\), then there would exist a nonempty open set \(O\subseteq\Omega\) such that
\[
B_{\mathrm{res}}(Y;X)=0,
\qquad
\forall X\in O.
\]
Equivalently,
\[
T_{\mathrm{res}}(Y;X)=0,
\qquad
\forall X\in O.
\]
Since \(T_{\mathrm{res}}(Y;\cdot)\) is real-analytic on the connected open set \(\Omega\), the analytic identity theorem gives
\[
T_{\mathrm{res}}(Y;X)\equiv 0,
\qquad
\forall X\in\Omega.
\]
By the pointwise equivalence above, this implies that
\[
A_{\mathrm{out}}(Y;X)
\in
\mathcal U_{\mathrm{in}}(X)\cap\mathcal U_{\mathrm{body}}(X),
\qquad
\forall X\in\Omega,
\]
so \(\mathfrak C_{\mathrm{res}}(Y)\) holds on \(\Omega\), contradicting the assumption. Hence the zero set of \(B_{\mathrm{res}}(Y;\cdot)\) has empty interior in \(\Omega\). Since \(B_{\mathrm{res}}(Y;\cdot)\) is continuous, its complement is open and dense. Therefore \(B_{\mathrm{res}}(Y;X)\neq0\) on a dense open subset of \(\Omega\). Equivalently, at least one of \(B_{\mathrm{id}}(Y;X)\) or \(B_{\mathrm{add}}(Y;X)\) is nonzero on that dense open subset.
\end{proof}

The condition \(\mathfrak C_{\mathrm{res}}(Y)\) is a residual-branch-consistency condition. The output-side aligned component
\[
A_{\mathrm{out}}(Y;X)
=
P_{\mathcal U_{\mathrm{out}}(X)}G_{\mathrm{out}}(Y;X)
\]
is first obtained by projecting the output-side target-behavior sensitivity onto the residual-output subspace induced by the paired sum of the identity branch and the body branch. Eliminating residual-wiring leakage requires this same output-aligned component to lie in both branch subspaces: it must be explainable as an identity-branch direction in \(\mathcal U_{\mathrm{in}}(X)\) and also as a body-branch direction in \(\mathcal U_{\mathrm{body}}(X)\). Whether this condition holds is determined by the branch-side geometry \(\mathcal U_{\mathrm{in}}(X)\) and \(\mathcal U_{\mathrm{body}}(X)\), together with the output-side sensitivity geometry \(G_{\mathrm{out}}(Y;X)\).

\subsubsection{Terminal source as an analytically constrained operator-level source}

Finally, we consider the terminal source. Recall that the target-behavior subspace
\(\mathcal Y\) is defined in the terminal hidden-state space. For a fixed
target-behavior direction \(Y\in\mathcal Y\), the terminal target-behavior
sensitivity direction is simply
\[
G_T(Y;X):=Y.
\]
Let \(\mathcal U_T(X)\) denote the propagated harmful-semantics-sensitive
subspace at the terminal hidden-state space. The terminal source
\[
R_T(Y;X)
:=
P_{\mathcal U_T(X)^\perp}Y.
\]
This is the local instance of the terminal source in the generic decomposition.

The corresponding target-subspace-consistency condition is
\[
\mathfrak C_T(Y):
\qquad
Y\in\mathcal U_T(X),
\qquad
\forall X\in\Omega.
\]

\begin{proposition}[Terminal source as an analytically constrained operator-level source]
\label{prop:terminal_rigid}
Under Assumption~\ref{ass:analytic}, the following hold:
\begin{enumerate}
    \item \(R_T(Y;X)=0\) for all \(X\in\mathcal X\) if and only if \(\mathfrak C_T(Y)\) holds on \(\Omega\).

    \item If \(\mathfrak C_T(Y)\) fails on \(\Omega\), then \(R_T(Y;X)\neq0\) on a dense open subset of \(\Omega\).
\end{enumerate}
\end{proposition}

\begin{proof}
Pointwise in \(X\),
\[
\begin{aligned}
R_T(Y;X)=0
&\iff
P_{\mathcal U_T(X)^\perp}Y=0 \\
&\iff
Y\in\mathcal U_T(X).
\end{aligned}
\]
Define
\[
T_T(Y;X)
:=
P_{\mathcal U_T(X)^\perp}Y.
\]
Then
\[
R_T(Y;X)=T_T(Y;X).
\]

Suppose \(R_T(Y;X)=0\) for all \(X\in\mathcal X\). Then
\[
T_T(Y;X)=0,
\qquad
\forall X\in\mathcal X.
\]
By Assumption~\ref{ass:analytic}, \(T_T(Y;\cdot)\) is real-analytic on \(\Omega\). Since \(\mathcal X\subseteq\Omega\) is nonempty and open and \(\Omega\) is connected, the analytic identity theorem gives
\[
T_T(Y;X)\equiv 0,
\qquad
\forall X\in\Omega.
\]
Therefore,
\[
P_{\mathcal U_T(X)^\perp}Y=0,
\qquad
\forall X\in\Omega.
\]
Equivalently,
\[
Y\in\mathcal U_T(X),
\qquad
\forall X\in\Omega.
\]
Hence \(\mathfrak C_T(Y)\) holds on \(\Omega\).

Conversely, if \(\mathfrak C_T(Y)\) holds on \(\Omega\), then
\[
Y\in\mathcal U_T(X),
\qquad
\forall X\in\Omega.
\]
Thus
\[
P_{\mathcal U_T(X)^\perp}Y=0,
\qquad
\forall X\in\Omega.
\]
Therefore,
\[
R_T(Y;X)=0,
\qquad
\forall X\in\Omega,
\]
and in particular for all \(X\in\mathcal X\). This proves the first claim.

For the second claim, suppose that \(\mathfrak C_T(Y)\) fails on \(\Omega\). If the zero set of \(R_T(Y;\cdot)\) had nonempty interior in \(\Omega\), then there would exist a nonempty open set \(O\subseteq\Omega\) such that
\[
R_T(Y;X)=0,
\qquad
\forall X\in O.
\]
Equivalently,
\[
T_T(Y;X)=0,
\qquad
\forall X\in O.
\]
Since \(T_T(Y;\cdot)\) is real-analytic on the connected open set \(\Omega\), the analytic identity theorem gives
\[
T_T(Y;X)\equiv 0,
\qquad
\forall X\in\Omega.
\]
By the pointwise equivalence above, this implies that \(Y\in\mathcal U_T(X)\) for all \(X\in\Omega\), so \(\mathfrak C_T(Y)\) holds on \(\Omega\), contradicting the assumption. Hence the zero set of \(R_T(Y;\cdot)\) has empty interior in \(\Omega\). Since \(R_T(Y;\cdot)\) is continuous, its complement is open and dense. Therefore \(R_T(Y;X)\neq0\) on a dense open subset of \(\Omega\).
\end{proof}

The condition \(\mathfrak C_T(Y)\) is a target-subspace-consistency condition. It requires the terminal target-behavior direction \(Y\) to be contained in the terminal harmful-semantics-sensitive subspace \(\mathcal U_T(X)\). Whether this condition holds is determined by the terminal subspace geometry \(\mathcal U_T(X)\).

In summary, normalization, residual wiring, and the terminal source satisfy the defining properties of analytically constrained operator-level sources: eliminating each source on \(\mathcal X\) is equivalent to an analytic vanishing condition holding on \(\Omega\), and if that condition fails, the corresponding source is nonzero on a dense open subset of \(\Omega\).

\subsection{Proof of Theorem~\ref{thm:tradeoff}}
\label{app:proof_tradeoff}

\begin{proof}
Suppose the non-identical-field condition in the theorem holds:
\[
\mathcal E_h^\Sigma\not\equiv\mathcal E_b^\Sigma
\qquad
\text{on }\Omega.
\]
Assume, for contradiction, that there exists a single parameter setting
\(\Theta\) satisfying both exact requirements:
\[
\mathcal E_\Theta^\Sigma(X)=\mathcal E_h^\Sigma(X),
\qquad
\forall X\in\mathcal X,
\]
and
\[
\mathcal E_\Theta^\Sigma(X)=\mathcal E_b^\Sigma(X),
\qquad
\forall X\in\mathcal X_b.
\]

We first extend the harmful-region equality from \(\mathcal X\) to the whole connected region \(\Omega\). Define
\[
H(X)
:=
\mathcal E_\Theta^\Sigma(X)-\mathcal E_h^\Sigma(X).
\]
By Assumption~\ref{ass:analytic}, both \(\mathcal E_\Theta^\Sigma\) and \(\mathcal E_h^\Sigma\) are real-analytic on \(\Omega\), so \(H\) is real-analytic on \(\Omega\). The harmful-region constraint gives
\[
H(X)=0,
\qquad
\forall X\in\mathcal X.
\]
Since \(\mathcal X\subseteq\Omega\) is a nonempty open set and \(\Omega\) is connected, the analytic identity theorem implies
\[
H(X)\equiv 0,
\qquad
\forall X\in\Omega.
\]
Therefore,
\[
\mathcal E_\Theta^\Sigma(X)=\mathcal E_h^\Sigma(X),
\qquad
\forall X\in\Omega.
\]

Restricting this identity to the benign region \(\mathcal X_b\subseteq\Omega\), we obtain
\[
\mathcal E_\Theta^\Sigma(X)=\mathcal E_h^\Sigma(X),
\qquad
\forall X\in\mathcal X_b.
\]
But the benign-region behavioral constraint also gives
\[
\mathcal E_\Theta^\Sigma(X)=\mathcal E_b^\Sigma(X),
\qquad
\forall X\in\mathcal X_b.
\]
Hence
\[
\mathcal E_h^\Sigma(X)=\mathcal E_b^\Sigma(X),
\qquad
\forall X\in\mathcal X_b.
\]

Now define
\[
K(X)
:=
\mathcal E_h^\Sigma(X)-\mathcal E_b^\Sigma(X).
\]
By Assumption~\ref{ass:analytic}, both \(\mathcal E_h^\Sigma\) and \(\mathcal E_b^\Sigma\) are real-analytic on \(\Omega\), so \(K\) is real-analytic on \(\Omega\). The equality above gives
\[
K(X)=0,
\qquad
\forall X\in\mathcal X_b.
\]
Since \(\mathcal X_b\subseteq\Omega\) is a nonempty open set and \(\Omega\) is connected, the analytic identity theorem implies
\[
K(X)\equiv 0,
\qquad
\forall X\in\Omega.
\]
Thus,
\[
\mathcal E_h^\Sigma\equiv\mathcal E_b^\Sigma
\qquad
\text{on }\Omega.
\]
This contradicts the non-identical-field condition of the theorem:
\[
\mathcal E_h^\Sigma\not\equiv\mathcal E_b^\Sigma
\qquad
\text{on }\Omega.
\]

Therefore, under the non-identical-field condition, no single parameter setting
\(\Theta\) can satisfy both exact requirements.
\end{proof}

\section{Experiments}

\subsection{Additional details and justifications for the experimental settings}
\label{app:settings}

\subsubsection{Model-specific decomposition details}

We analyze five aligned pre-norm residual LLMs: Qwen3-4B and Qwen3-14B \citep{yang2025qwen3}, Llama-3.1-8B-Instruct \citep{grattafiori2024llama3}, and Gemma-3-4B-IT and Gemma-3-12B-IT \citep{gemmateam2025gemma3}. For Qwen3 and Llama-3.1-8B-Instruct, the architecture follows the standard pre-norm residual LLM template assumed in Appendix~\ref{app:operators}, so the exact operator-level decomposition of \(R_Y(X)\) derived there applies directly to the computation of the operator-level contributions for \(R_{\mathrm{ref}}(X)\), by taking \(Y=Y_{\mathrm{ref}}\).

Gemma-3-4B-IT and Gemma-3-12B-IT also follow a decoder-only pre-norm residual transformer architecture. Relative to the standard model template assumed in Appendix~\ref{app:operators}, the only difference relevant to our decomposition is that each layer contains two additional normalization modules: an attention post-layernorm and an MLP post-layernorm. Their leakage sources are computed using the same decomposition procedure as in Appendix~\ref{app:operators}, and their transported contributions are included in the normalization-family sum \(S_{\mathrm{norm}}^\Sigma(Y_{\mathrm{ref}};X)\).

\subsubsection{Jailbreak sample collection}

We analyze five jailbreak attacks: GCG \citep{zou2023universal}, AutoDAN \citep{liu2023autodan}, GPTFuzzer \citep{yu2023gptfuzzer}, TAP \citep{mehrotra2024tree}, and ReNeLLM \citep{ding2024renellm}. These attacks provide a reasonably broad view of jailbreak diversity. In our implementation, GCG, AutoDAN, and GPTFuzzer retain the original harmful prompt and wrap it with additional adversarial context, whereas TAP and ReNeLLM modify the original harmful prompt while preserving its underlying harmful semantics.

To increase the diversity of harmful prompts, we collect them from multiple sources, including AdvBench \citep{zou2023universal}, StrongREJECT \citep{souly2024strongreject}, MaliciousInstruct \citep{huang2024catastrophic}, HarmBench \citep{mazeika2024harmbench}, and JBB-Behaviors \citep{chao2024jailbreakbench}. From these sources, we sample prompts to construct jailbreak examples. For each model and each attack method, we collect 100 harmful--jailbreak prompt pairs, where the harmful prompt is refused and the jailbreak prompt is answered.

\subsubsection{Jailbreak evaluation}

We use Qwen3Guard-Gen-8B \citep{zhao2025qwen3guard} as the judge model for determining whether a target-model output constitutes a successful jailbreak. For each test case, we feed the jailbreak prompt together with the target model's response into Qwen3Guard-Gen-8B and parse its output for the field ``Refusal: (Yes|No)''. We treat ``Refusal: No'' as a successful jailbreak and ``Refusal: Yes'' as a failed one. This protocol is consistent with the intended use of Qwen3Guard-Gen-8B as a generative safety-guardrail model.

\subsubsection{Construction and justification of the reference input transformation}

To study real discrete jailbreak prompts under the continuous input-transformation view, we construct a reference input transformation from a harmful prompt to its jailbreak counterpart. We first align the token-embedding matrices of the two prompts into a common input space of matched dimension. This produces two boundary points \(X(\eta_0)\) and \(X(\eta_\star)\), corresponding to the harmful prompt and the jailbreak prompt after alignment.

The alignment is designed to preserve shared harmful semantic fragments as much as possible, so that the reference input transformation minimally changes the underlying harmful semantics of the input. This reflects the structure of jailbreak attacks, which often alter the prompt form while preserving the harmful request. We use a recursive alignment strategy that first matches the longest common contiguous token blocks. For attacks that retain the original harmful prompt, this naturally preserves the harmful fragments shared by the harmful prompt and the jailbreak prompt. For attacks that modify the original harmful prompt while preserving its harmful semantics, the same recursive procedure preserves shared harmful fragments whenever exact overlap remains. When no identical shared token block is available, we align the token pair with the most similar embeddings, measured by maximum embedding inner product. In this way, the alignment preserves exact harmful fragments when available and semantically similar fragments when exact matches are absent.

To avoid introducing additional lexical semantics during alignment, any remaining unmatched spans are padded with all-zero token embeddings, which serve as placeholders without explicit lexical meaning. This yields aligned boundary points \(X(\eta_0)\) and \(X(\eta_\star)\) in the same input space.

We then take the linear interpolation
\[
X(\eta)=X(\eta_0)+\frac{\eta-\eta_0}{\eta_\star-\eta_0}\bigl(X(\eta_\star)-X(\eta_0)\bigr)
\]
as the reference input transformation. Since the harmful fragments in the two boundary points have already been aligned as much as possible, this interpolation provides a controlled embedding-space transformation that approximately preserves shared harmful semantics while connecting the harmful prompt to its jailbreak counterpart.

The alignment algorithm is summarized in Algorithm~\ref{alg:alignment}.

\begin{algorithm}[t]
\caption{Recursive token alignment by exact match and embedding similarity}
\label{alg:alignment}
\begin{algorithmic}[1]
\Require Token-id sequences \(A\), \(B\); token-embedding map \(E(\cdot)\); placeholder symbol \(\mathtt{gap}\) with zero embedding
\Ensure Aligned token sequences \(\widehat A,\widehat B\)

\State Initialize \(\widehat A,\widehat B \gets [\,]\)

\Function{AlignSpan}{$A[a_\ell{:}a_r),\, B[b_\ell{:}b_r)$}
    \If{\(a_\ell \ge a_r\) and \(b_\ell \ge b_r\)} \Return \EndIf
    \If{\(a_\ell \ge a_r\)}
        \State append placeholders of length \(b_r-b_\ell\) to \(\widehat A\)
        \State append \(B[b_\ell{:}b_r)\) to \(\widehat B\); \Return
    \EndIf
    \If{\(b_\ell \ge b_r\)}
        \State append \(A[a_\ell{:}a_r)\) to \(\widehat A\)
        \State append placeholders of length \(a_r-a_\ell\) to \(\widehat B\); \Return
    \EndIf

    \State find the longest common contiguous token block between \(A[a_\ell{:}a_r)\) and \(B[b_\ell{:}b_r)\), denoted by \((s_A,s_B,L)\)
    \If{\(L>0\)}
        \State \Call{AlignSpan}{$A[a_\ell{:}s_A),\,B[b_\ell{:}s_B)$}
        \State append the matched block \(A[s_A{:}s_A+L)=B[s_B{:}s_B+L)\) to both \(\widehat A,\widehat B\)
        \State \Call{AlignSpan}{$A[s_A+L{:}a_r),\,B[s_B+L{:}b_r)$}
        \State \Return
    \EndIf

    \State find the token pair \((i^\star,j^\star)\) with \(a_\ell \le i^\star < a_r\), \(b_\ell \le j^\star < b_r\) maximizing
    \[
    \langle E(A[i]),\,E(B[j])\rangle
    \]
    \State \Call{AlignSpan}{$A[a_\ell{:}i^\star),\,B[b_\ell{:}j^\star)$}
    \State append \(A[i^\star]\) to \(\widehat A\) and \(B[j^\star]\) to \(\widehat B\)
    \State \Call{AlignSpan}{$A[i^\star+1{:}a_r),\,B[j^\star+1{:}b_r)$}
\EndFunction

\State \Call{AlignSpan}{$A[0{:}|A|),\,B[0{:}|B|)$}
\State \Return \((\widehat A,\widehat B)\)
\end{algorithmic}
\end{algorithm}

\subsubsection{Construction and justification of the reference target-behavior subspace}

We use the final-token logit-difference direction between the answer-inducing jailbreak prompt and the refusal-inducing harmful prompt as a reference proxy for the model's answer-versus-refusal behavior. For a given harmful prompt and its jailbreak counterpart, let \(\ell_r\in\mathbb R^{|\mathcal V|}\) denote the final-token logits of the harmful prompt, and let \(\ell_a\in\mathbb R^{|\mathcal V|}\) denote the final-token logits of the jailbreak prompt. We define the logit-space contrast
\[
v:=\ell_a-\ell_r .
\]
This contrast points from the refusal-inducing harmful prompt toward the answer-inducing jailbreak prompt in logit space.

The target-behavior subspace in our framework is defined in the model's hidden-state space. We therefore pull the logit-space contrast back through the output head. Let the final-token output head be
\[
\ell=W_{\mathrm{out}}h+b ,
\]
where \(h\in\mathbb R^{d_T}\) is the final-token hidden state. We define the corresponding hidden-state direction
\[
y:=W_{\mathrm{out}}^\top v .
\]
We then construct a hidden-state matrix \(Y_{\mathrm{raw}}\in\mathbb R^{n_T\times d_T}\) by placing \(y\) at the final-token position and setting all other token positions to zero. The Frobenius-normalized target-behavior direction is
\[
Y_{\mathrm{ref}}
:=
\frac{Y_{\mathrm{raw}}}{\normsmall{Y_{\mathrm{raw}}}} .
\]
The reference target-behavior subspace is the one-dimensional subspace
\[
\mathcal Y_{\mathrm{ref}}
:=
\operatorname{span}\{Y_{\mathrm{ref}}\}.
\]
By construction, the positive direction of \(Y_{\mathrm{ref}}\) points away from refusal and toward answering for the paired harmful--jailbreak prompts.

\subsubsection{Construction and justification of the reference harmful-semantics-sensitive subspace}

We next specify the reference harmful-semantics-sensitive subspace used in the experiments. Under Assumption~\ref{ass:u_star}, the theoretical object is a local subspace field \(X\mapsto\mathcal U(X)\) whose directions affect the model's harmful-semantics interpretation. In the experiments, we instantiate this object with a pair-specific one-dimensional reference subspace, chosen to make the RED analysis operational for each harmful--jailbreak pair.

The construction is motivated by the relation between harmful-semantics interpretation and refusal behavior. Prior work suggests that pretrained models can already recognize harmfulness, and that safety alignment associates this harmfulness recognition with refusal behavior \citep{zhou2024alignment}. Around the original refusal-inducing harmful input, we therefore use the assumption that the model's answer-versus-refusal behavior is mainly governed by its harmful-semantics interpretation. Under this assumption, the input-side target-behavior sensitivity at the clean harmful input provides a natural reference proxy for harmful-semantics-sensitive variation.

Let \(\widehat X\) denote the original harmful input before placeholder padding. We compute the input-side target-behavior sensitivity subspace
\[
J_{0:T}(\widehat X)^*\mathcal Y_{\mathrm{ref}}
\]
at this clean, placeholder-free harmful input, and denote it by
\[
\widehat{\mathcal U}
:=
J_{0:T}(\widehat X)^*\mathcal Y_{\mathrm{ref}}.
\]

To map this reference subspace into the aligned input space used by the reference input transformation, we pad its basis direction with zeros according to the same alignment pattern as the harmful input. The resulting one-dimensional subspace is denoted by \(\mathcal U_{\mathrm{ref}}\), and is used as the fixed reference harmful-semantics-sensitive subspace along the reference input transformation.

\subsubsection{Intermediate-point extraction algorithm}
\label{app:attack_sample}

To analyze the refusal-to-answer transition along the reference input transformation, we do not sample intermediate points uniformly in progress. Instead, for each harmful--jailbreak pair, we extract representative points from the part of the transformation where the reference target-behavior signal is actively changing. This focuses the analysis on the effective refusal-to-answer transition, rather than on inactive prefixes or already saturated suffixes of the transformation.

Since \(\mathcal Y_{\mathrm{ref}}\) is one-dimensional, we use
\[
\normsmall{\Phi_{\mathcal Y_{\mathrm{ref}}}(X(\eta))}
\]
to track movement along the reference target-behavior direction. Let \(\{\bar\eta_i\}_{i=0}^{N_{\mathrm{grid}}-1}\) be a fixed grid on \([\eta_0,\eta_\star]\), and evaluate \(\normsmall{\Phi_{\mathcal Y_{\mathrm{ref}}}(X(\bar\eta_i))}\) on this grid. We first identify the active increase window. Let \(\eta_5\) be the earliest grid point at which \(\normsmall{\Phi_{\mathcal Y_{\mathrm{ref}}}(X(\eta))}\) reaches its final value \(\normsmall{\Phi_{\mathcal Y_{\mathrm{ref}}}(X(\eta_\star))}\), and let \(\eta_1\) be the latest grid point up to \(\eta_5\) at which it is still no larger than its initial value \(\normsmall{\Phi_{\mathcal Y_{\mathrm{ref}}}(X(\eta_0))}\). Thus, \([\eta_1,\eta_5]\) localizes the segment of the reference input transformation over which the reference target-behavior signal effectively rises.

Within this interval, we refine the sampling according to signal activity. For any subinterval \(I=[\eta_L,\eta_R]\), define its activity by
\[
\mathcal A(I)
:=
\frac{
\left|
\normsmall{\Phi_{\mathcal Y_{\mathrm{ref}}}(X(\eta_R))}
-
\normsmall{\Phi_{\mathcal Y_{\mathrm{ref}}}(X(\eta_L))}
\right|
}{
\left|
\normsmall{\Phi_{\mathcal Y_{\mathrm{ref}}}(X(\eta_\star))}
-
\normsmall{\Phi_{\mathcal Y_{\mathrm{ref}}}(X(\eta_0))}
\right|
+\varepsilon
},
\]
where \(\varepsilon>0\) is a small constant for numerical stability. Starting from the active interval \([\eta_1,\eta_5]\), we repeatedly split the subinterval with the largest activity, producing a set of leaf intervals that allocates more resolution to regions where the reference target-behavior signal changes more strongly.

We then select the three interior points \(\eta_2,\eta_3,\eta_4\) according to cumulative signal activity, rather than uniformly in \(\eta\): \(\eta_2,\eta_3,\eta_4\) are chosen to approximately correspond to the \(25\%\), \(50\%\), and \(75\%\) levels of the accumulated target-behavior signal change over \([\eta_1,\eta_5]\). Thus, the sampled points cover different stages of the refusal-to-answer transition. The full procedure is summarized in Algorithm~\ref{alg:attack_sample}. In the main analysis, we report statistics over the first four sampled points \(\eta_1,\eta_2,\eta_3,\eta_4\), since \(\eta_5\) has already reached the final target-behavior signal and therefore does not represent the signal-increasing process itself.

\begin{algorithm}[t]
\caption{Intermediate-point extraction from the reference input transformation}
\label{alg:attack_sample}
\begin{algorithmic}[1]
\Require Reference input transformation \(X(\eta)\), reference target-behavior signal \(\Phi_{\mathcal Y_{\mathrm{ref}}}\), grid \(\{\bar\eta_i\}_{i=0}^{N_{\mathrm{grid}}-1}\)
\Ensure Four sampled progress values \(\eta_1<\eta_2<\eta_3<\eta_4\)

\State Evaluate \(\normsmall{\Phi_{\mathcal Y_{\mathrm{ref}}}(X(\bar\eta_i))}\) on the grid.
\State Set \(\eta_5\) to the earliest grid point whose signal reaches the final value \(\normsmall{\Phi_{\mathcal Y_{\mathrm{ref}}}(X(\eta_\star))}\).
\State Set \(\eta_1\) to the latest grid point up to \(\eta_5\) whose signal is still no larger than the initial value \(\normsmall{\Phi_{\mathcal Y_{\mathrm{ref}}}(X(\eta_0))}\).
\State Initialize the active interval list with \([\eta_1,\eta_5]\).
\State For each interval \(I\), compute its activity \(\mathcal A(I)\).
\State Iteratively split the interval with the largest activity into two halves, until the refinement budget is reached or all intervals have sufficiently small activity.
\State Select \(\eta_2,\eta_3,\eta_4\) from the refined interval boundaries so that they approximately match the \(25\%\), \(50\%\), and \(75\%\) cumulative target-behavior signal change levels.
\State Return \(\eta_1<\eta_2<\eta_3<\eta_4\).
\end{algorithmic}
\end{algorithm}

This procedure is designed to capture the stages of the reference input transformation that matter most for the refusal-to-answer transition. The endpoints \(\eta_1\) and \(\eta_5\) remove the inactive prefix and the already-saturated tail, while the activity-based refinement allocates more sampling resolution to regions where the reference target-behavior signal changes more strongly. As a result, the extracted points provide a compact summary of how the local target-behavior change aligns with \(R_{\mathrm{ref}}\) and with its operator-level contributions under the reference construction.
 
\subsubsection{Compute resources}
\label{app:compute}

All experiments are run on NVIDIA A100 80GB GPUs. As a representative reference point, under AutoDAN with 100 harmful--jailbreak prompt pairs, the full experiment takes about 18, 33, 22, 30, and 50 hours on a single A100 80GB GPU for Qwen3-4B, Qwen3-14B, Llama-3.1-8B-Instruct, Gemma-3-4B-IT, and Gemma-3-12B-IT, respectively. The actual runtime varies with attack method, and sequence length, but these numbers provide practical estimates of the computational cost required to reproduce the main results.

\subsection{Additional experimental results}

\subsubsection{Model-specific results of model analysis}
\label{app:model_results}

\begin{figure}[t]
    \centering
    \begin{subfigure}[t]{0.49\linewidth}
        \centering
        \includegraphics[width=\linewidth]{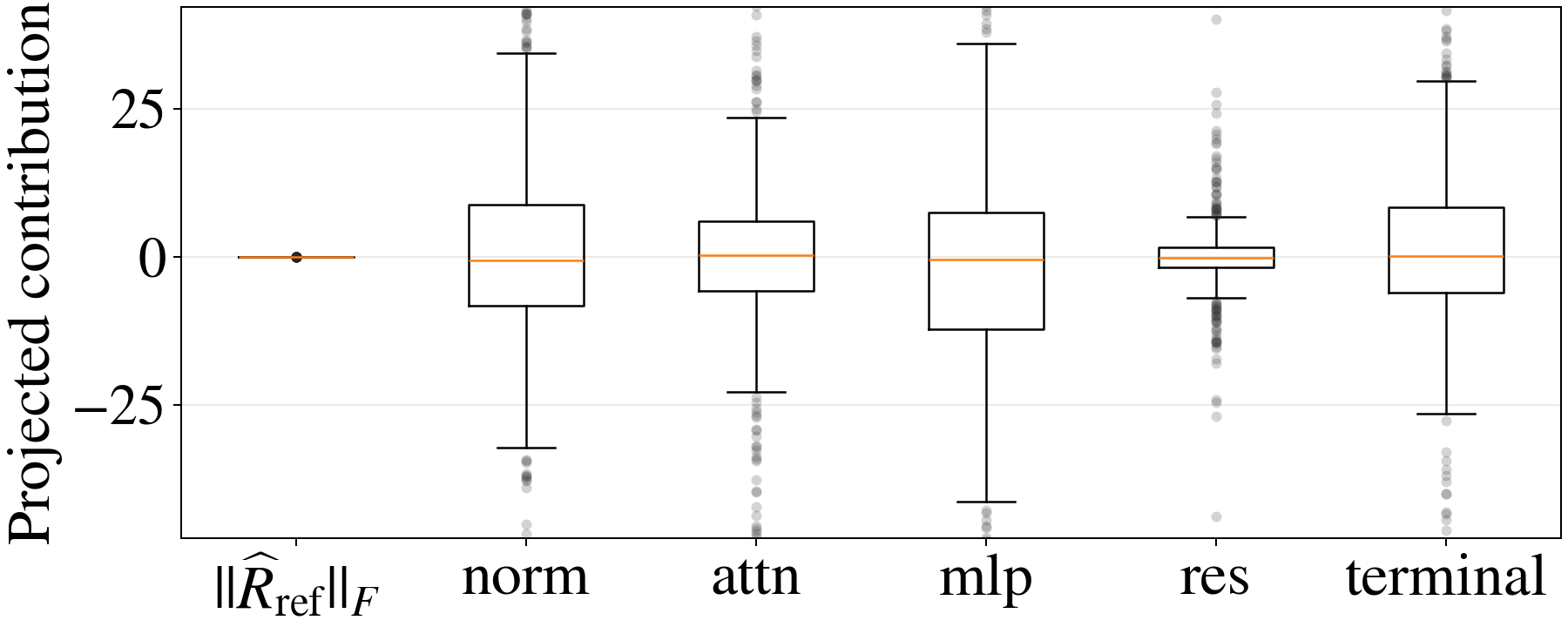}
        \caption{Before padding (Qwen3-4B).}
        \label{fig:operator_level_leakage_before_qwen3-4b}
    \end{subfigure}
    \hfill
    \begin{subfigure}[t]{0.49\linewidth}
        \centering
        \includegraphics[width=\linewidth]{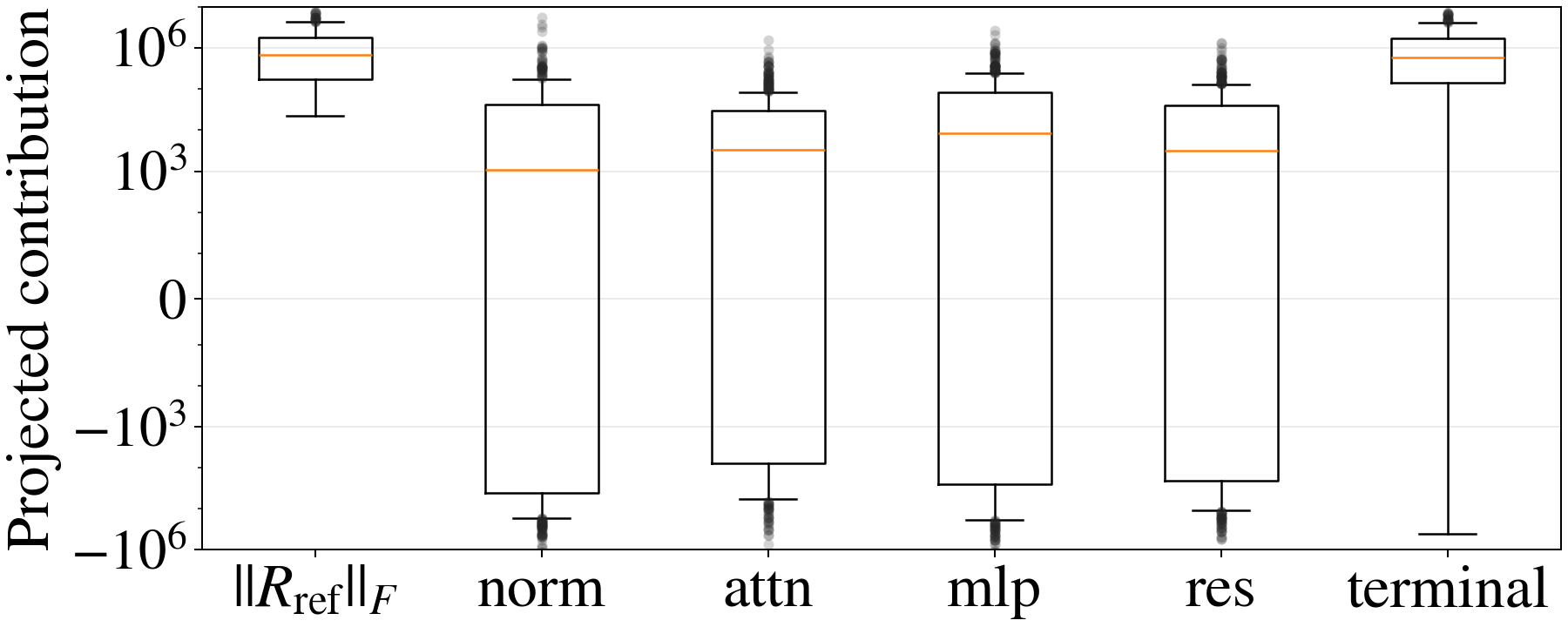}
        \caption{After padding (Qwen3-4B).}
        \label{fig:operator_level_leakage_after_qwen3-4b}
    \end{subfigure}

    \begin{subfigure}[t]{0.49\linewidth}
        \centering
        \includegraphics[width=\linewidth]{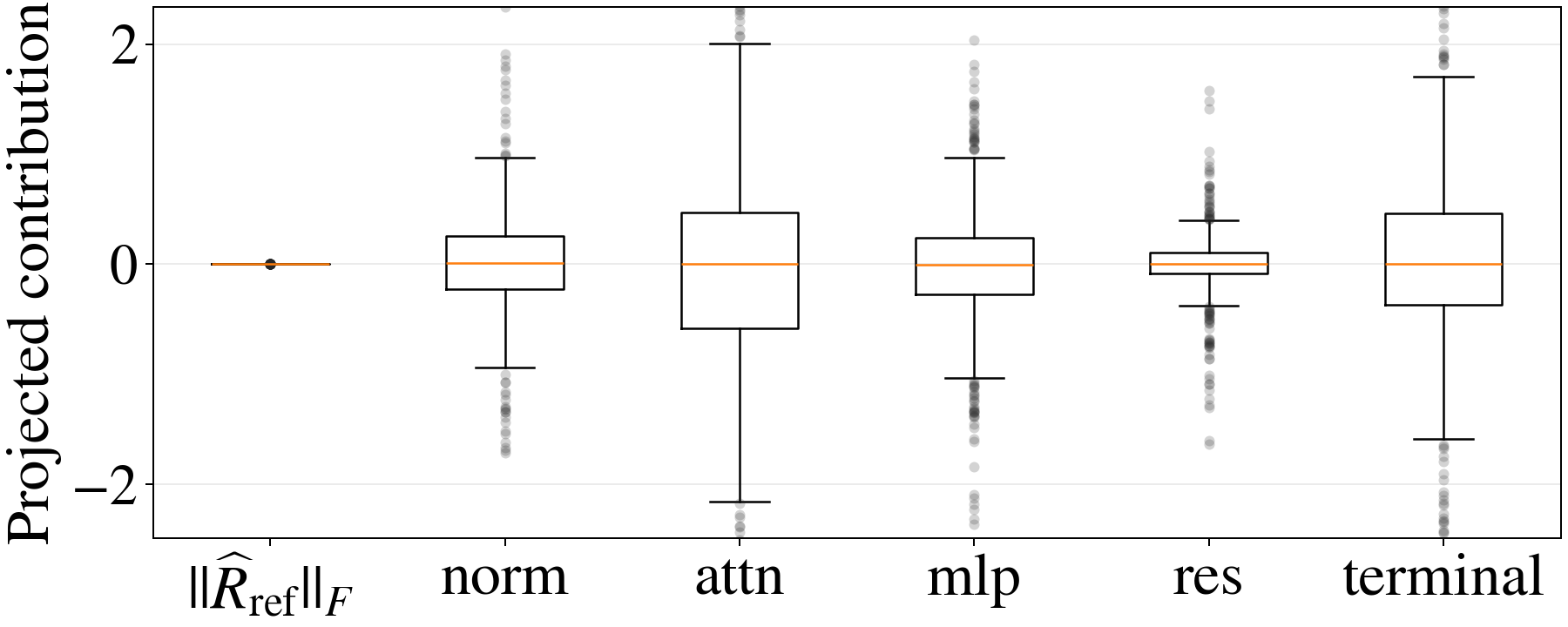}
        \caption{Before padding (Qwen3-14B).}
        \label{fig:operator_level_leakage_before_qwen3-14b}
    \end{subfigure}
    \hfill
    \begin{subfigure}[t]{0.49\linewidth}
        \centering
        \includegraphics[width=\linewidth]{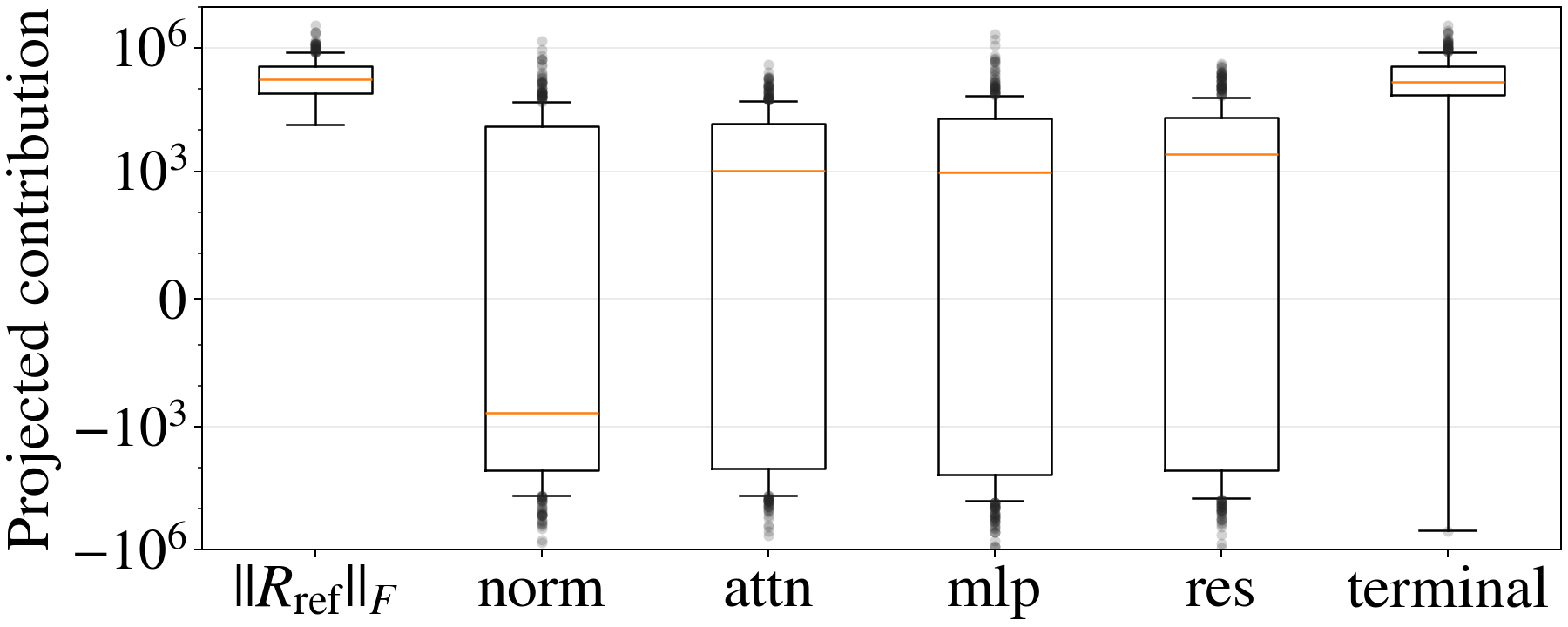}
        \caption{After padding (Qwen3-14B).}
        \label{fig:operator_level_leakage_after_qwen3-14b}
    \end{subfigure}

    \begin{subfigure}[t]{0.49\linewidth}
        \centering
        \includegraphics[width=\linewidth]{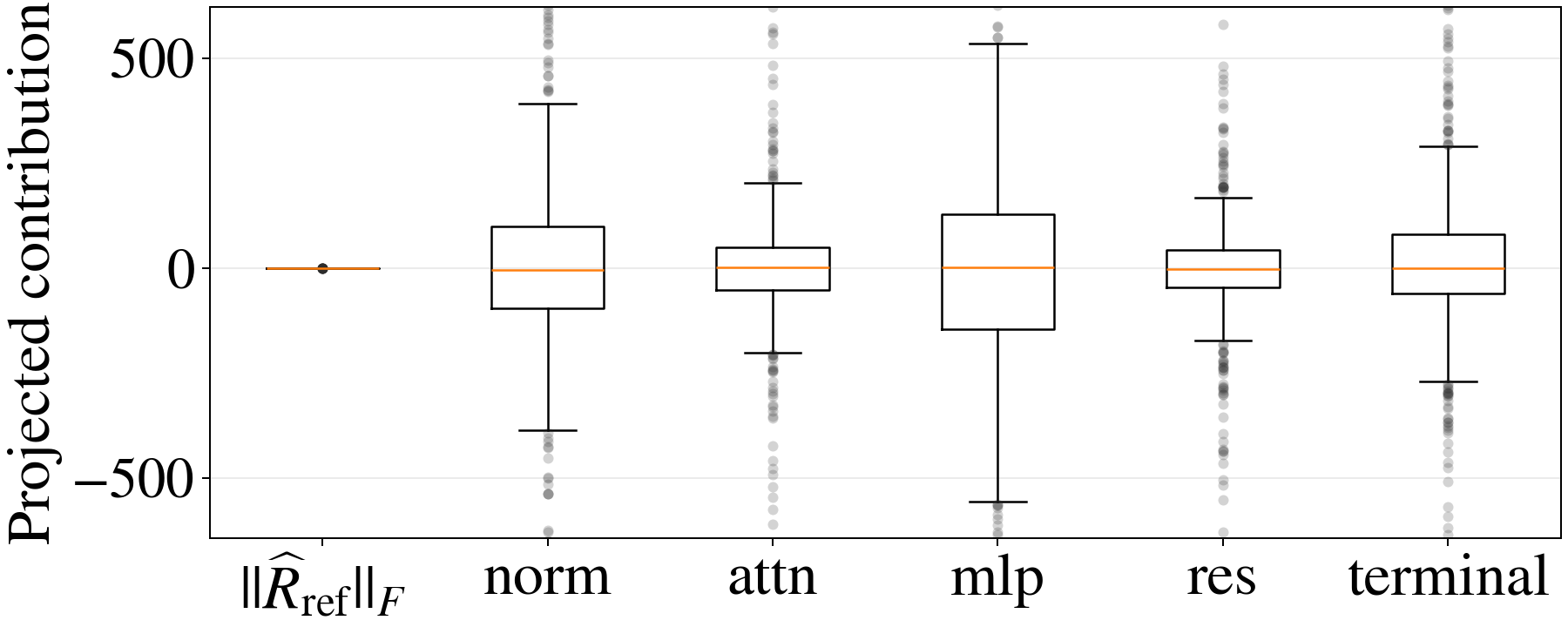}
        \caption{Before padding (Llama-3.1-8B-Instruct).}
        \label{fig:operator_level_leakage_before_llama3-8b}
    \end{subfigure}
    \hfill
    \begin{subfigure}[t]{0.49\linewidth}
        \centering
        \includegraphics[width=\linewidth]{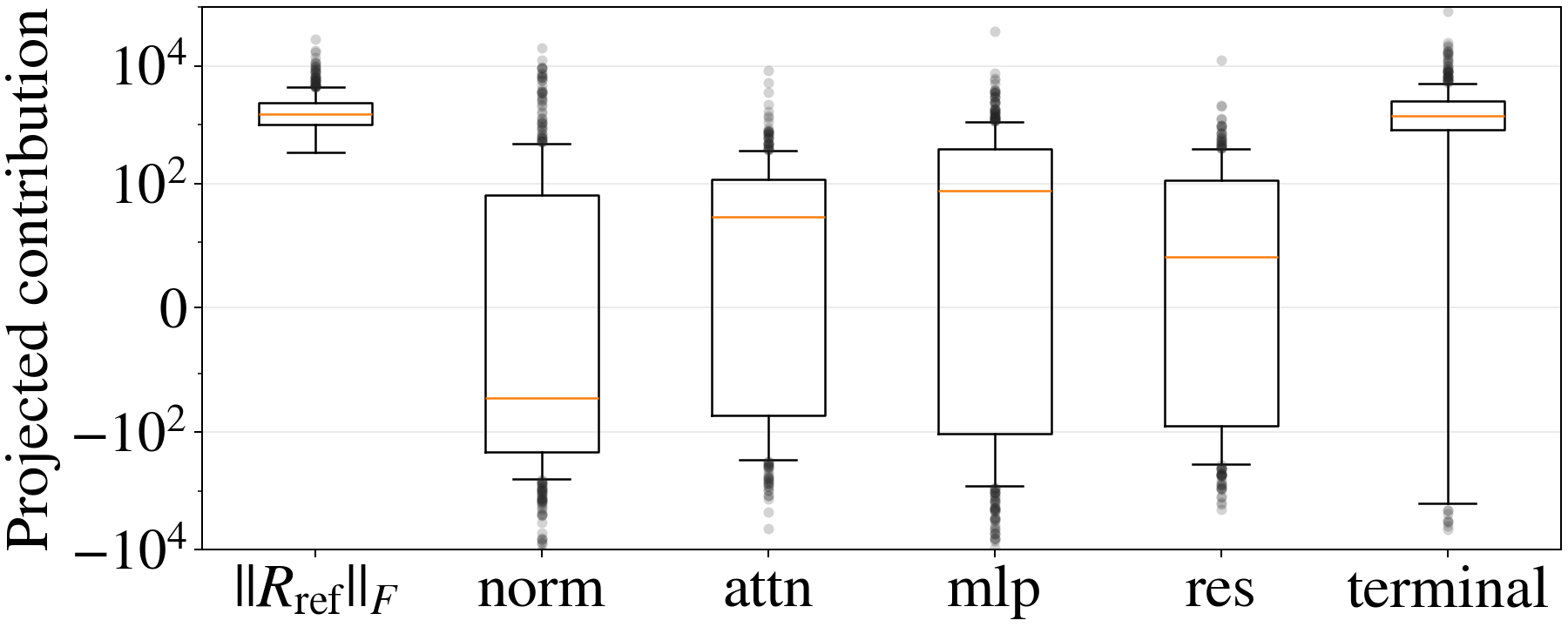}
        \caption{After padding (Llama-3.1-8B-Instruct).}
        \label{fig:operator_level_leakage_after_llama3-8b}
    \end{subfigure}

    \begin{subfigure}[t]{0.49\linewidth}
        \centering
        \includegraphics[width=\linewidth]{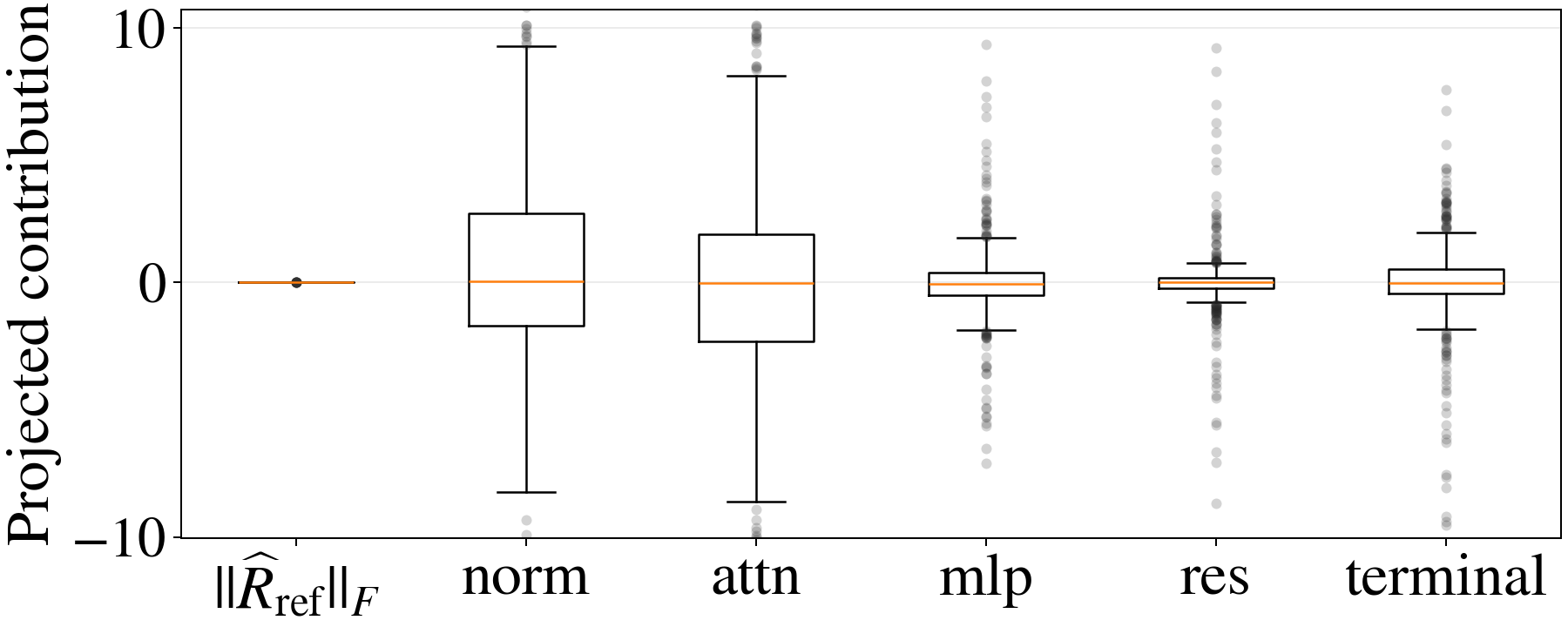}
        \caption{Before padding (Gemma-3-4B-IT).}
        \label{fig:operator_level_leakage_before_gemma-3-4b-it}
    \end{subfigure}
    \hfill
    \begin{subfigure}[t]{0.49\linewidth}
        \centering
        \includegraphics[width=\linewidth]{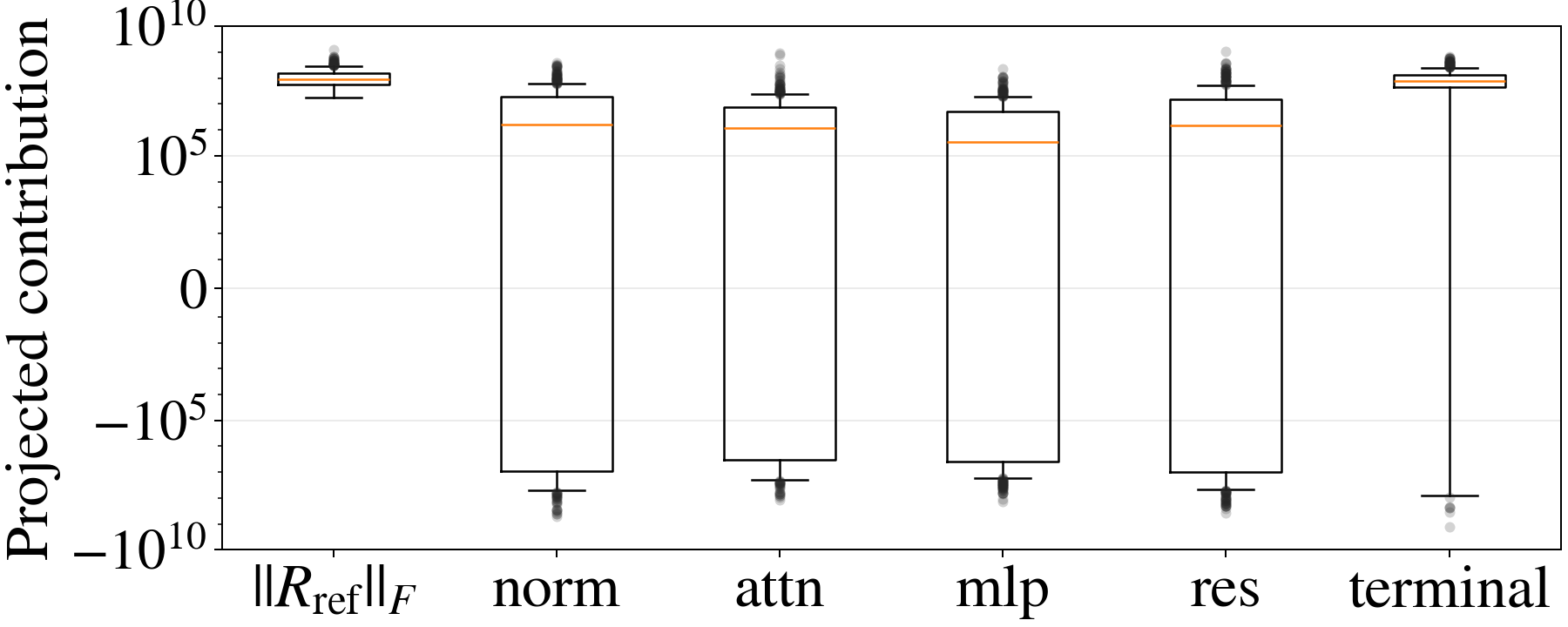}
        \caption{After padding (Gemma-3-4B-IT).}
        \label{fig:operator_level_leakage_after_gemma-3-4b-it}
    \end{subfigure}

    \begin{subfigure}[t]{0.49\linewidth}
        \centering
        \includegraphics[width=\linewidth]{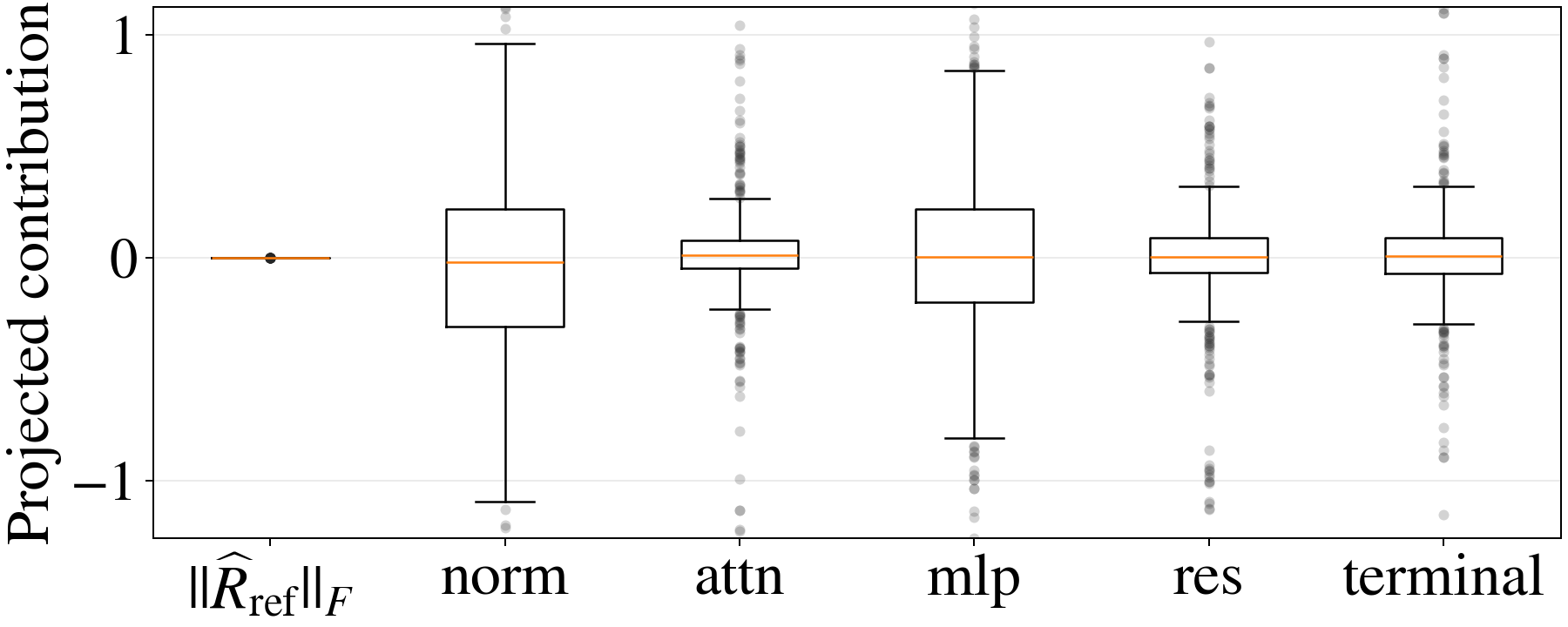}
        \caption{Before padding (Gemma-3-12B-IT).}
        \label{fig:operator_level_leakage_before_gemma-3-12b-it}
    \end{subfigure}
    \hfill
    \begin{subfigure}[t]{0.49\linewidth}
        \centering
        \includegraphics[width=\linewidth]{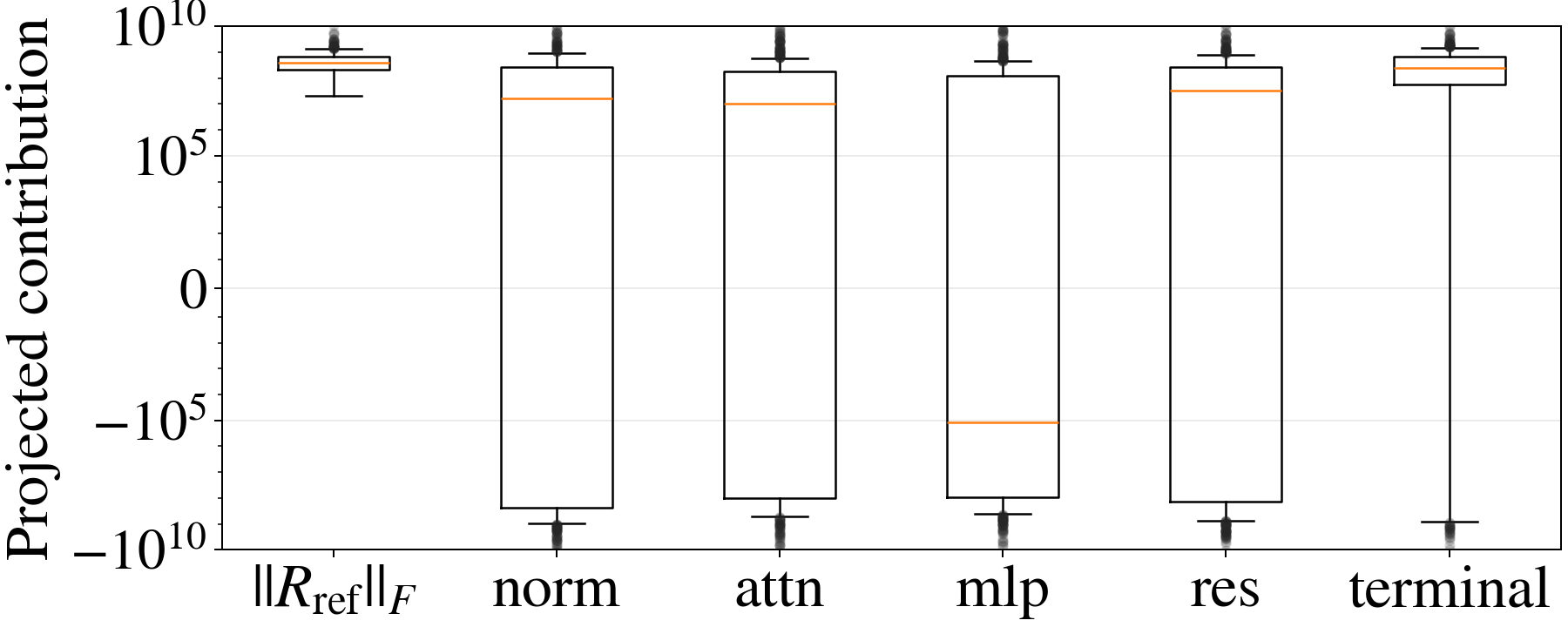}
        \caption{After padding (Gemma-3-12B-IT).}
        \label{fig:operator_level_leakage_after_gemma-3-12b-it}
    \end{subfigure}

    \caption{Model-specific reference refusal-escape direction and operator-level contributions at harmful inputs. Each row reports one model, aggregated over all harmful samples for that model. Left panels show the unpadded baseline, where the corresponding refusal-escape direction is zero by construction and operator-level contributions are visualized by signed projection onto \(e_1\). Right panels show the padded harmful input, where added token dimensions expose a nonzero \(R_{\mathrm{ref}}(X)\) and operator-level contributions are visualized by signed projection onto \(R_{\mathrm{ref}}(X)/\normsmall{R_{\mathrm{ref}}(X)}\).}
    \label{fig:operator_level_leakage_model_specific}
\end{figure}

Figure~\ref{fig:operator_level_leakage_model_specific} shows the model-specific results for Observation~1. Although the scales vary across model families, the qualitative pattern is consistent with the aggregate result in Section~\ref{model_analysis}. Before placeholder padding, the controlled baseline satisfies \(\widehat R_{\mathrm{ref}}(\widehat X)=0\) by construction, while nonzero operator-level contributions are already present and cancel each other at the input side. After placeholder padding, this cancellation no longer holds: added token dimensions expose much larger nonzero RED at harmful inputs. Across models, all operator-level sources leave observable contributions after padding, with the terminal source appearing as a particularly stable and prominent contributor. These model-specific results support the same conclusion as the aggregate result: the influence from harmful-semantics interpretation to answer-versus-refusal behavior is often relatively robust, while the model still retains substantial side-channels outside the harmful-semantics interpretation. 

\subsubsection{Attack-specific and model-specific results of jailbreak analysis}
\label{app:attack_results}

\begin{figure}[t]
    \centering
    \begin{subfigure}[t]{0.32\linewidth}
        \centering
        \includegraphics[width=\linewidth]{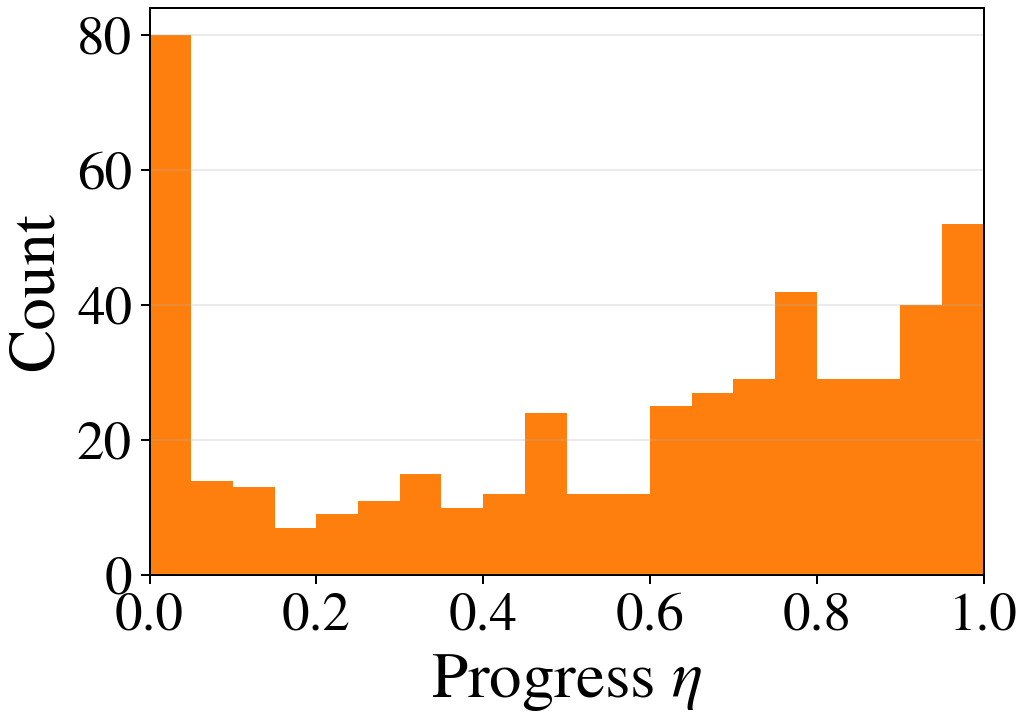}
        \caption{Earliest jailbreak progress (GCG).}
        \label{fig:jailbreak_progress_gcg}
    \end{subfigure}
    \hfill
    \begin{subfigure}[t]{0.32\linewidth}
        \centering
        \includegraphics[width=\linewidth]{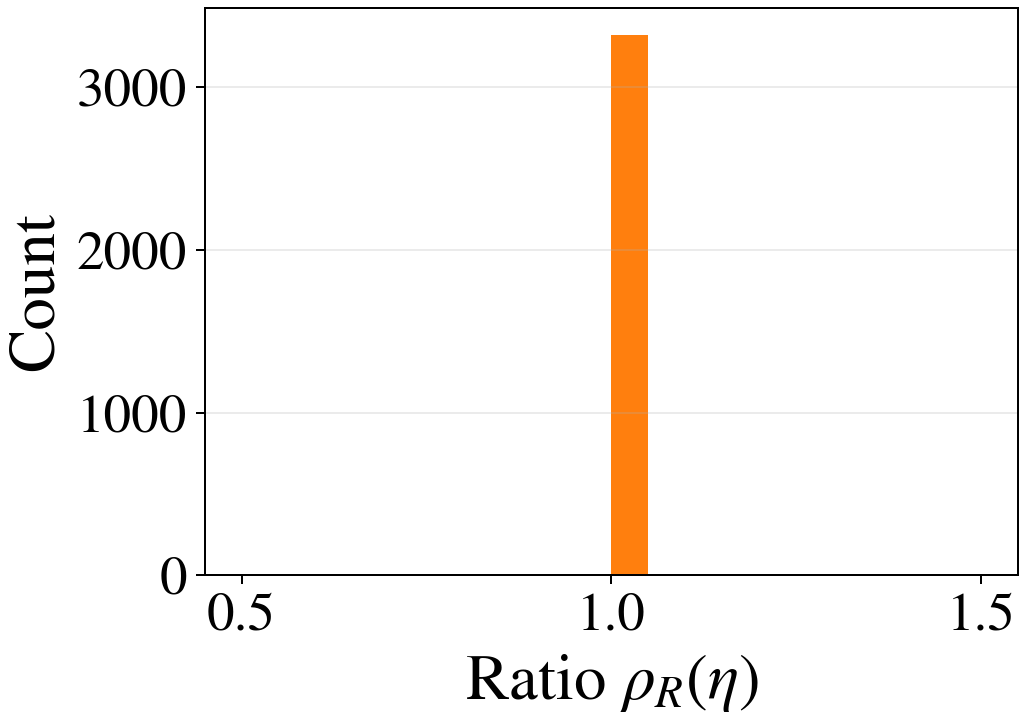}
        \caption{Signed \(R_{\mathrm{ref}}\) ratio (GCG).}
        \label{fig:jailbreak_ratio_gcg}
    \end{subfigure}
    \hfill
    \begin{subfigure}[t]{0.32\linewidth}
        \centering
        \includegraphics[width=\linewidth]{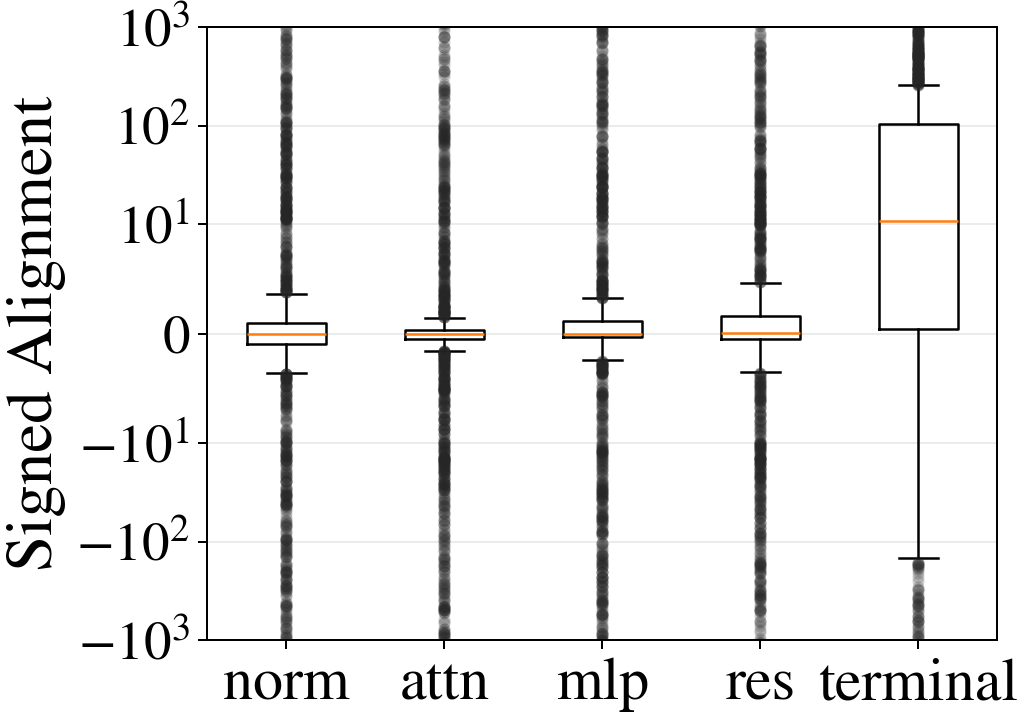}
        \caption{Signed operator-level alignment (GCG).}
        \label{fig:jailbreak_family_gcg}
    \end{subfigure}

    \begin{subfigure}[t]{0.32\linewidth}
        \centering
        \includegraphics[width=\linewidth]{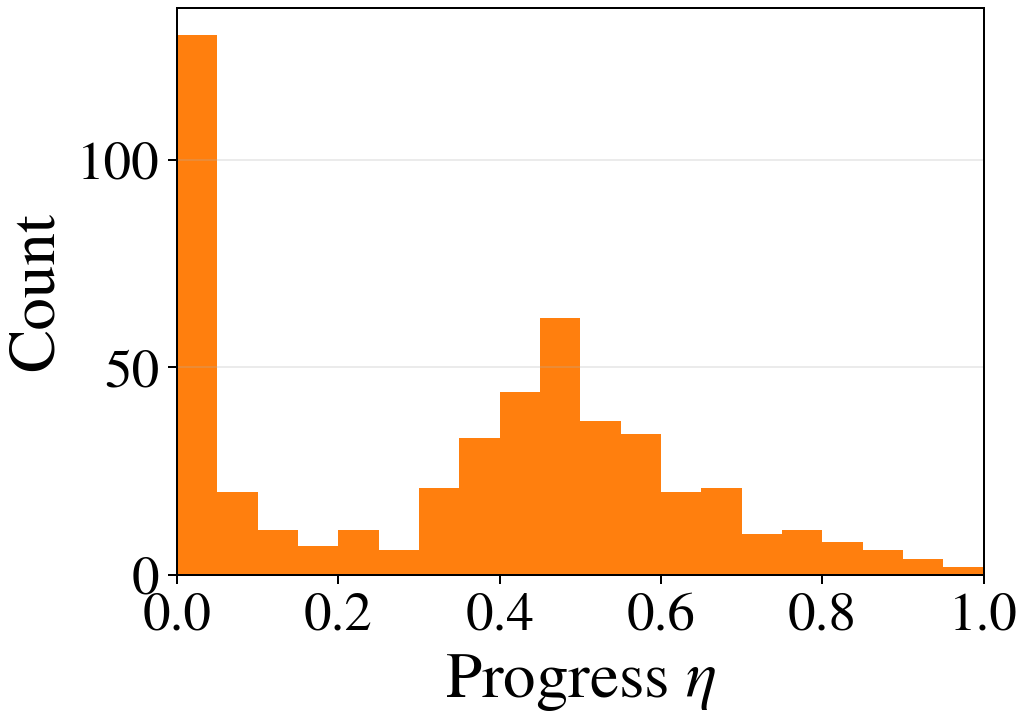}
        \caption{Earliest jailbreak progress (AutoDAN).}
        \label{fig:jailbreak_progress_autodan}
    \end{subfigure}
    \hfill
    \begin{subfigure}[t]{0.32\linewidth}
        \centering
        \includegraphics[width=\linewidth]{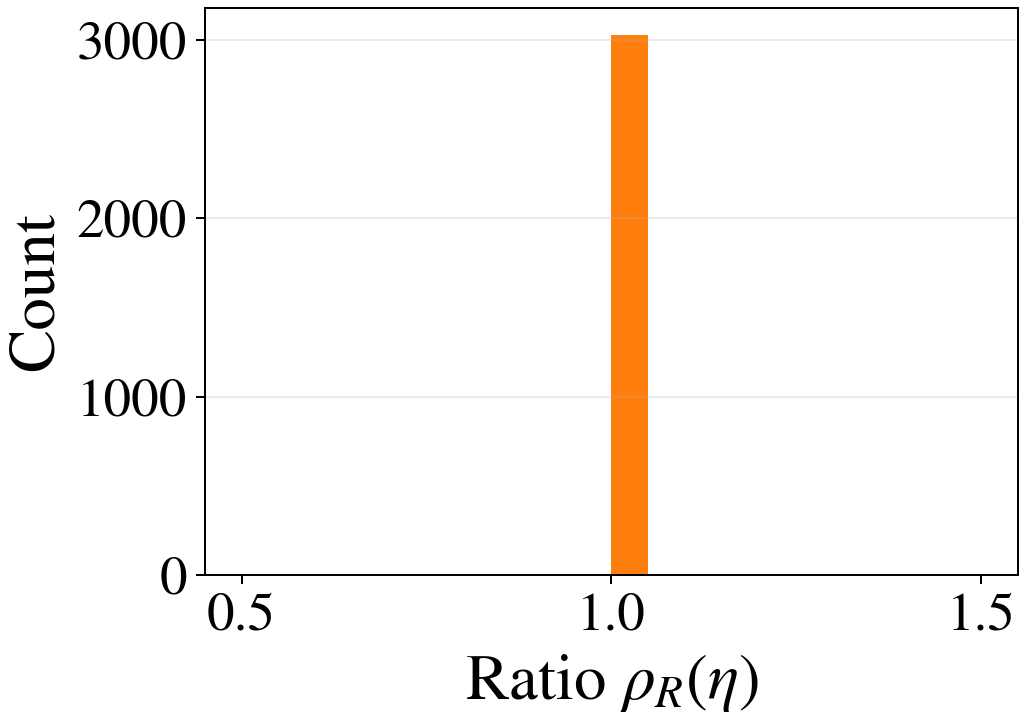}
        \caption{Signed \(R_{\mathrm{ref}}\) ratio (AutoDAN).}
        \label{fig:jailbreak_ratio_autodan}
    \end{subfigure}
    \hfill
    \begin{subfigure}[t]{0.32\linewidth}
        \centering
        \includegraphics[width=\linewidth]{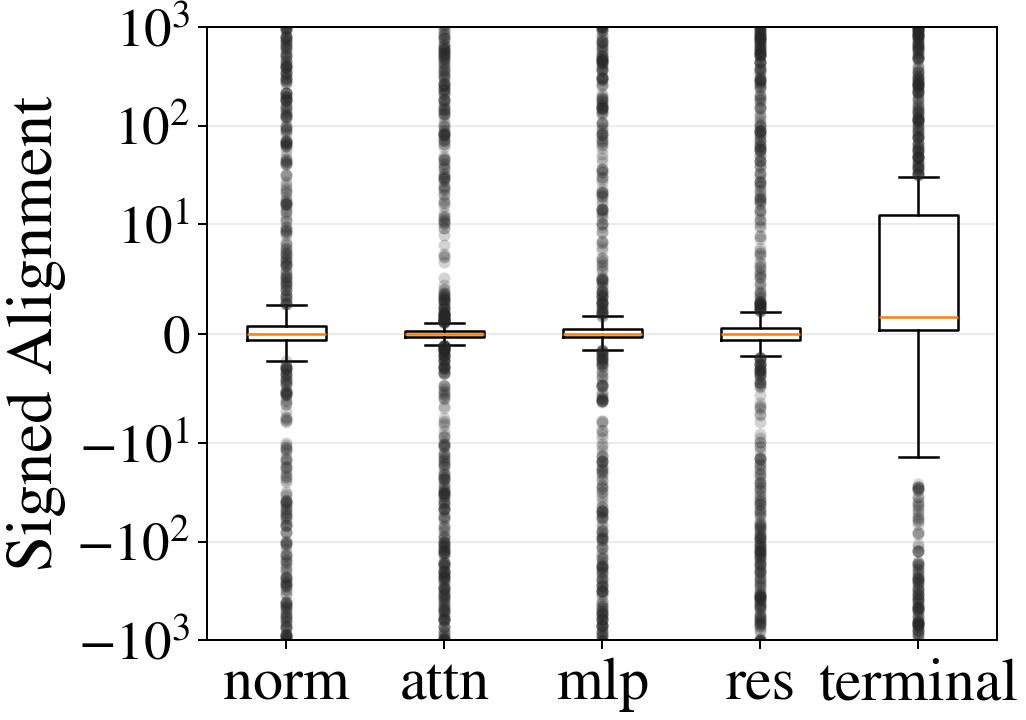}
        \caption{Signed operator-level alignment (AutoDAN).}
        \label{fig:jailbreak_family_autodan}
    \end{subfigure}

    \begin{subfigure}[t]{0.32\linewidth}
        \centering
        \includegraphics[width=\linewidth]{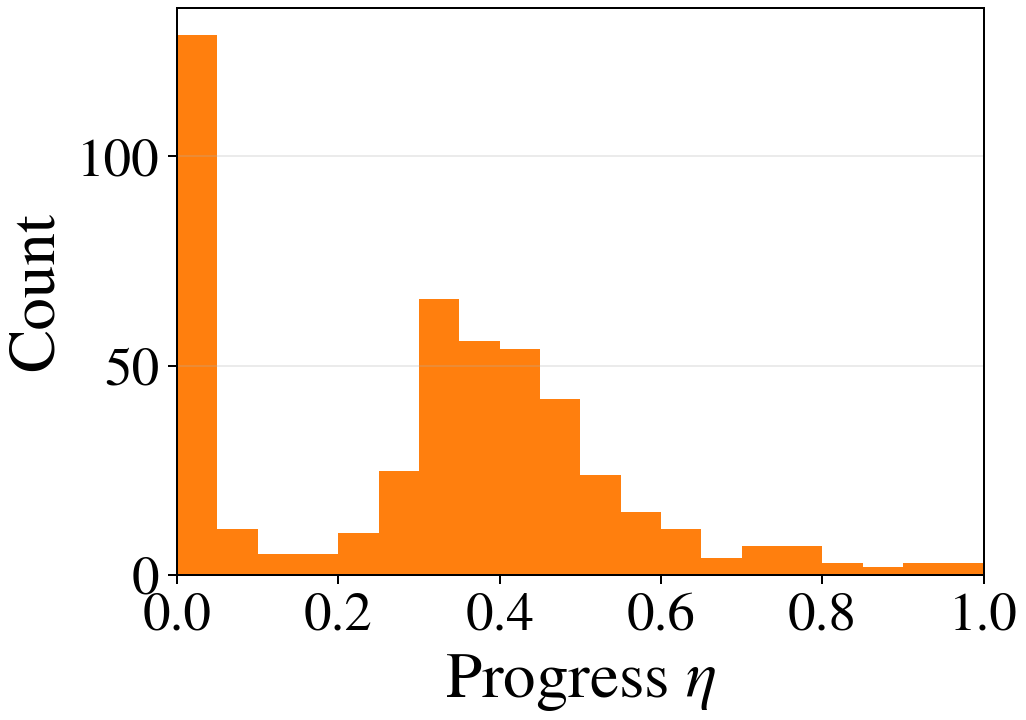}
        \caption{Earliest jailbreak progress (GPTFuzzer).}
        \label{fig:jailbreak_progress_gptfuzzer}
    \end{subfigure}
    \hfill
    \begin{subfigure}[t]{0.32\linewidth}
        \centering
        \includegraphics[width=\linewidth]{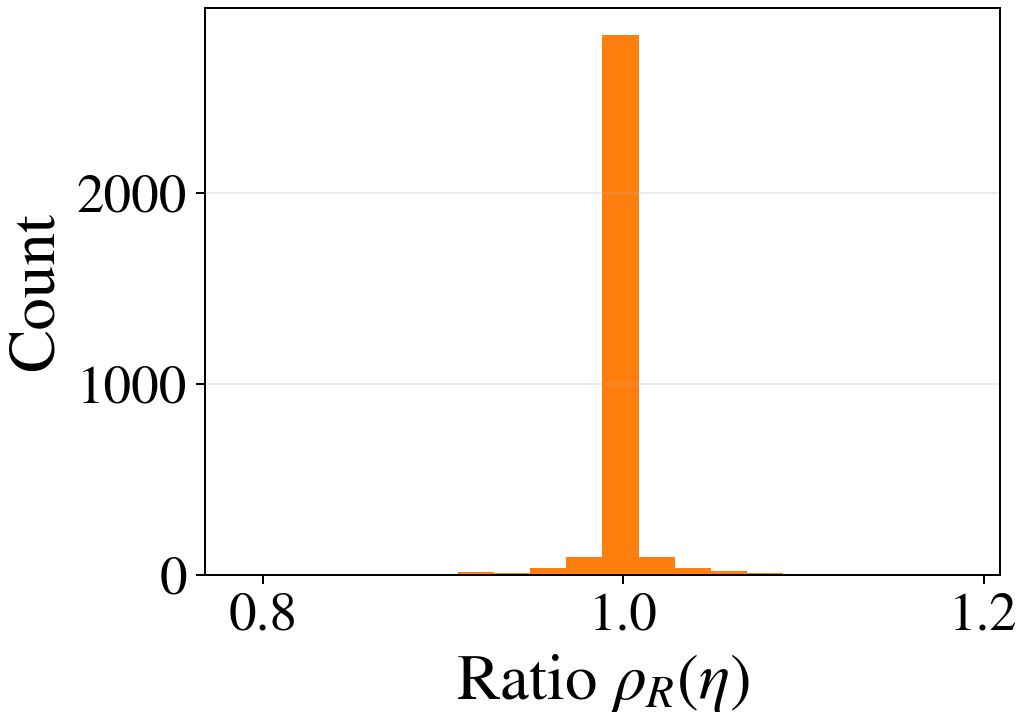}
        \caption{Signed \(R_{\mathrm{ref}}\) ratio (GPTFuzzer).}
        \label{fig:jailbreak_ratio_gptfuzzer}
    \end{subfigure}
    \hfill
    \begin{subfigure}[t]{0.32\linewidth}
        \centering
        \includegraphics[width=\linewidth]{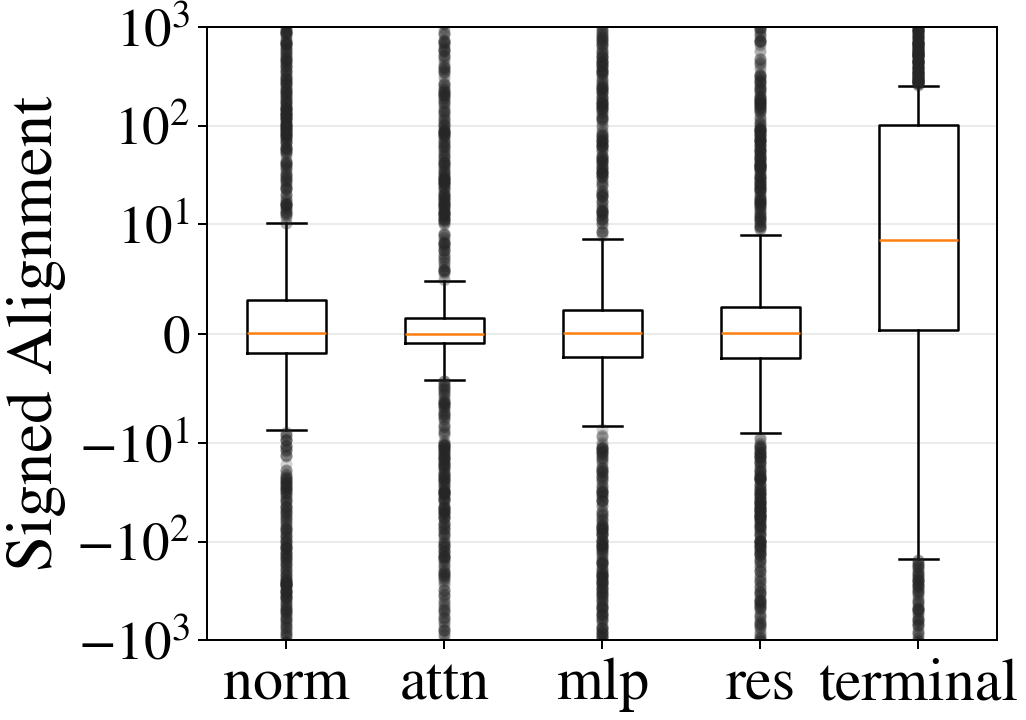}
        \caption{Signed operator-level alignment (GPTFuzzer).}
        \label{fig:jailbreak_family_gptfuzzer}
    \end{subfigure}

    \begin{subfigure}[t]{0.32\linewidth}
        \centering
        \includegraphics[width=\linewidth]{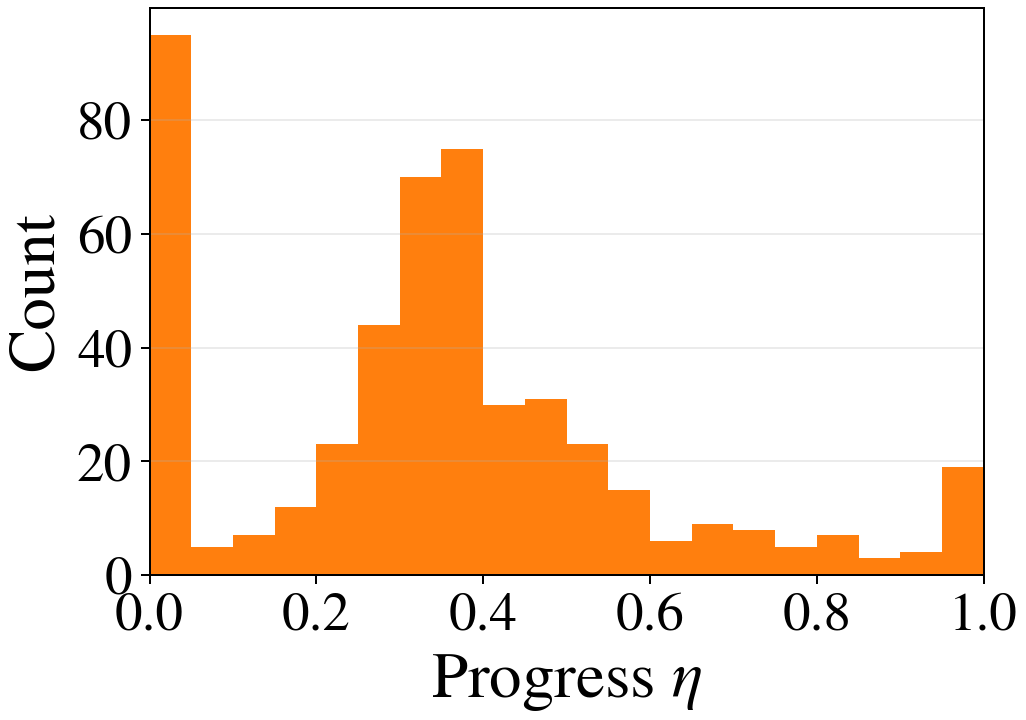}
        \caption{Earliest jailbreak progress (TAP).}
        \label{fig:jailbreak_progress_tap}
    \end{subfigure}
    \hfill
    \begin{subfigure}[t]{0.32\linewidth}
        \centering
        \includegraphics[width=\linewidth]{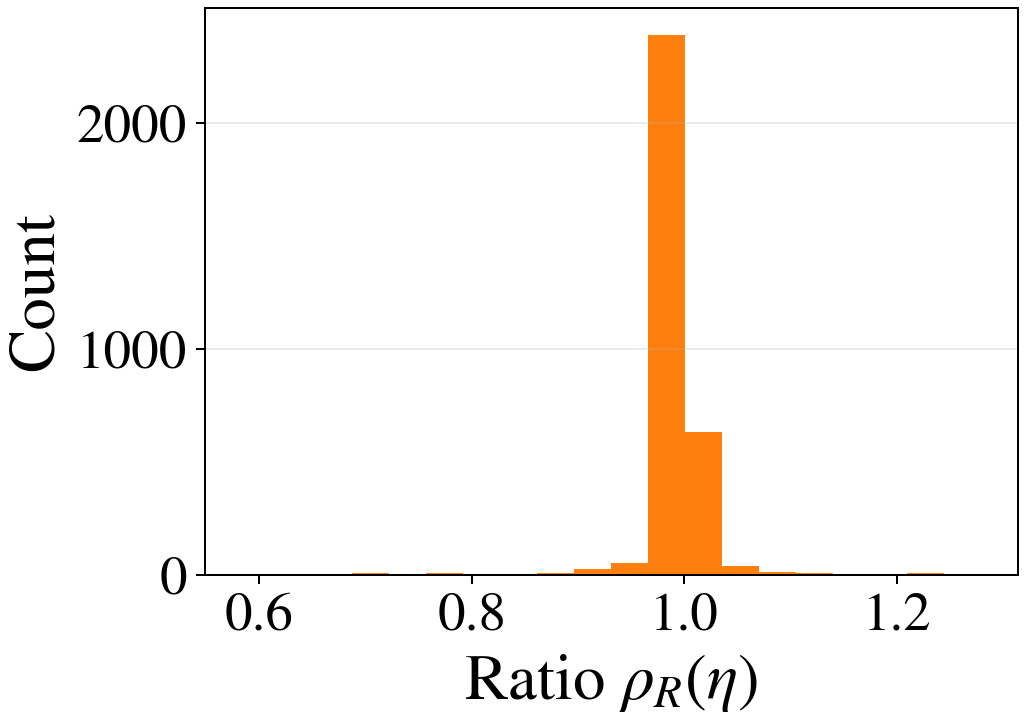}
        \caption{Signed \(R_{\mathrm{ref}}\) ratio (TAP).}
        \label{fig:jailbreak_ratio_tap}
    \end{subfigure}
    \hfill
    \begin{subfigure}[t]{0.32\linewidth}
        \centering
        \includegraphics[width=\linewidth]{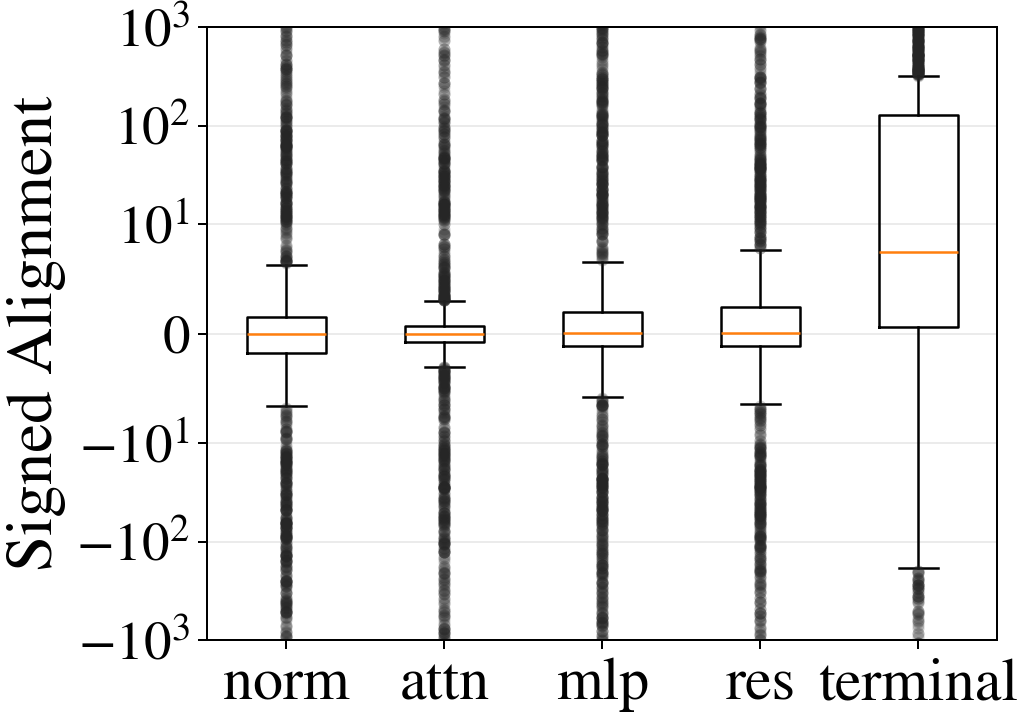}
        \caption{Signed operator-level alignment (TAP).}
        \label{fig:jailbreak_family_tap}
    \end{subfigure}

    \begin{subfigure}[t]{0.32\linewidth}
        \centering
        \includegraphics[width=\linewidth]{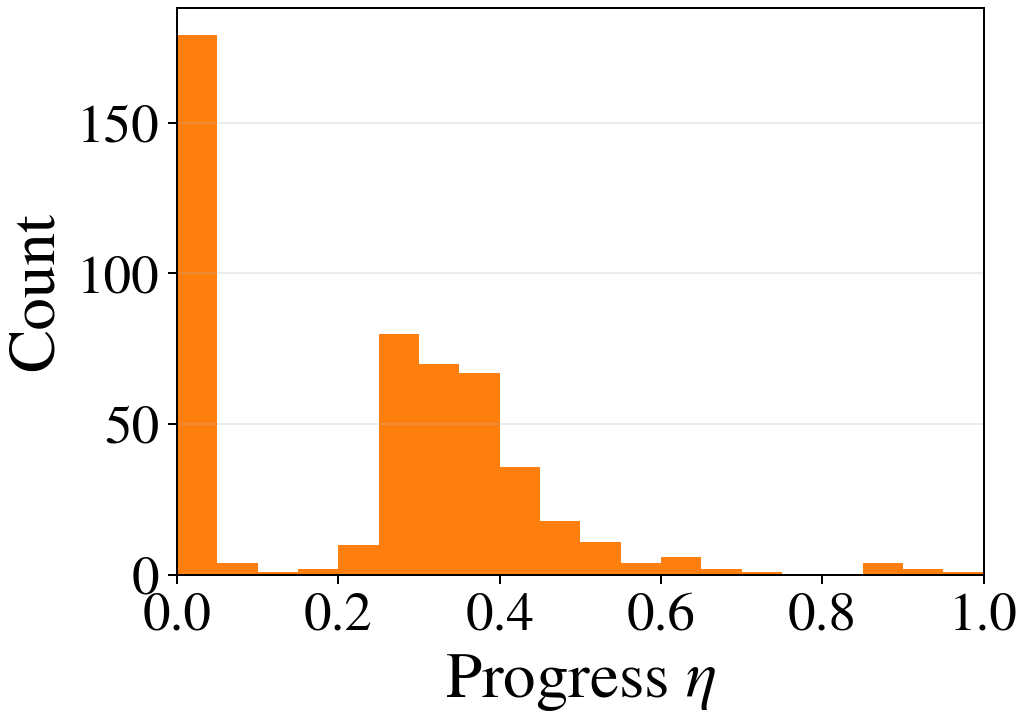}
        \caption{Earliest jailbreak progress (ReNeLLM).}
        \label{fig:jailbreak_progress_renellm}
    \end{subfigure}
    \hfill
    \begin{subfigure}[t]{0.32\linewidth}
        \centering
        \includegraphics[width=\linewidth]{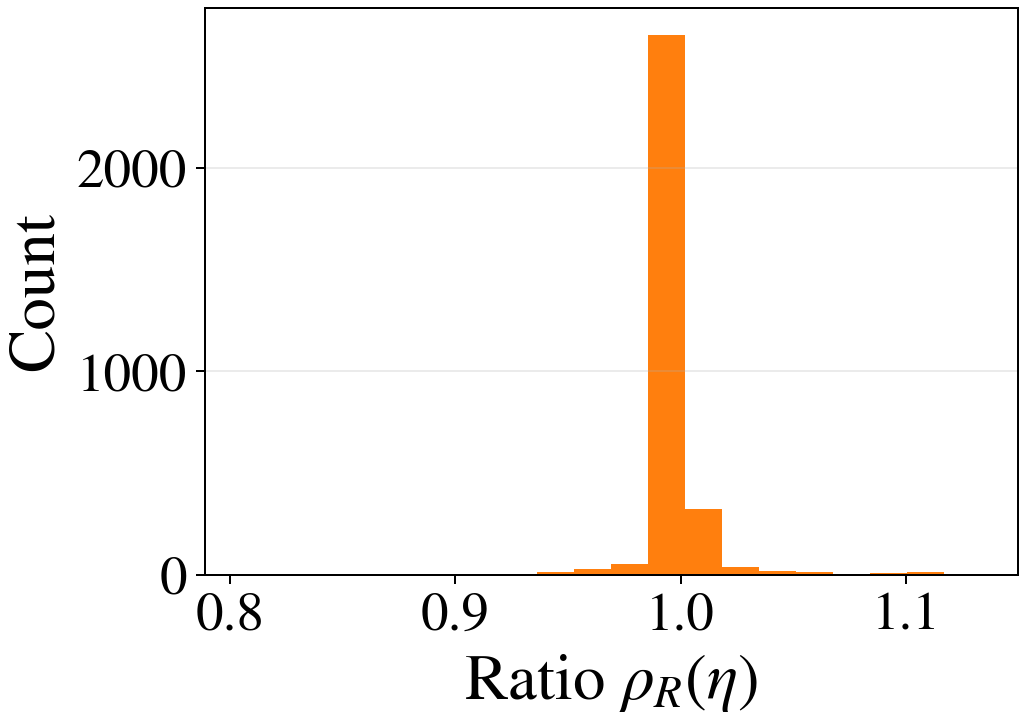}
        \caption{Signed \(R_{\mathrm{ref}}\) ratio (ReNeLLM).}
        \label{fig:jailbreak_ratio_renellm}
    \end{subfigure}
    \hfill
    \begin{subfigure}[t]{0.32\linewidth}
        \centering
        \includegraphics[width=\linewidth]{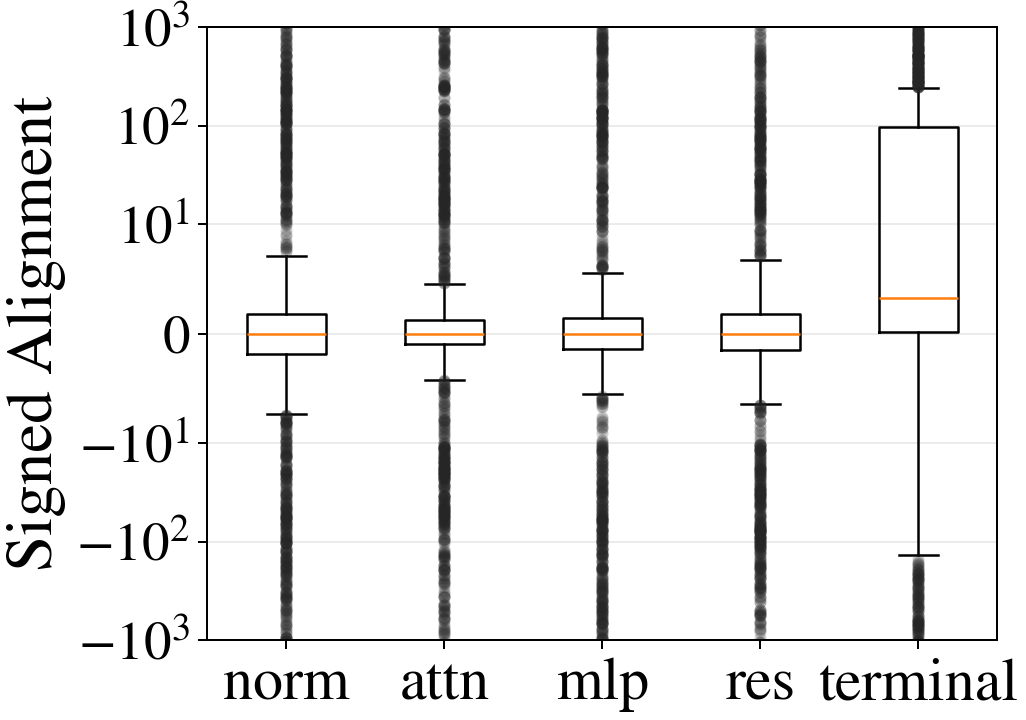}
        \caption{Signed operator-level alignment (ReNeLLM).}
        \label{fig:jailbreak_family_renellm}
    \end{subfigure}

    \caption{Attack-specific jailbreak analysis under the continuous input-transformation view. Results are reported separately for each attack method and aggregated over all samples and models within that attack method. In each row, the left panel shows the earliest jailbreak progress, the middle panel shows the signed \(R_{\mathrm{ref}}\) ratio, and the right panel shows the signed alignment between the transformation tangent and each operator-level contribution \(S_f^\Sigma\).}
    \label{fig:jailbreak_main_attack_specific}
\end{figure}

Figure~\ref{fig:jailbreak_main_attack_specific} shows the attack-specific results for Observation~2. The earliest-progress distributions reveal some differences across attack methods. For GCG, the successful-jailbreak progress values are more dispersed along the reference input transformation, which may reflect that gradient-guided nonsensical adversarial suffixes rely more strongly on the completed prompt construction. In contrast, AutoDAN, GPTFuzzer, TAP, and ReNeLLM are closer to the aggregate pattern in Section~\ref{jailbreak_analysis}: a substantial fraction of cases are judged successful in the first half of the transformation, with many occurring near the beginning.

Despite these differences in earliest progress, the remaining results are consistent across attack methods. The signed ratio \(\rho_R(\eta)\) remains close to \(1\), indicating that most of the local reference target-behavior change is accounted for by \(R_{\mathrm{ref}}\) under our reference construction. The signed operator-level alignments also show that terminal-source contributions are most consistently positive, whereas other operator-level contributions are more dispersed. Across attack methods, these results support the same interpretation as the aggregate result: successful jailbreaks can be understood as exploiting refusal-escape directions already visible at harmful inputs, especially the terminal-source side-channel, rather than as creating a separate mechanism along the reference input transformation.

\begin{figure}[t]
    \centering
    \begin{subfigure}[t]{0.32\linewidth}
        \centering
        \includegraphics[width=\linewidth]{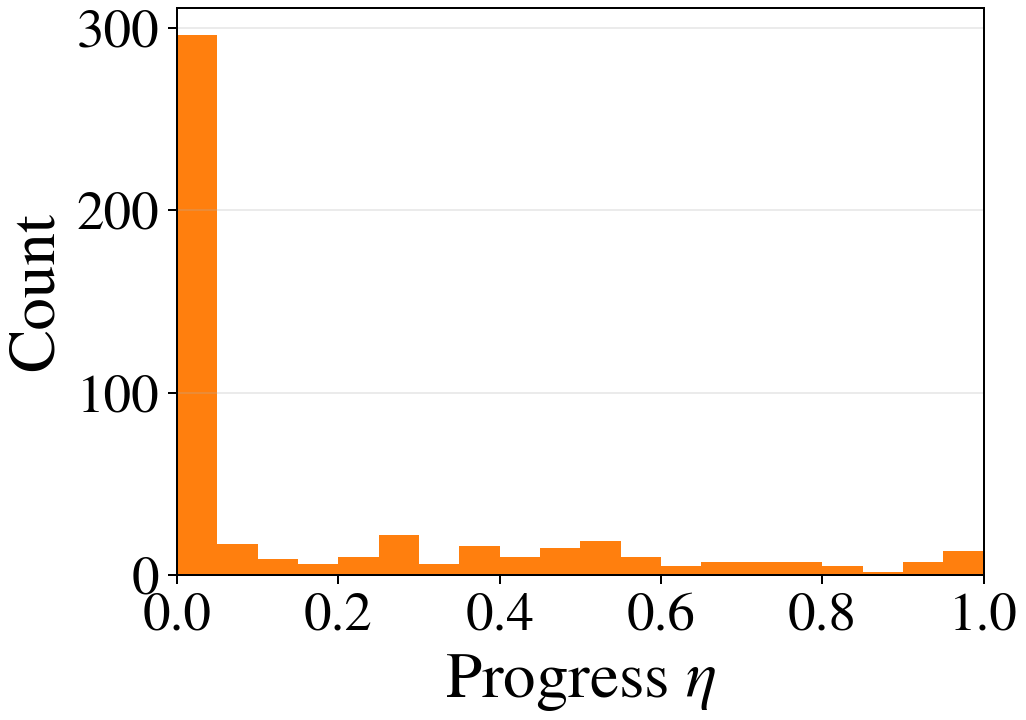}
        \caption{Earliest jailbreak progress (Qwen3-4B).}
        \label{fig:jailbreak_progress_qwen3-4b}
    \end{subfigure}
    \hfill
    \begin{subfigure}[t]{0.32\linewidth}
        \centering
        \includegraphics[width=\linewidth]{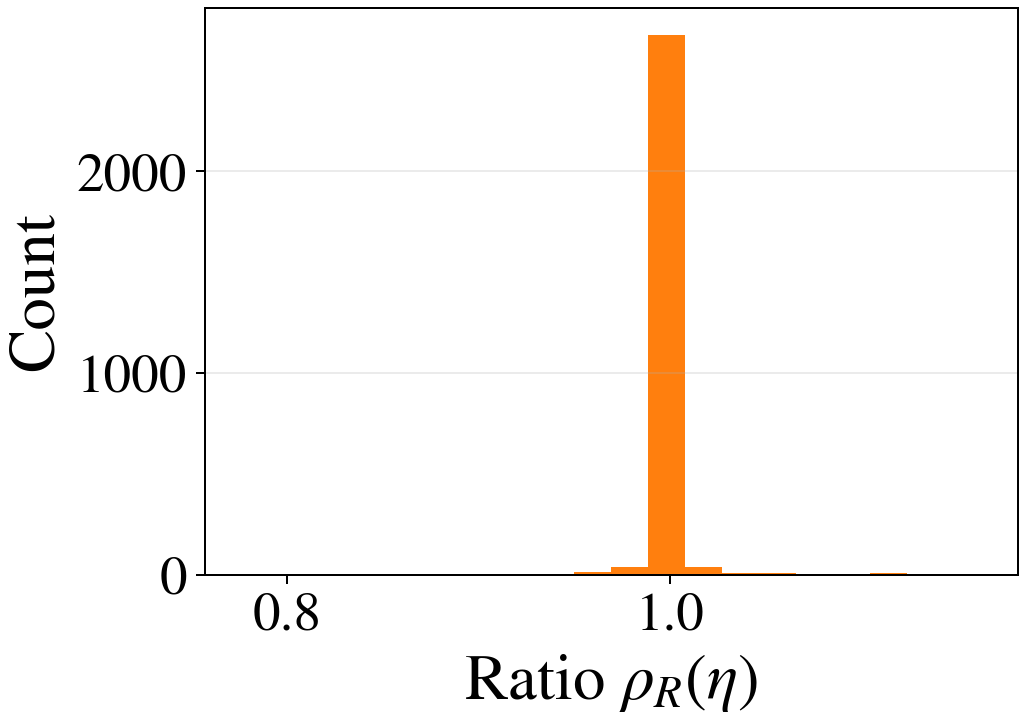}
        \caption{Signed \(R_{\mathrm{ref}}\) ratio (Qwen3-4B).}
        \label{fig:jailbreak_ratio_qwen3-4b}
    \end{subfigure}
    \hfill
    \begin{subfigure}[t]{0.32\linewidth}
        \centering
        \includegraphics[width=\linewidth]{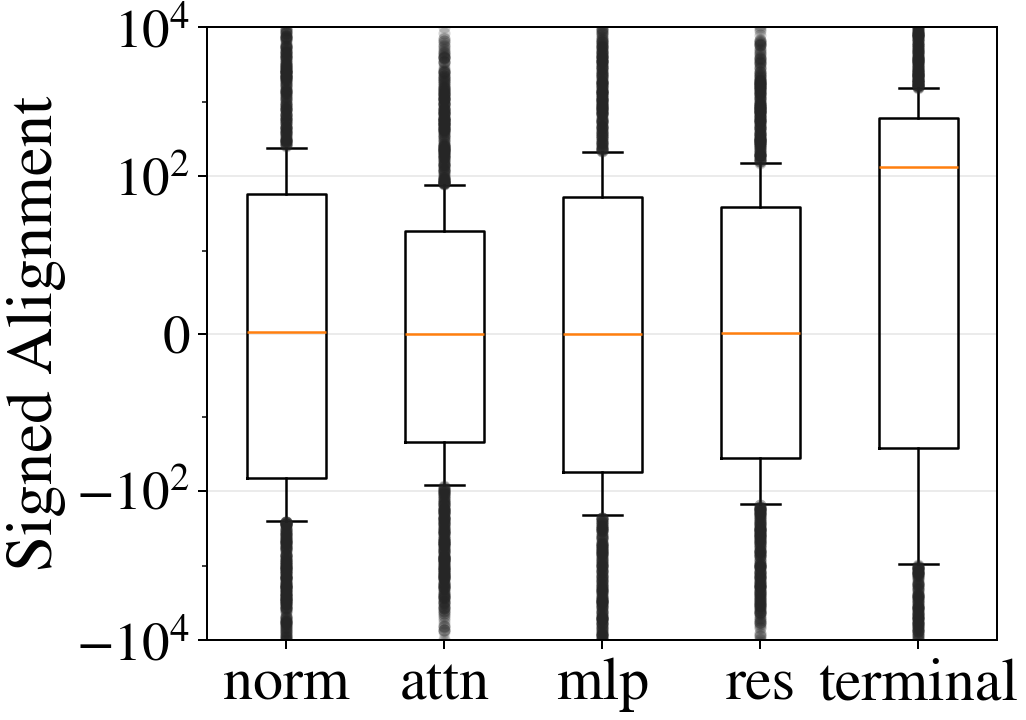}
        \caption{Signed operator-level alignment (Qwen3-4B).}
        \label{fig:jailbreak_family_qwen3-4b}
    \end{subfigure}

    \begin{subfigure}[t]{0.32\linewidth}
        \centering
        \includegraphics[width=\linewidth]{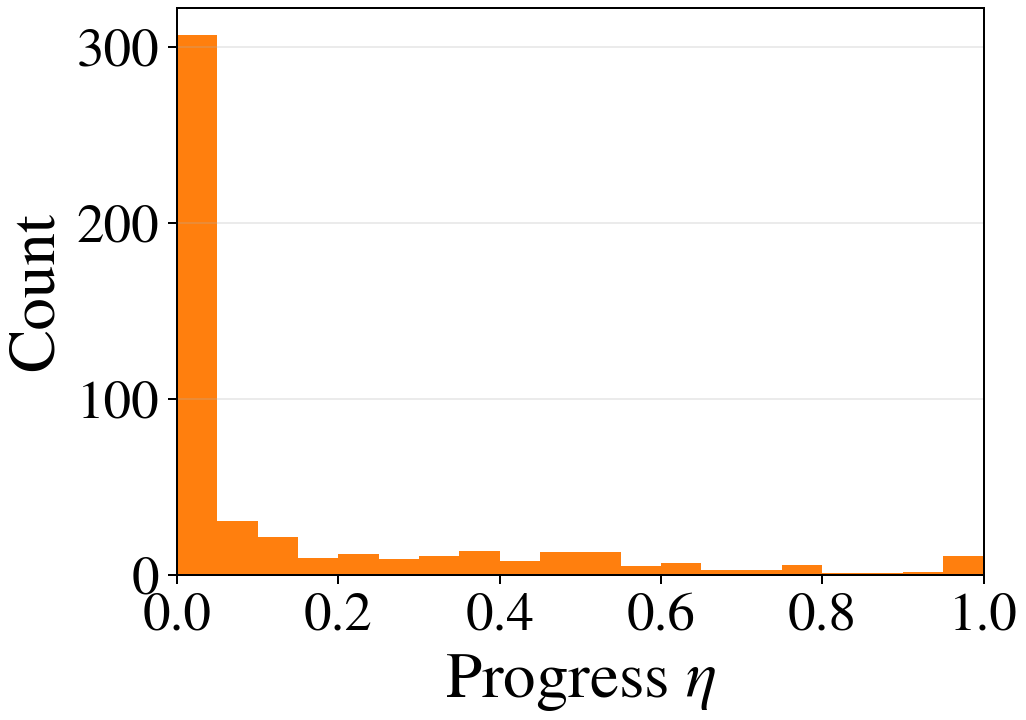}
        \caption{Earliest jailbreak progress (Qwen3-14B).}
        \label{fig:jailbreak_progress_qwen3-14b}
    \end{subfigure}
    \hfill
    \begin{subfigure}[t]{0.32\linewidth}
        \centering
        \includegraphics[width=\linewidth]{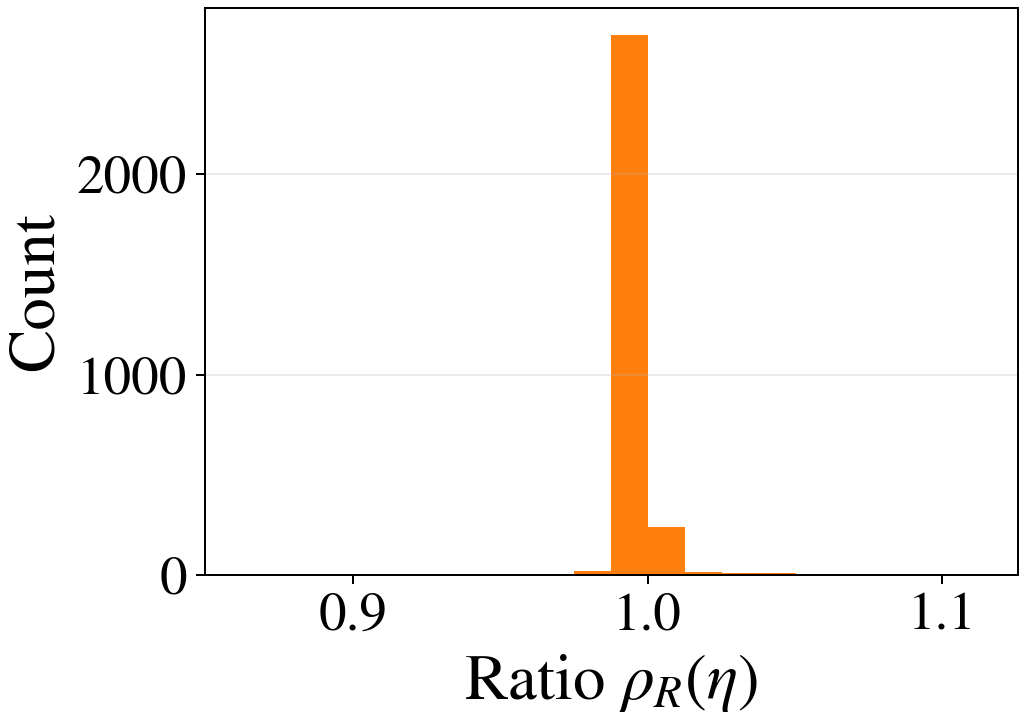}
        \caption{Signed \(R_{\mathrm{ref}}\) ratio (Qwen3-14B).}
        \label{fig:jailbreak_ratio_qwen3-14b}
    \end{subfigure}
    \hfill
    \begin{subfigure}[t]{0.32\linewidth}
        \centering
        \includegraphics[width=\linewidth]{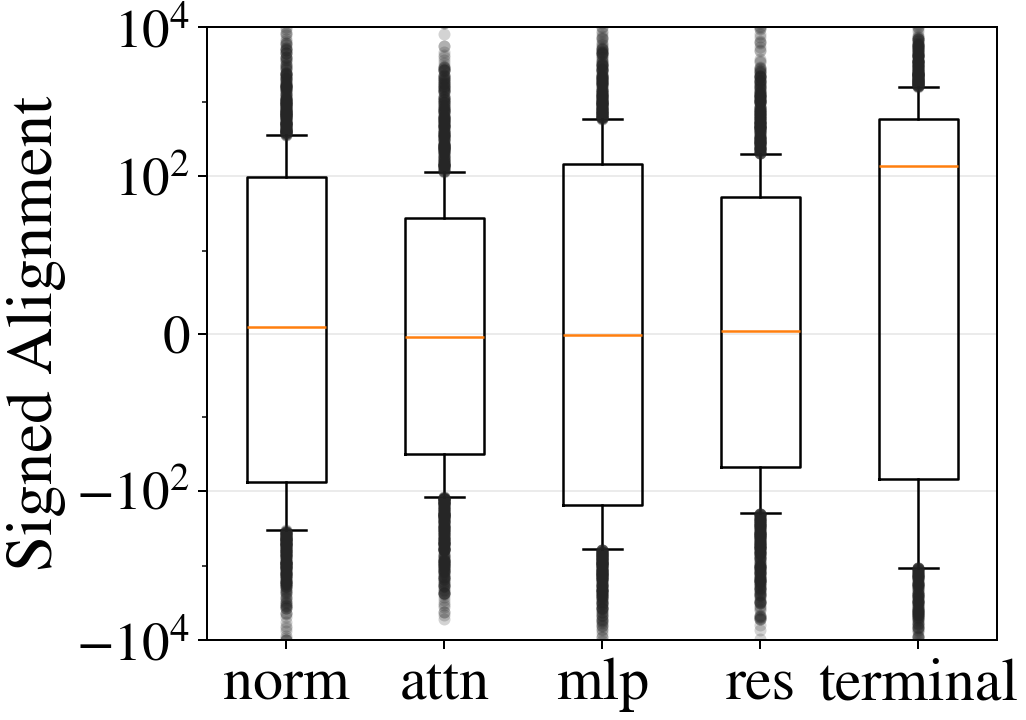}
        \caption{Signed operator-level alignment (Qwen3-14B).}
        \label{fig:jailbreak_family_qwen3-14b}
    \end{subfigure}

    \begin{subfigure}[t]{0.32\linewidth}
        \centering
        \includegraphics[width=\linewidth]{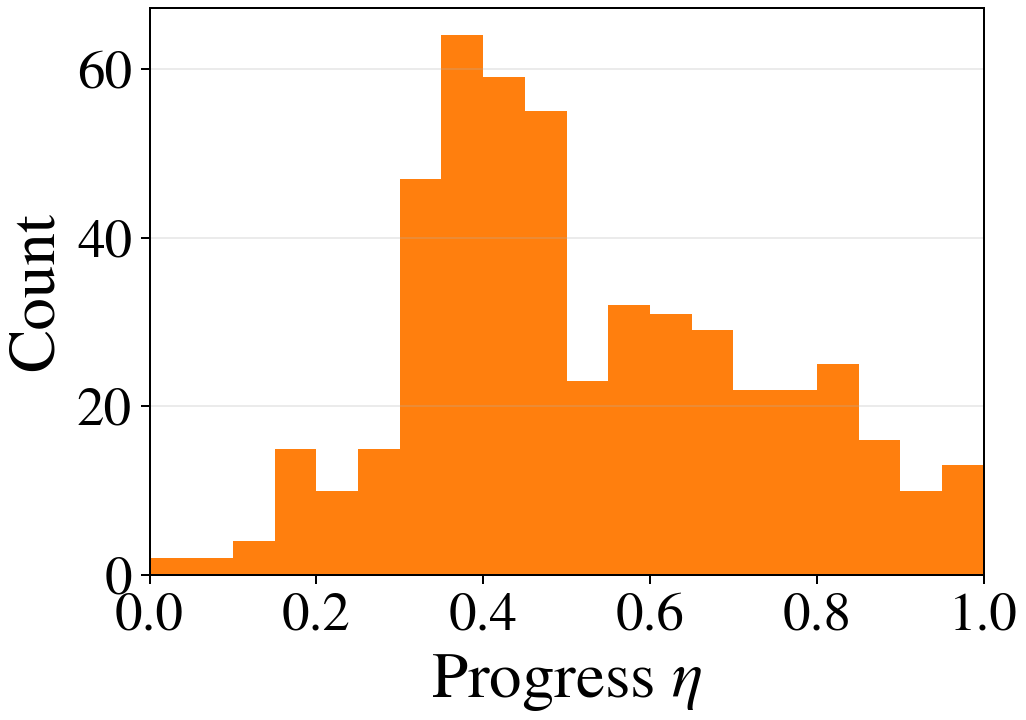}
        \caption{Earliest jailbreak progress (Llama-3.1-8B-Instruct).}
        \label{fig:jailbreak_progress_llama-3.1-8b-instruct}
    \end{subfigure}
    \hfill
    \begin{subfigure}[t]{0.32\linewidth}
        \centering
        \includegraphics[width=\linewidth]{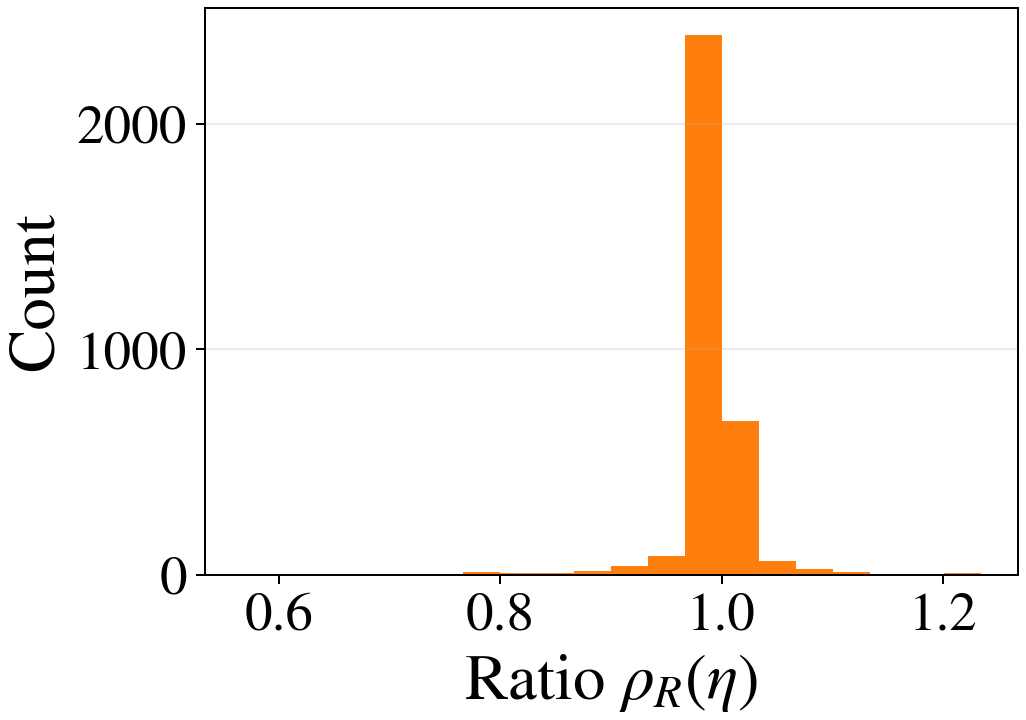}
        \caption{Signed \(R_{\mathrm{ref}}\) ratio (Llama-3.1-8B-Instruct).}
        \label{fig:jailbreak_ratio_llama-3.1-8b-instruct}
    \end{subfigure}
    \hfill
    \begin{subfigure}[t]{0.32\linewidth}
        \centering
        \includegraphics[width=\linewidth]{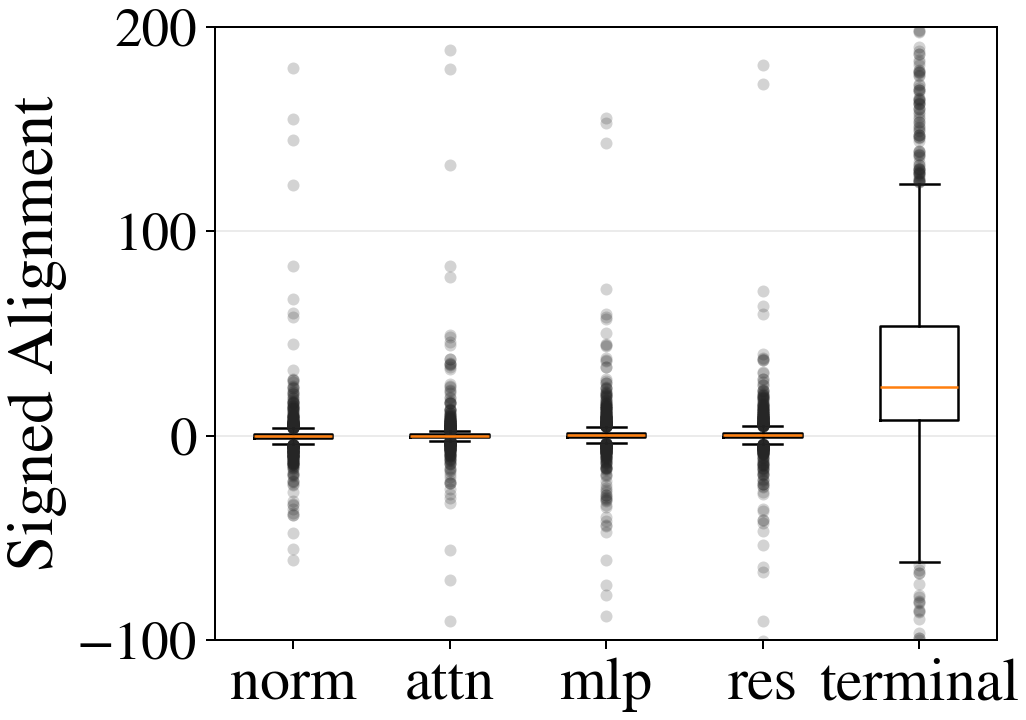}
        \caption{Signed operator-level alignment (Llama-3.1-8B-Instruct).}
        \label{fig:jailbreak_family_llama-3.1-8b-instruct}
    \end{subfigure}

    \begin{subfigure}[t]{0.32\linewidth}
        \centering
        \includegraphics[width=\linewidth]{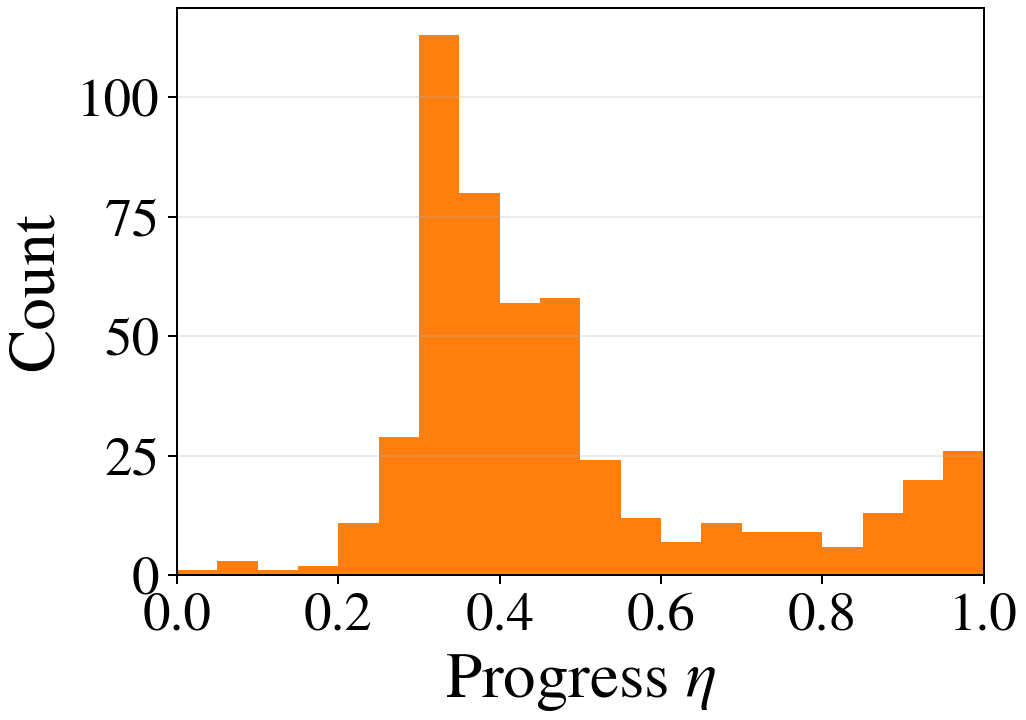}
        \caption{Earliest jailbreak progress (Gemma-3-4B-IT).}
        \label{fig:jailbreak_progress_gemma-3-4b-it}
    \end{subfigure}
    \hfill
    \begin{subfigure}[t]{0.32\linewidth}
        \centering
        \includegraphics[width=\linewidth]{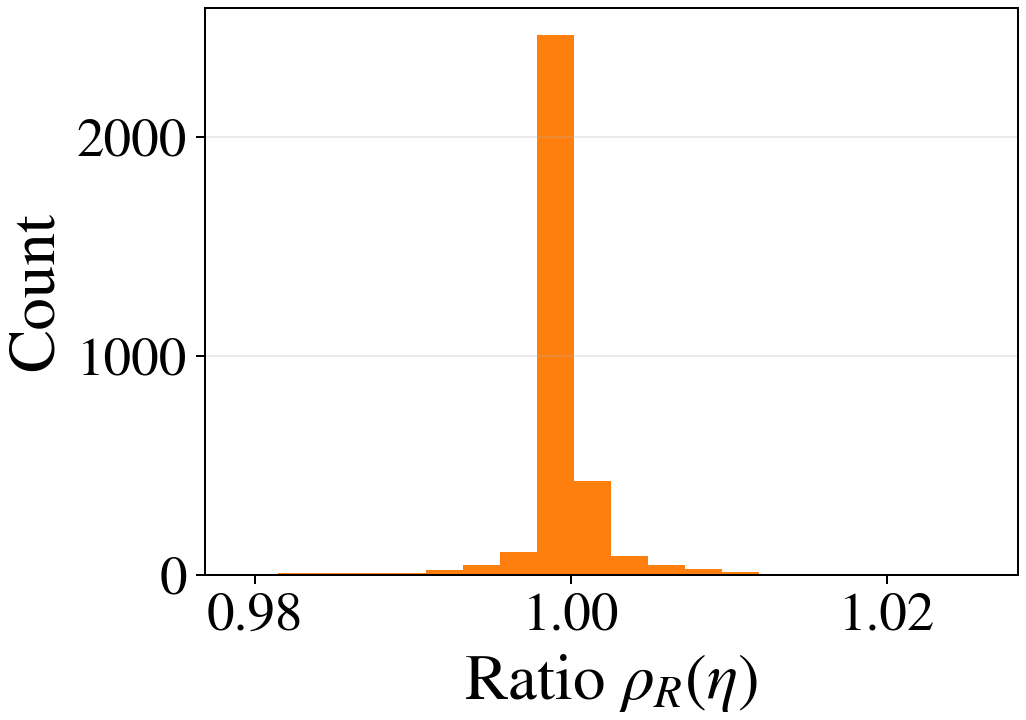}
        \caption{Signed \(R_{\mathrm{ref}}\) ratio (Gemma-3-4B-IT).}
        \label{fig:jailbreak_ratio_gemma-3-4b-it}
    \end{subfigure}
    \hfill
    \begin{subfigure}[t]{0.32\linewidth}
        \centering
        \includegraphics[width=\linewidth]{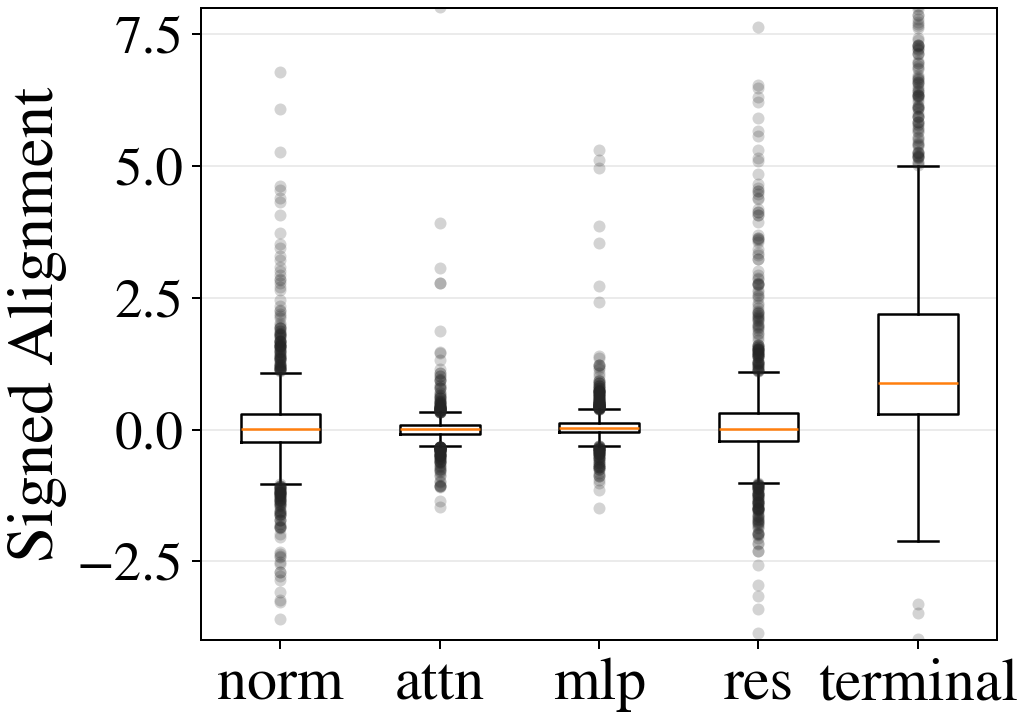}
        \caption{Signed operator-level alignment (Gemma-3-4B-IT).}
        \label{fig:jailbreak_family_gemma-3-4b-it}
    \end{subfigure}

    \begin{subfigure}[t]{0.32\linewidth}
        \centering
        \includegraphics[width=\linewidth]{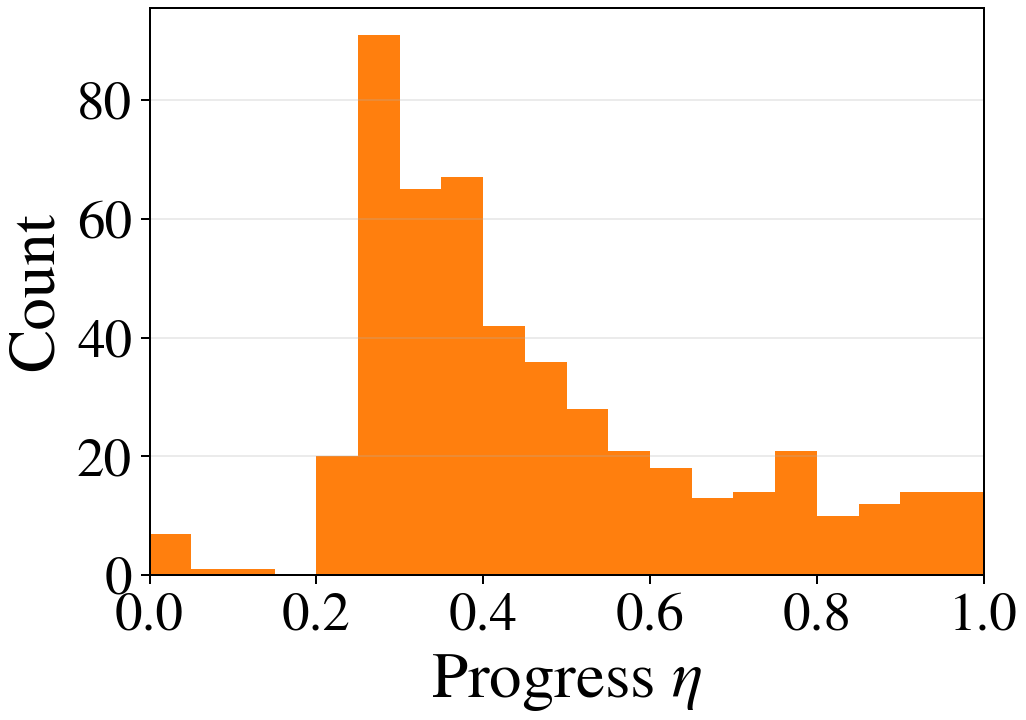}
        \caption{Earliest jailbreak progress (Gemma-3-12B-IT).}
        \label{fig:jailbreak_progress_gemma-3-12b-it}
    \end{subfigure}
    \hfill
    \begin{subfigure}[t]{0.32\linewidth}
        \centering
        \includegraphics[width=\linewidth]{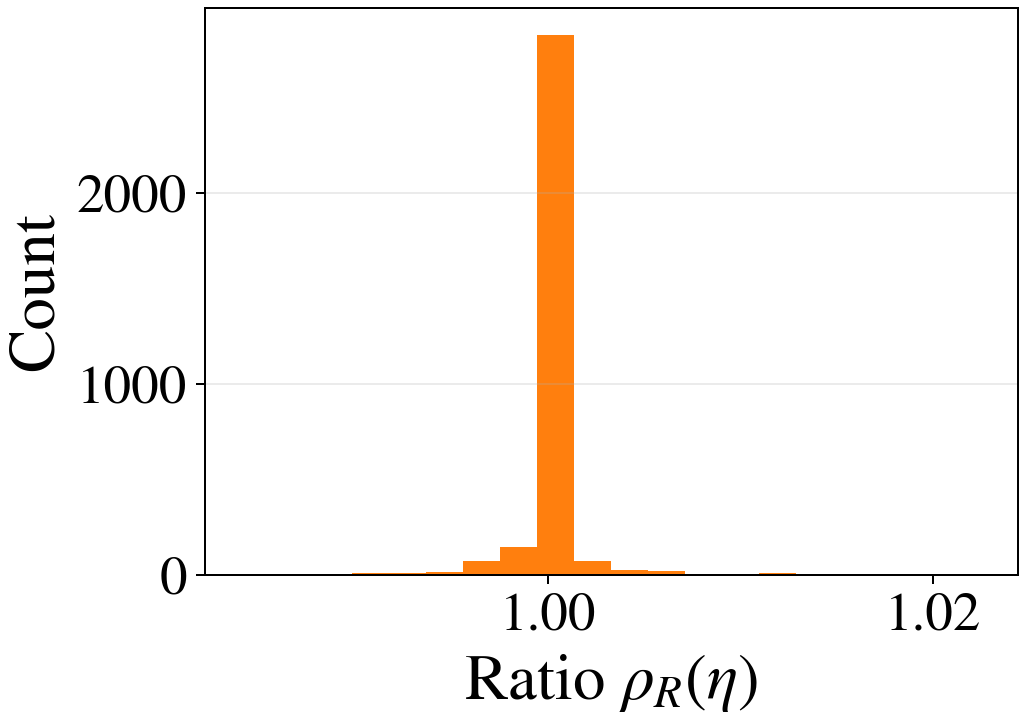}
        \caption{Signed \(R_{\mathrm{ref}}\) ratio (Gemma-3-12B-IT).}
        \label{fig:jailbreak_ratio_gemma-3-12b-it}
    \end{subfigure}
    \hfill
    \begin{subfigure}[t]{0.32\linewidth}
        \centering
        \includegraphics[width=\linewidth]{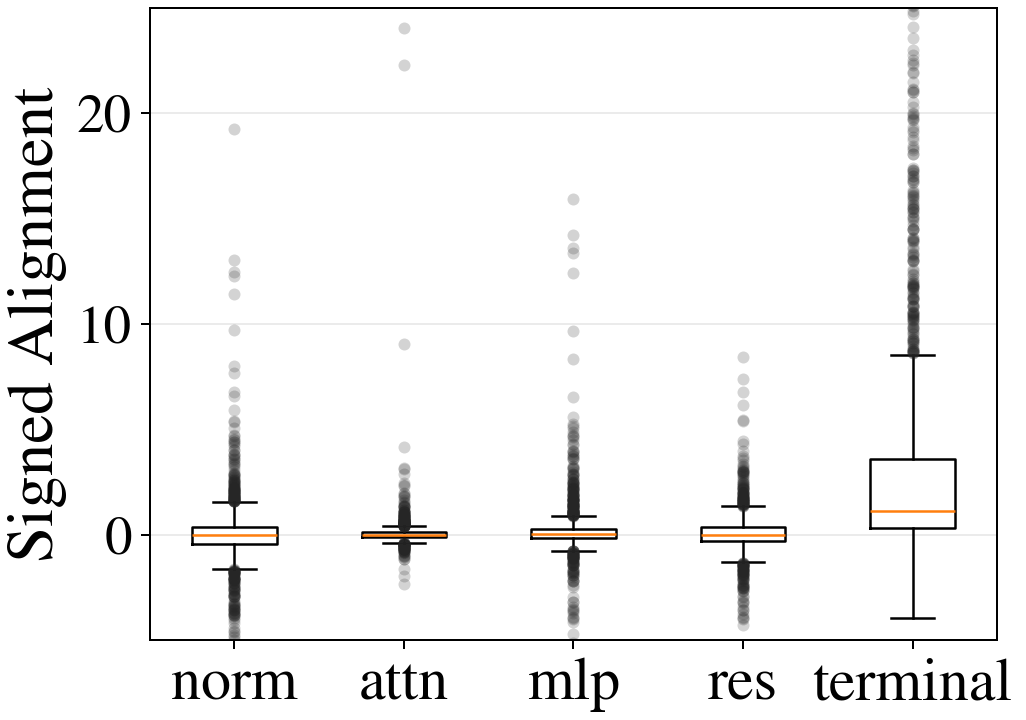}
        \caption{Signed operator-level alignment (Gemma-3-12B-IT).}
        \label{fig:jailbreak_family_gemma-3-12b-it}
    \end{subfigure}

    \caption{Model-specific jailbreak analysis under the continuous input-transformation view. Results are reported separately for each model and aggregated over all samples and attack methods within that model. In each row, the left panel shows the earliest jailbreak progress, the middle panel shows the signed \(R_{\mathrm{ref}}\) ratio, and the right panel shows the signed alignment between the transformation tangent and each operator-level contribution \(S_f^\Sigma\).}
    \label{fig:jailbreak_main_model_specific}
\end{figure}

Figure~\ref{fig:jailbreak_main_model_specific} shows the model-specific results for Observation~2. The earliest-progress results reveal noticeable differences across models. For Qwen3-4B and Qwen3-14B, many jailbreak transitions are judged successful within the first \(5\%\) of the reference input transformation, suggesting that their refusal-to-answer behavior can be changed by very local perturbations around the harmful input under our reference settings. In contrast, for Llama-3.1-8B-Instruct, Gemma-3-4B-IT, and Gemma-3-12B-IT, successful judgments typically appear later, often between \(20\%\) and \(60\%\) of the transformation. This suggests that, under the same reference setup, the safety alignment of Qwen3-4B and Qwen3-14B is more vulnerable to local jailbreak perturbations than that of Llama-3.1-8B-Instruct, Gemma-3-4B-IT, and Gemma-3-12B-IT.

The signed \(R_{\mathrm{ref}}\) ratios remain close to \(1\) across models, indicating that most of the local reference target-behavior change is accounted for by the reference refusal-escape direction. The operator-level alignments further help explain the model-specific difference above: Qwen3-4B and Qwen3-14B show larger signed alignment scales than Llama-3.1-8B-Instruct, Gemma-3-4B-IT, and Gemma-3-12B-IT, indicating stronger local operator-level contributions along the reference input transformation. This is consistent with their earlier jailbreak progress and suggests that their safety behavior is more sensitive to local RED-aligned perturbations. At the same time, the qualitative pattern remains consistent with the aggregate result in Section~\ref{jailbreak_analysis}: terminal-source contributions are most consistently aligned in the positive direction, whereas other operator-level contributions are more dispersed. These model-specific results support the observation that successful jailbreak transitions exhibit refusal-to-answer shifts largely aligned with RED, with terminal-source contributions forming a stable and prominent part of this shift.

\end{document}